\documentclass[11pt]{article}

\pdfoutput=1 

\usepackage{svg}

\usepackage[ruled,vlined,linesnumbered]{algorithm2e}
\usepackage[utf8]{inputenc}
\usepackage[T1]{fontenc}
\usepackage{authblk}
\usepackage{orcidlink}
\usepackage{lmodern}
\usepackage[letterpaper, margin=1in]{geometry}
\usepackage{setspace}
\usepackage{natbib}
\setcitestyle{round}
\usepackage{float}
\usepackage{graphicx}
\usepackage{amsmath,amssymb,amsthm}
\usepackage{hyperref}
\usepackage{subcaption} 
\usepackage{cleveref}
\usepackage{booktabs}
\usepackage{enumitem}
\usepackage{adjustbox}
\usepackage{thm-restate}
\usepackage{caption}
\usepackage{xcolor}
\usepackage{booktabs}
\usepackage{adjustbox}
\usepackage{array}
\usepackage{tabularx} 
\usepackage{tabulary}
\usepackage{colortbl}

\definecolor{famBiSB}{HTML}{F9D4D4}
\definecolor{famSnarl}{HTML}{CFE0EF}

\newtheorem{theorem}{Theorem}
\newtheorem{proposition}{Proposition}%
\newtheorem*{remark*}{Remark}%
\newtheorem{definition}{Definition}
\newtheorem{corollary}{Corollary}
\newtheorem{lemma}{Lemma}

\newtheorem{observation}[theorem]{Observation}

\newtheorem*{note*}{Note}

\newcommand{\noextremity}[2]{\mathsf{State_{#1,#2}[NoInnerExtr]}}
\newcommand{\acyclic}[2]{\mathsf{State_{#1,#2}[Acyclic]}}
\newcommand{\reachesst}[2]{\mathsf{State_{#1,#2}[Reaches_{st}]}}
\newcommand{\reachests}[2]{\mathsf{State_{#1,#2}[Reaches_{ts}]}}
\newcommand{\reachesuvab}[2]{\mathsf{State_{#1,#2}[Reaches_{uv\alpha\beta}]}}

\newcommand{\state}[2]{\mathsf{State_{#1,#2}[\cdot]}}
\newcommand{\true}{\mathsf{True}}
\newcommand{\false}{\mathsf{False}}
\newcommand{\Null}{\mathsf{Null}}

\DeclareMathOperator{\skel}{skeleton}
\DeclareMathOperator{\dirskel}{skeleton^*}
\DeclareMathOperator{\expansion}{expansion}

\newcommand{\iink}{i \in \{1,\dots,k\}}
\newcommand{\iinkz}{i \in \{0,\dots,k\}}
\newcommand{\jink}{j \in \{1,\dots,k\}}
\newcommand{\jinkz}{j \in \{0,\dots,k\}}

\newcommand{\signs}{\{+,-\}}

\newcommand{\figr}[1]{\textcolor{cyan}{(figure)}}

\newcommand{\equalcontrib}{\thanks{These authors contributed equally.}}
\newcommand{\cosupervised}{\thanks{These authors jointly supervised this work.}}

\title{Identifying bubble-like subgraphs in linear-time \\via a unified SPQR-tree framework} 

\author[1]{Francisco Sena\equalcontrib}
\author[1]{Aleksandr Politov$^{*}$}
\author[1,2]{Corentin Moumard}
\author[1]{Massimo Cairo\orcidlink{0000-0001-7247-756X}}
\author[3]{Romeo Rizzi\orcidlink{0000-0002-2387-0952}}
\author[4]{Manuel Cáceres\orcidlink{0000-0003-0235-6951}\cosupervised}
\author[1]{Sebastian Schmidt\orcidlink{0000-0003-4878-2809}$^\dagger$}
\author[1]{Juha Harviainen\orcidlink{0000-0002-4581-840X}$^\dagger$}
\author[1]{Alexandru I. Tomescu\orcidlink{0000-0002-5747-8350}$^\dagger$}

\affil[1]{\small Department of Computer Science, University of Helsinki, Helsinki, Finland\\
\texttt{\{francisco.sena,aleksandr.politov\\sebastian.schmidt,juha.harviainen,alexandru.tomescu\}@helsinki.fi}}
\affil[2]{\small ENS Lyon, Lyon, France\\
\texttt{corentin.moumard@ens-lyon.fr}}
\affil[3]{\small Department of Computer Science, University of Verona, Verona, Italy\\
\texttt{romeo.rizzi@univr.it}}

\affil[4]{\small Department of Computer Science, Aalto University, Espoo, Finland\\
\texttt{manuel.caceres@aalto.fi}}
\date{}

\begin{document}
\pagenumbering{roman}
\maketitle
\vspace{-1cm}

\begin{abstract}
A fundamental algorithmic problem in computational biology is to find all subgraphs of a given type---\emph{superbubbles}, \emph{snarls}, and \emph{ultrabubbles}---in a directed or \emph{bi}directed input graph. Such \emph{bubble-like} subgraphs correspond to regions of genetic variation, whose identification is useful in analyzing \emph{collections} of genomes (a \emph{pangenome graph}). At present, all subgraphs of the latter two types can only be found in quadratic time, which constitutes a major bottleneck for applications involving massive inputs. Although all superbubbles can be identified in linear time, the existing algorithm is highly specialized and the result of a long sequence of tailored developments.

In this work, we present the first linear-time algorithms for identifying all snarls and all ultrabubbles, resolving problems that have remained open since 2018. Additionally, the algorithm for snarls relies on a new linear-size representation of all snarls with respect to the number of vertices in the graph. For the first time in this context, we make use of the well-known SPQR-tree decomposition, which encodes all 2-separators of a biconnected graph. By performing several dynamic-programming–style traversals of this tree, we maintain key properties (such as acyclicity) that allow us to decide whether a given 2-separator defines a subgraph to be reported.

A crucial ingredient for achieving linear-time complexity is the observation that the acyclicity of linearly many subgraphs can be tested simultaneously via a reduction to the classical problem of computing all arcs in a directed graph whose removal renders it acyclic (so-called \emph{feedback arcs}). As such, we prove a fundamental result for bidirected graphs, that may be of independent interest: all feedback arcs can be computed in linear time for tipless bidirected graphs, while in general graphs the problem is at least as hard as matrix multiplication, assuming the $k$-Clique Conjecture.

Altogether, our results form a \emph{unified framework} that also yields a completely different linear-time algorithm for finding all superbubbles. Although some of the results are technically involved, the underlying ideas are conceptually simple, and may extend to other bubble-like subgraphs. 

More broadly, our work contributes to the \emph{theoretical foundations of computational biology} and advances a growing line of research that uses SPQR-tree decompositions as a general tool for designing efficient algorithms, beyond their traditional role in graph drawing.

\end{abstract}

\newpage
\tableofcontents

\newpage
\pagenumbering{arabic}
\setcounter{page}{1} 
\section{Introduction}
\label{sec:introduction}


\subsection{Background and motivation} Recent bioinformatics applications in \emph{pangenomics} are concerned with building a single massive graph representing \emph{many} genomes in a population. For example, the Human Pangenome Reference Consortium (HPRC)~\citep{hprc} released in May 2025 a pangenome graph created from 232 individual genomes,\footnote{\url{https://humanpangenome.org/hprc-data-release-2/}} which contains over 206 million edges~\citep{billi}. The size of such graphs is expected to significantly increase in the future. For example, HPRC is planning to release in Summer 2026 a larger graph created from 350 individual genomes,\footnote{\url{https://humanpangenome.org/release-timeline/}} and there are other initiatives worldwide aiming to construct such graphs from many more human genomes, such as the European ``1+ Million Genomes'' initiative.\footnote{\url{https://digital-strategy.ec.europa.eu/en/policies/1-million-genomes}}

An advantage of pangenome graphs is that genetic variation translates into local substructures, with clean graph-theoretic characterizations, that can be found and analyzed for biological meaning. One of the most used type of such subgraph of a directed graph is a \emph{superbubble}~\citep{onodera2013detecting,dabbaghie2022bubblegun}, see \Cref{fig:bubbles}(a). A superbubble is an acyclic subgraph $B$ identified by two endpoint vertices, say $s$ and $t$, such that $s$ is the only source of $B$ and $s$ does not have out-neighbors outside of $B$; and symmetrically, $t$ is the only sink of $B$ and $t$ does not have in-neighbors outside of $B$. Additionally, there are no other edges from $B$ to the rest of the graph, and a minimality property is also imposed to ensure overall a linear number of superbubbles. See also~\citep{kolmogorov2020metaflye,rautiainen2023telomere,minkin2020scalable,shafin2020nanopore,garg2018graph} for other applications of superbubbles.

The development of pangenome graphs led to the introduction of generalizations of superbubbles, to explicitly handle the \emph{bi}directed nature of the graphs that arises from the reverse-complementary symmetry of DNA. In a \emph{bidirected graph}\footnote{Bidirected graphs, though not as widely known, are well represented in the literature; see e.g.~\citep{gabow1983efficient} (STOC 1983) and \cite[Chapter 36]{schrijver2003combinatorial}.}, every edge also carries a sign $+$ or $-$ at each endpoint. For example, between two vertices $u$ and $v$ there can be four bidirected edges: $\{u-,v-\}$, $\{u-,v+\}$, $\{u+,v-\}$ and $\{u+,v+\}$ (as unordered pairs). A vertex $v$ together with a sign $+$ or $-$ is called a \emph{vertex-side}~\citep{rahman2022assembler}. Directed graphs are a special case of bidirected graphs: every directed edge $(u,v)$ is encoded as $\{u+,v-\}$. Bidirected graphs are heavily used in bioinformatics, and we refer the reader to  e.g.~\citep{bessouf2019transitive,medvedev2007computability,kita2017bidirected,rahman2022assembler} for further motivation on bidirected graphs.

\begin{figure}[t!]
    \centering
    \includegraphics[width=0.95\linewidth]{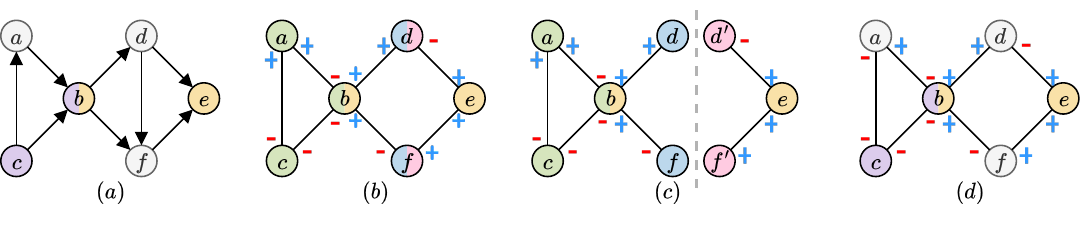}
    \caption{Examples of superbubbles (a), snarls (b), and ultrabubbles (d) in directed graphs and bidirected graphs, with the extremities of the bubbles highlighted in the same color. The superbubbles are $(c, b)$ and $(b, e)$, the snarls are $\{a+, b-\}$, $\{a+, c-\}$, $\{b-, c-\}$, $\{b+, e+\}$, $\{d+, f-\}$, and $\{d-, f+\}$, and the ultrabubbles are $\{b-, c-\}$ and $\{b+, e+\}$. Subfigure (c) illustrates splitting vertex sides $d+$ and $f-$ from (b), after which $d$ and $f$ are in the same connected component that is distinct from the component of $d'$ and $f'$. Note that all pairs of $a+$, $b-$, and $c-$ form snarls in (b), whereas with superbubbles and ultrabubbles each vertex-side can be an extremity of only a single bubble.}\label{fig:bubbles}
\end{figure}

In bidirected graphs, the most well-known \emph{bubble-like} notions are \emph{snarls} (see \Cref{fig:bubbles}(b)) and \emph{ultrabubbles} (see \Cref{fig:bubbles}(d)), introduced by~\cite{paten2018ultrabubbles}. At a high level, a snarl is a minimal connected subgraph identified by two vertex-sides that separate its interior from the rest of the graph; moreover, at each of the two endpoints, all vertex-sides of one sign lie inside the snarl and all vertex-sides of the opposite sign lie outside. Ultrabubbles are snarls whose interior contains no \emph{tip} and is \emph{acyclic}, i.e.~it contains no \emph{closed bidirected walk}. A \emph{tip} is a generalization of a source/sink in directed graphs: it is a vertex $v$ for which all incident edges carry the same sign at $v$. A \emph{bidirected walk} alternates the sign at every internal vertex; it is \emph{closed} if it starts and ends at the same vertex, at which it does not need to alternate sign. We also call such a walk a \emph{cycloid}. See also~\citep{garrison2018variation,hickey2024pangenome,wang2025population,chang2020distance,siren2021pangenomics,siren2024personalized} for important applications of snarls and ultrabubbles.


\subsection{Existing bubble-finding algorithms}

Given a directed graph $(V,E)$, the first superbubble finding\footnote{The interior of a bubble is uniquely identified by its endpoints, and thus bubble-finding algorithms compute the set of pairs of endpoint vertices of all bubbles of a certain type. In fact, since bubble subgraphs can contain smaller bubble subgraphs, one cannot obtain linear-time algorithms if the full subgraphs are reported in output.} algorithm ran in time $O(|V|(|E|+|V|))$~\citep{onodera2013detecting}, and was later improved to $O(\log|V|(|E|+|V|))$-time~\citep{loglinear}. Linear time $O(|V| + |E|)$ was first achieved for acyclic graphs by~\cite{brankovic2016linear}, then for general graphs by~\cite{gartner-revisited}, and further simplified by~\cite{gartner2019direct}.

Despite the massive size of pangenome graphs, existing worst-case bounds for these bidirected structures are far from linear: snarls can be computed in $O(|V||E|)$ time in the worst case, and ultrabubbles in $O((|V|+|E|)^2)$ time~\citep{paten2018ultrabubbles}. Another intrinsic difficulty is output size, as the total number of snarls can be quadratic in the number of vertices. \cite{paten2018ultrabubbles} therefore used the \emph{cactus graph}~\citep{paten2011cactus} to prune snarls and compute in linear time a \emph{snarl decomposition} of linear size. However, because of their possible biological significance, we are interested in identifying \emph{all} bubble-like structures in the graph. \cite{zisis2026ultrabubbleenumerationlowestcommon} recently proposed an algorithm which, given a set of $k$ snarls, can verify in time $O(k|V|+(|V|+|E|))$ which ones are ultrabubbles. Finally, \cite{Harviainen2026.03.28.714704} showed that ultrabubbles can be computed in linear time on the special class of bidirected graphs that either contain a tip or their underlying undirected graph contains a cutvertex. This is based on transforming any bidirected graph in this class into a directed graph which is at most of double size, such that the ultrabubbles of the bidirected graph correspond exactly to (a slight generalization of) superbubbles in the directed graph. However, note that this reduction does not work in general.

Despite the similarities between the several existing bubble-like structures (see also e.g.~\emph{bibubbles}~\citep{li2024exploring}, \emph{panbubbles}~\citep{billi}, and \emph{flubbles}~\citep{mwaniki2024popping}), there is no unified methodology for efficiently computing them. Moreover, since superbubbles can indeed be computed in linear time, it would seem reasonable to assume that such an approach can be adapted also for snarls. However, this achievement is heavily tailored to superbubbles and crucially relies on the \emph{directed} nature of the input graph. Finally, this result has been obtained only after a series of papers~\citep{onodera2013detecting,loglinear,brankovic2016linear,gartner-revisited,gartner2019direct}, begging the question of how large undertaking obtaining a linear-time snarl or ultrabubble identification algorithm is.



\subsection{Contributions} In this paper we show that all snarls and all ultrabubbles can be identified in linear time in the size of the graph, solving problems that have been open since 2018. For snarls, we further prove that the set of all snarls admits a representation of size linear in the number of vertices of the input graph, which can itself be constructed in linear time. Thus, we can \emph{identify} all snarls in time linear in the input size. 

We obtain both algorithms via a \emph{unified} framework for finding these structures. This framework also yields a new linear-time algorithm for computing all superbubbles in directed graphs. 

The key insight underlying our approach is that the two endpoints of a bubble subgraph form a \emph{2-separator} (i.e.~2-vertex cut) in the underlying undirected graph, except for special cases which we can handle separately. We then leverage classical decomposition machinery for 2-separators, namely the \emph{SPQR tree} of a (bi)connected graph~\citep{battista1990on-line,gutwenger2001linear}, which compactly encodes all its 2-separators. 

While SPQR trees have been used in bioinformatics in other contexts (e.g.~\citep{metagenomescope,Jafarzadeh2025.07.05.662656}), they have not been used to algorithmically capture bubble-like structures. More broadly, although SPQR trees are most commonly associated with graph embedding and drawing applications~\citep{mutzel2003spqr}, our work contributes to a growing line of research that employs them to design efficient algorithms. For example, recognizing if a graph belongs to a certain class~\citep{DEMACEDOFILHO2018101}, or constructing efficient indexing schemes for e.g.~shortest-path queries~\citep{maniu2017indexing}. 

Although some of the results are technically involved, the underlying ideas are conceptually simple. This suggests that the same framework may extend to other bubble-like notions, such as bibubbles~\citep{li2024exploring} and panbubbles~\citep{billi}, for which linear-time algorithms are currently unknown.

\subsection{Organization of the paper} In \Cref{sec:preliminaries} we introduce key technical notions, and in \Cref{sec:overview} we give an extended overview of our main ideas. To better illustrate our framework, we first apply it to standard directed graphs and present the new linear-time superbubble algorithm, in \Cref{sec:superbubbles}. Our goal in the presentation is to illustrate the key ideas, which we can then adapt to the more involved case of bidirected graphs. We show the linear-time snarl algorithm in \Cref{sec:snarls}, and the results on ultrabubbles in \Cref{sec:ultrabubbles}.

\section{Preliminaries}
\label{sec:preliminaries}

Many of our results are based on analyzing the underlying undirected graph of the (bi)directed graph we wish to find superbubbles, ultrabubbles, or snarls. We begin by giving preliminaries on undirected graphs and basic terminology on connectivity. Then we introduce terminology for bidirected and directed graphs.

\subsection{Undirected graphs and connectivity}
Let $H=(V,E)$ be an undirected graph with vertex set $V(H)$ and edge set $E(H)$. Let $u,v \in V(H)$ be vertices. If $H$ has an edge with endpoints $u$ and $v$ then we denote that edge as $\{u,v\}$ (parallel edges are allowed).
A \emph{$u$-$v$ undirected walk} in $H$ is a standard walk between $u$ and $v$; a \emph{$u$-$v$ undirected path} is a $u$-$v$ undirected walk without repeated vertices. The \emph{internal vertices} of a $u$-$v$ undirected path are the vertices contained in the path except $u$ and $v$. (When clear from the context, we simply say ``path'' instead of undirected ``path''.)

A subgraph $H$ of $G$ is \emph{maximal} w.r.t.~a given property if no proper supergraph of $H$ contained in $G$ has that property.
The subgraph \emph{induced} by a subset $C \subseteq V(G),E(G)$ of vertices or edges of $G$ is denoted by $G[C]$. The \emph{vertex-induced} subgraph is the graph with vertex set $C$ and the subset of edges in $E(G)$ whose endpoints are in $C$. The \emph{edge-induced} subgraph has edge set $C$ and a vertex set consisting of all endpoints of edges in $C$. If $G'=(V',E')$ is an undirected graph, then $G'\cup G := (V\cup V',E\cup E')$.

Graph $H$ is \emph{disconnected} if it has two vertices without a path between them.
Graph $H$ is \emph{$k$-connected} if it has more than $k$ vertices and no subset of fewer than $k$ vertices disconnects the graph. By Menger's theorem~\citep{menger}, if a graph $H$ is $k$-connected then $H$ has $k$ internally vertex-disjoint paths between any two of its vertices.
A \emph{component} of $H$ is a maximally (1-)connected subgraph.
Vertex $v$ is a \emph{cutvertex} of $H$ if $H-v$ has more components than $H$; if $H$ has no cutvertex then it is \emph{biconnected}. The vertex pair $\{u,v\}$ is a \emph{separation pair} of $H$ if $(H - u) - v$ has more components than $H$; if $H$ has no separation pairs then it is \emph{triconnected}.
Notice that we allow biconnected (resp. triconnected) graphs to have fewer than three (resp. four) vertices.
We call an edge a \emph{bridge} if its removal increases the number of components of the graph; a set of at least two parallel edges whose removal increases the number of components of the graph is called a \emph{multi-bridge}.
If every $u$-$v$ path contains vertex $w \neq u,v$ then $w$ is a $u$-$v$ cutvertex.
A \emph{separation} of $H$ is a pair of vertex sets $(A,B)$ such that $V = A \cup B$, $A\setminus B$ and $B\setminus A$ are nonempty, and there is no edge between $A\setminus B$ and $B\setminus A$. A maximal set of vertices $X \subseteq V(H)$ with $|X| \geq k$ such that no two vertices of $X$ can be separated by removing fewer than $k$ other vertices is called a \emph{$k$-block}.

\subsection{Bidirected and directed graphs}
A \emph{bidirected graph} $G = (V,E)$ has a set of vertices $V=V(G)$ and a set of \emph{bidirected edges} $E=E(G)$.
A \emph{sign} is an element in $\alpha \in \{+, -\}$, and the \emph{opposite sign} $\hat{\alpha}$ of $\alpha$ is defined as $\hat{-} = +$ and $\hat{+} = -$.
A pair $(v,\alpha)$ where $v \in V(G)$ and $\alpha \in \signs$ is a \emph{vertex-side},\footnote{We adopt this nomenclature from~\citep{rahman2022uncovering} and note that~\citep{kita2017bidirected} opts for ``signed vertex''. We remark that our definition of bidirected graph, while appropriate for this work, differs from commonly used definitions (e.g., ~\citep{kita2017bidirected,ghorbani2025generalisation,bessouf2019transitive,ando1996decomposition,schrijver2003combinatorial}) although all are essentially equivalent.} which we concisely write as $v\alpha$, e.g.~$v+$ or $v-$.
A \emph{bidirected edge} $e \in E(G)$ is an unordered pair of vertex-sides $\{u\alpha, v\beta\}$ (for simplicity we may refer to bidirected edges as just edges); we say that $e$ is incident to $u$ (resp. $v$) with sign $\alpha$ (resp. $\beta$), that $G$ has a vertex-side of sign $\alpha$ in $u$, and that $u$ and $v$ are the \emph{endpoints} of $e$.
The set of \emph{vertex-sides} of $G$ is the set $\bigcup_{e \in E(G)} e$.
A bidirected graph $H$ is a \emph{subgraph} of $G$ if $V(H) \subseteq V(G)$ and $E(H) \subseteq E(G)$, and write $H \subseteq G$.
The set of \emph{out-neighbors} of $v$ is denoted as $N^+_G(v)$ and consists of the vertices $x$ for which there is an edge $\{v+,x\beta\}$ in $G$ (the set of \emph{in-neighbors} is defined analogously). We say that a vertex $v$ is a \emph{tip} in $G$ if no two edges of $G$ are incident to $v$ with different signs.
The undirected graph of $G$ is denoted by $U(G)$ and is obtained from $G$ by ignoring the signs in its vertex-sides, and keeping parallel edges that possibly appear (edges are thus labeled with unique identifiers that are retained during the conversion).

A \emph{bidirected walk} $W$ is a sequence $\{v_1\alpha_1, v_2\alpha_2\}$, $\{v_2\hat{\alpha}_2, v_3\alpha_3\}, \dots, \{v_{k-1}\hat{\alpha}_{k-1}, v_k\alpha_k\} \in E(G)$.
We also say that $W$ is a \emph{$v_1$-$v_k$ bidirected walk} (also a $v_1\alpha_1$-$v_k\alpha_k$, a $v_1$-$v_k\alpha_k$, or a $v_1\alpha_1$-$v_k$ bidirected walk).
When clear from the context we may say simply ``walk'' instead of ``bidirected walk''. Vertices $x$ and $y$ are the \emph{first} and \emph{last} vertices of the walk, respectively, and all remaining vertices are its \emph{internal} vertices.
Observe that to any walk $W$ we can associate its reversed walk (i.e., walks in bidirected graphs have two possible orientations).
A \emph{bidirected path} is a walk without repeated vertices.
A \emph{cycloid} is a path with the exception that the first and last vertex are the same and where at most one of its vertices (called the \emph{exceptional vertex}) has the same sign over the two edges incident to it in the sequence.
A graph with no cycloid is \emph{acyclic}.

When every edge of $G$ has one vertex-side with a $+$ sign and the other with a $-$ sign then we can say that $G$ is a \emph{directed graph} where a bidirected edge $\{u+,v-\}$ is seen as a \emph{directed edge} from $u$ to $v$, which we concisely denote as $uv$.
A \emph{directed path} is a sequence of vertices $v_1\dots v_k$ such that $v_iv_{i+1} \in E(G)$ for $i=1,\dots,k-1$, in which case we also say that $v_1$ \emph{reaches} $v_k$ in $G$ or that this sequence is a \emph{$v_1$-$v_k$ directed path}. A vertex $v$ is a \emph{source} of $G$ if $|N^-_{G}(v)|=0$ and a \emph{sink} if $|N^+_{G}(v)|=0$.
A vertex $v$ is an \emph{extremity} of $G$ if $v$ is a source or sink of $G$, or a cutvertex of $U(G)$.

\subsection{Block-cut trees} 

Let $H=(V,E)$ be an undirected connected graph with at least two vertices. It follows from the definition of $k$-block that a subgraph induced by a 2-block is a maximal connected subgraph without cutvertices (see~\citep{diestel}). For simplicity, we will refer to the subgraphs induced by 2-blocks simply as blocks. The \emph{block-cut tree} of $H$ is a tree with \emph{node set} $N$ and edge set $A$. A node is a \emph{block node}, which is either a maximal 2-connected subgraph or a multi-bridge of $H$, or a \emph{cutnode}, which is a cutvertex of $H$.

The edges in $A$ represent how the blocks of $H$ interact via the cutvertices of $H$ as follows.
Let $v$ be a cutvertex of $H$ and let $\mu$ be the cutnode of $N$ corresponding to $v$. Then $H-v$ consists of components $C_1,\dots,C_\ell$ ($\ell \geq 2$) and $\mu$ has $\ell$ neighbours in the tree, each corresponding to the block contained in $H[V(C_i) + v]$ that meets $v$.
The blocks partition the edge set of $H$~\citep{diestel} and there is at most one block containing any two vertices.

\subsection{SPQR trees}

SPQR trees represent the decomposition of a biconnected graph according to its separation pairs in a tree-like way, thus exposing the 3-blocks of the graph (analogously to cutvertices and blocks in block-cut trees). They were first formally defined by Tamassia and Di Battista~\citep{battista1990on-line}, but were informally known before~\citep{lane1937structural,hopcroft1973dividing,bienstock1988complexity}.
They can be constructed in linear time~\citep{hopcroft1973dividing,gutwenger2001linear}, and if unrooted they are unique~\citep{battista1990on-line}. SPQR trees are a valuable tool in the design of algorithms for different problems~\citep{rotenberg,di1996line,onlinegraphalgorithms}. The definition of SPQR trees given in this paper is essentially the same given by Di Battista and Tamassia in~\citep{battista1990on-line,battista1996on-line,di1996line}; note that the exact same recursive definition is used in~\citep{gutwenger2001linear,gutwenger2005inserting} (for a non-recursive definition see~\citep{rotenberg}). We remark that our definition contains some minor technical adjustments to the definition of Di Battista and Tamassia. 

\begin{figure}[t]
    \centering
    \includegraphics[width=\linewidth]{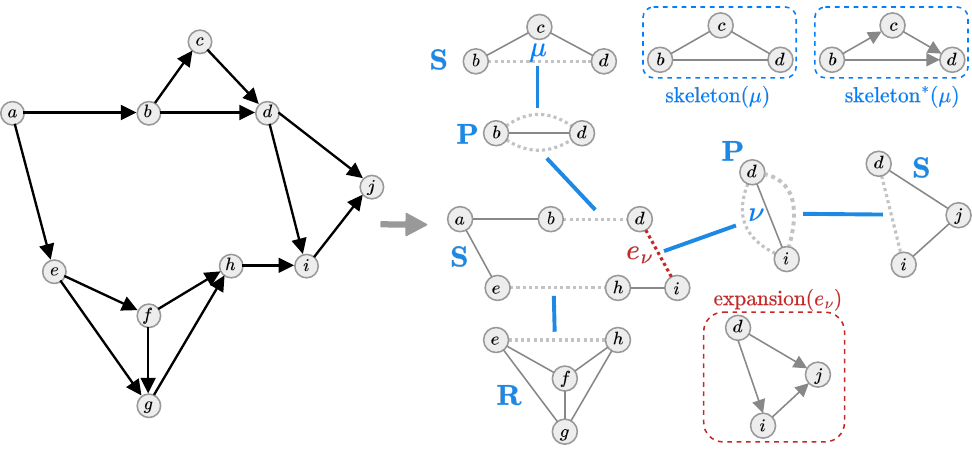}
    \caption{A directed graph (left) and the SPQR tree of its (biconnected) underlying undirected graph (right). Virtual edges are dashed, tree edges (in blue) connect pairs of virtual edges in adjacent tree nodes, and node types (S, P, R) are shown in blue. We highlight in red a virtual edge $e_\nu$ and show its $\expansion(e_\nu)$ in the red box. The $\skel(\mu)$ of the tree node $\mu$ (the top-most node) is shown in the blue box. The directed skeleton $\dirskel(\mu)$ is a directed graph obtained from $\skel(\mu)$ where real edges are directed as in the input directed graph, and for each $e = \{s,t\}$ in $\skel(\mu)$ we add an edge from $s$ to $t$ (and from $t$ to $s$) if $s$ reaches $t$ (or $t$ reaches $s$, resp.).
    \label{fig:spqr-tree}
    }
\end{figure}

To define SPQR trees we need some basic definitions. Let $H$ be an undirected biconnected graph with at least two edges. A \emph{split pair} of $H$ is a separation pair or an edge of $H$. A \emph{split component} of a split pair $\{u,v\}$ is an edge $\{u,v\}$ or a maximal subgraph $C$ of $H$ containing the vertices $u$ and $v$ such that $\{u, v\}$ is not a split pair of $C$. Let $\{s,t\}$ be a split pair of $H$. A \emph{maximal split pair} $\{u, v\}$ of $H$ with respect to $\{s, t\}$ is such that for any other split pair $\{u', v'\}$ vertices $u, v, s$, and $t$ are in the same split component of $\{u',v'\}$.

Fix an arbitrary edge $e=\{s,t\} \in E(H)$, called the \emph{reference edge}. The SPQR tree $T$ of $H$ with respect to $e$ is defined as a rooted tree with nodes of four types: S~(series), P~(parallel), Q~(single edge), and R~(rigid). Each node $\mu$ in $T$ has an associated biconnected graph \emph{$\skel(\mu)$}, called the \emph{skeleton} of $\mu$ with $V(\skel(\mu)) \subseteq V(H)$. The tree $T$ is recursively defined as follows:

\begin{description}
    \item[Trivial case:] If $H$ consists of exactly two parallel edges between $s$ and $t$, then $T$ consists of a single Q-node whose skeleton is $H$ itself.
    \item[Parallel case:] If the split pair $e = \{s,t\}$ has exactly $k+1\geq 3$ split components $H_0,\dots,H_k$ where $H_0$ denotes the split component containing $e$, then the root of $T$ is a P-node $\mu$ whose skeleton consists of $k+1$ parallel edges $e_0,\dots,e_k$ between $s$ and $t$ with $e_0 = e$.
    \item[Series case:] Otherwise the split pair $\{s,t\}$ has exactly two split components, where one is $e$ trivially and the other let us denote by $H'$. If $H'$ is a sequence of blocks $H_1,\dots,H_k$ separated by cutvertices $c_1,\dots,c_{k-1}$ $(k \geq 2)$ in this order from $s$ to $t$, then the root of $T$ is an S-node $\mu$ whose skeleton is a cycle $e_0, e_1,..., e_k$ where $e_0$ = $e$, $c_0=s$, $c_k=t$, and $e_i = (c_{i-1}, c_i)$ $(i=1,\dots,k)$.
    \item[Rigid case:] If none of the previous cases applies, let $\{s_1,t_1\},\dots,\{s_k,t_k\}$ be the maximal split pairs of $H$ with respect to $\{s,t\}$ $(k\geq1)$, and, for $i=1,\dots,k$ let $H_i$ be the union of all the split components of $\{s_i,t_i\}$ except the one containing $e$. The root of $T$ is an R-node $\mu$ whose skeleton is obtained from $H$ by replacing each subgraph $H_i$ with the edge $e_i = \{s_i, t_i\}$.
\end{description}
Except for the trivial case, $\mu$ has children $\mu_1,\dots,\mu_k$, such that $\mu_i$ is the root of the SPQR tree of $H_i \cup e_i$ with respect to $e_i$ for $i=1,\dots,k$. Notice how the (reference) edge $e_i$ in $\skel(\mu_i)$ ensures that $\skel(\mu_i)$ is biconnected (e.g., in the case of the S-node).
Once the recursion terminates, we add a Q-node with vertex set $\{s,t\}$ representing the first reference edge $e=\{s,t\}$ and make it a child of the root of $T$.
Node $\mu_i$ is associated with edge $e_i$ of the skeleton of its parent $\mu$, called the \emph{virtual edge} of $\mu_i$ in $\skel(\mu)$ $(i=1,\dots,k)$. Conversely, $\mu$ is implicitly associated with the reference edge $e_i$ in $\skel(\mu_i)$.
Notice that reference and virtual edges encode the same information: two subgraphs of $H$ and how they attach to each other. Indeed, a reference edge of some node is just another virtual edge with the additional property of pointing to the parent of that node.
We say that $\mu$ is the \emph{pertinent node} of $e_i \in E(\skel(\mu_i))$ (or that $e_i \in E(\skel(\mu_i))$ pertains to $\mu$), and that $\mu_i$ is the pertinent node of $e_i \in E(\skel(\mu))$ (or that $e_i \in E(\skel(\mu))$ pertains to $\mu_i$).

\textbf{Additional definitions.}
For simplicity, we will omit Q-nodes from the SPQR tree.
This amounts to replacing every virtual edge pertaining to a Q-node by a \emph{real edge} and removing every Q-node from the tree. The edges of a skeleton are then either real or virtual.

Suppose now that $\nu$ is the parent of $\mu$ in $T$. Let $e_\nu \in \skel(\mu)$ be the edge pertaining to $\nu$ and let $e_\mu \in \skel(\nu)$ be the edge pertaining to $\mu$. Let $\{s,t\}$ be the endpoints of $e_\nu$ and $e_\mu$. Deleting the edge $\{\nu,\mu\}$ from $T$ disconnects $T$ into two subtrees, $T_\nu$ containing $\nu$ and $T_\mu$ containing $\mu$. The \emph{expansion graph} of $e_\nu$, denoted as $\expansion(e_\nu)$, is the subgraph induced in $H$ by the real edges contained in the skeletons of the nodes in $T_\nu$. The graph $\expansion(e_\mu)$ is defined analogously with respect to $T_\mu$. The expansion of a real edge is the graph consisting of that edge alone.

\begin{sloppypar}
Each edge of $T$ (importantly, without Q-nodes) encodes a separation of $H$ (i.e., $(V(\expansion(e_\nu)),V(\expansion(e_\mu)))$ is a separation of $H$). Further, we have $V(\expansion(e_\nu)) \cap V(\expansion(e_\mu)) = V(\skel(\nu)) \cap V(\skel(\mu)) = \{s,t\}$ and $\expansion(e_\nu)\cup \expansion(e_{\mu})=H$.
For every node $\mu$ of $T$ whose skeleton has edge set $\{e_1,\dots,e_k\}$, the graph $\bigcup_{i=1}^k \expansion(e_i)$ is exactly $H$.
In SPQR trees no two S-nodes and no two P-nodes are adjacent~\citep{di1996line}.
\end{sloppypar}

The next two statements are well known results about SPQR trees. \Cref{lem:spqr-total-size} below is given in a context where Q-nodes are part of the tree. Clearly, removing Q-nodes maintains the bounds.

\begin{lemma}[SPQR trees and separation/split pairs]\label{lem:spqr-tree-contains-split-pairs}
    Let $H$ be an undirected 2-connected graph and let $T$ be its SPQR trees with Q-nodes omitted. For each S-node $\mu$ of $T$, let $X_\mu$ denote the set of all pairs of nonadjacent vertices in $\skel(\mu)$. Then the union of the virtual edges over the skeletons of the nodes of $T$ together with the union of all the $X_\mu$ is exactly the set of separation pairs of $H$. If Q-nodes are included in the tree, then this union is exactly the set of split pairs of $H$.
\end{lemma}

\begin{lemma}[SPQR trees require linear space~\citep{onlinegraphalgorithms}]
\label{lem:spqr-total-size}
    Let $H=(V,E)$ be an undirected biconnected graph. The SPQR tree $T$ of $H$ has $O(|V(H)|)$ nodes and the total number of edges in the skeletons is $O(|E(H)|)$.
\end{lemma}

The next simple result is merely technical and will be used throughout in the paper.

\begin{lemma}
\label{lem:reaches-in-expansion}
    Let $G$ be a bidirected graph, $H$ be a 2-connected subgraph of $G$, and $T$ be the SPQR tree of $H$. Let $\mu$ be a node of $T$, $e=\{u,v\}$ be a virtual edge of $\skel(\mu)$, and $a \in V(H) \setminus \{u\}$ be a vertex. If $a \in V(\expansion(e))$ then $\expansion(e)$ has an $a$-$v$ path avoiding $u$.
\end{lemma}
\begin{proof}
    If $a=v$ we are done, so $a \neq v$. Suppose for a contradiction that every $a$-$v$ path in $\expansion(e)$ contains $u$. Since $a$ is contained in a split component of $\{u,v\}$, every $a$-$v$ path in $H$ also contains $u$. So $u$ is an $a$-$v$ cutvertex in $H$, contradicting the fact that $H$ is 2-connected.
\end{proof}

\paragraph{A remark on notation.}
To simplify the writing we will refer to connectivity properties of the underlying undirected graph of a bidirected graph by saying that the bidirected graph itself has the property, as to minimize the use of the cumbersome notation $U(G)$.
For instance, if $G$ is a bidirected graph such that $U(G)$ is 2-connected and where $x$ is a $u$-$v$ cutvertex, then we say that $G$ is 2-connected and that $x$ is a $u$-$v$ cutvertex with the meaning that $U(G)$ has those properties. The only possible ambiguity arising from this choice is on the notion of ``walk''. For that, we carefully specify what kind of walk we are referring to by using the terms ``(bi)directed'' and ``undirected'' as they were defined previously; only when it is clear from the context, we allow ourselves to simply say ``walk''.

We also build block-cut trees and SPQR trees directly on bidirected graphs (connectivity is seen from their underlying undirected graphs).
The edges contained in the blocks of block-cut trees, in the skeletons of the nodes of SPQR trees, in split components, etc, additionally encode their relevant properties in the (bi)directed graph they live in.
This applies also to the $\expansion$ operator in SPQR trees.

We will routinely solve subproblems on the skeletons whose edges are assigned directions depending on the reachability relation of $G$ restricted to their respective expansions. Formally, let $\mu$ be a node of $T$ and let $e_1=\{s_1,t_1\},e_2=\{s_2,t_2\},\dots,e_k=\{s_k,t_k\}$ be the edges of $\skel(\mu)$ $(k \geq 2)$.
Define the set of directed edges $B_1 = \{ s_it_i : \text{$s_i$ reaches $t_i$ in $\expansion(e_i) ,\; i=1,\dots,k$} \}$ and $B_2 = \{ t_is_i : \text{$t_i$ reaches $s_i$ in $\expansion(e_i),\; i=1,\dots,k $} \}$. We define the \emph{directed skeleton} of $\mu$ as $\dirskel(\mu) := (V(\skel(\mu)),B_1 \cup B_2)$.


\section{Overview of our results and techniques}
\label{sec:overview}

In this section, we will give a high-level description of the results and the techniques behind them. We start by informally defining the bubble-like structures (or just bubbles) we are interested in, as they will be formally defined in~\Cref{sec:superbubbles,sec:snarls,sec:ultrabubbles}, respectively. Next, we give more details on how we handle the bubbles of each type.
In this paper we assume that (bi)directed graphs have no parallel edges since they have no effect on superbubbles, ultrabubbles, and snarls (two edges $\{x\alpha,y\beta\}$ and $\{z\gamma,w\delta\}$ are parallel if $x=z$, $\alpha=\gamma$, $y=w$, and $\beta=\delta$).

\subsection{Bubble-like subgraphs} All bubbles we consider are characterized by two vertices $u$ and $v$ that are the ``extremities'', or ``endpoints`` of the bubble. Intuitively, if one enters a bubble from the outside of the bubble, they have to go through an extremity and similarly exit through an extremity. In other words, the extremities form a 2-separator of the underlying undirected graph in the sense that their removal separates the interior of the bubble from the rest of the graph (except fo special cases which we can handle separately). We actually require a mildly stronger property that is formalized under the notion of \emph{splitting}, where we pick a vertex $v$ and a direction (for directed graphs) or a vertex-side $v+$ or $v-$ (for bidirected graphs), create a copy $v'$ of $v$, and finally detach the edges of the opposing direction/vertex-side from $v$ and reattach them to $v'$. This is illustrated in \Cref{fig:bubbles}~(c). We then require that a bubble characterized by the extremities $u$ and $v$ has $u$ and $v$ in the same connected component that is \emph{separated} from $u'$ and $v'$ after splitting $u$ and $v$.
Further, all our bubbles have to be \emph{minimal}, intuitively meaning that they are not obtainable by concatenating smaller bubbles.

\emph{Snarls} are precisely defined by these separability and minimality properties in bidirected graphs, with examples provided in \Cref{fig:bubbles}~(b)--(c). Snarls are relatively weak bubbles in the sense that they lack any assumptions about their interior such as all internal vertices being reachable from an extremity in (bi)directed sense. In contrast, \emph{superbubbles} and \emph{ultrabubbles} require that the component containing $u$ and $v$ after splitting them has no cycloids and that for any internal vertex $w$ there is a (bi)directed walk (path, in fact) from one extremity to another that goes through~$w$. Superbubbles and ultrabubbles are illustrated in \Cref{fig:bubbles}~(a) and (d), respectively.

\begin{figure}[t]
    \centering
    \includegraphics[width=\linewidth]{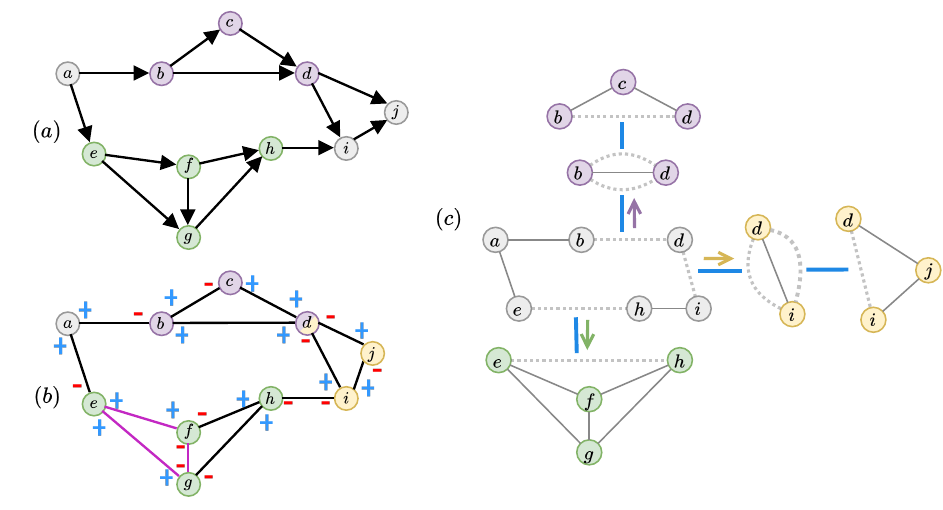}
    \caption{Superbubbles, snarls and ultrabubbles in a directed or bidirected graph naturally map to the SPQR tree of their underlying undirected graph. In (a) we have a directed graph with superbubbles $(b,d)$ (violet) and $(e,h)$ (green). (Note that also $(a,i)$ is a superbubble, but since it is the whole graph, it will be handled separately.) In (b) we have a bidirected graph with the same underlying undirected graph as the graph in (a). We have that $\{b+,d+\}$ (violet), $\{d-,i+\}$ (yellow) are snarls, that are also ultrabubbles. The pair $\{b+,i+\}$ is not a snarl because of minimality. The snarl $\{e+,h+\}$ (green) is not an ultrabubble as it contains a cycloid (purple) in its interior (alternating signs at all vertices except $e$). In (c) we show the SPQR tree of the underlying undirected graph of these graphs, where the different bubble subgraphs are shown with the same color as in the graphs in (a) and (b). Some snarls, such as the single-edge snarl $\{h-,i-\}$ do not arise from a 2-separator, but are handled separately when traversing the SPQR tree (in this case, when processing the middle S-node).}
    \label{fig:spqr-tree-bubbles}
\end{figure}

\subsection{Superbubbles} Since (nearly all) superbubbles correspond to some $2$-separator of the graph, our main technique is to exploit the decomposition of the $2$-separators provided by the SPQR tree, encoded by its virtual edges (see \Cref{fig:spqr-tree-bubbles}). For each $2$-separator, we need to decide whether the subgraph induced by the vertex set $C$ that the virtual edge points to corresponds to a superbubble, but we also need the other direction of whether the subgraph induced by the complement of $C$ is a superbubble. Most of our efforts in finding superbubbles and ultrabubbles is in computing the required properties in these ``two sides'' of the separations. (Exceptionally, P-nodes require special treatment since they encode many separations alone, but ultimately do not raise any issues due to their fixed topology.)

To solve these problems, we start by observing that the property of the desired walks existing inside the superbubble is equivalently captured by the lack of internal sources and sinks supposing that we know the graph to be acyclic. If the induced subgraph corresponding to some virtual edge contains cycles, sources, or sinks, then so do all of its induced supergraphs. Therefore, we traverse the SPQR tree with depth-first search starting from an arbitrary root~$r$, compute the information for the subgraphs, and finally combine the subresults to deduce the acyclicity and the existence of sources and sinks. 
For high-level visualization, see \Cref{fig:phase1}. To identify superbubbles from this information (and also to identify other bubbles), we then need to perform careful case analysis on how they can manifest in each type of a node of the SPQR tree.

\begin{figure}[t]
    \centering
    \includegraphics[width=\linewidth]{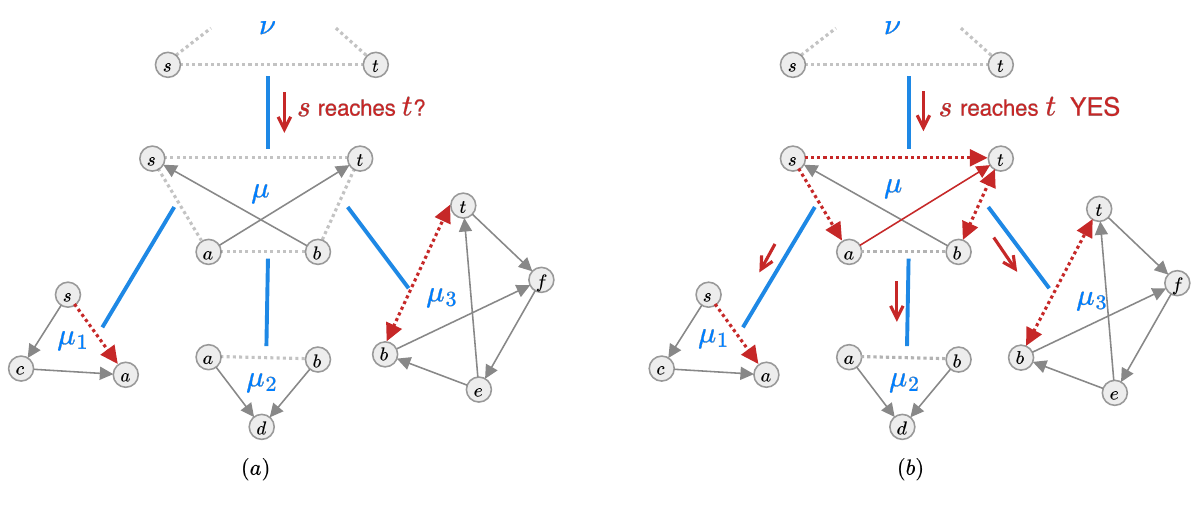}
    \caption{Phase~1 of the superbubble-finding algorithm (bottom-up DFS traversal of the SPQR tree). In~(a), the SPQR tree is shown with an R-node~$\mu$ having children $\mu_1, \mu_2, \mu_3$ and parent~$\nu$. The skeletons are depicted with dotted edges (virtual) and solid edges (real, maintaining the directions in the input directed graph). In this example, the goal is to determine whether $s$ reaches~$t$ in $\expansion(e_\mu)$, knowing the reachabilities between the 2-separators $\{s,a\}$, $\{a,b\}$, and $\{b,t\}$ in the children $\mu_1$, $\mu_2$, and $\mu_3$ (in red), respectively. In~(b), retrieving these information from the children, each expansion is collapsed into a directed edge in $\dirskel(\mu)$ (red arrows). A traversal of this directed skeleton confirms that $s$ reaches~$t$ (notice that the algorithm does not require an explicit traversal in order to determine this reachability relation, see \Cref{lem:phases-correct}).}
    \label{fig:phase1}
\end{figure}

\begin{figure}
    \centering
    \includegraphics[width=0.5\linewidth]{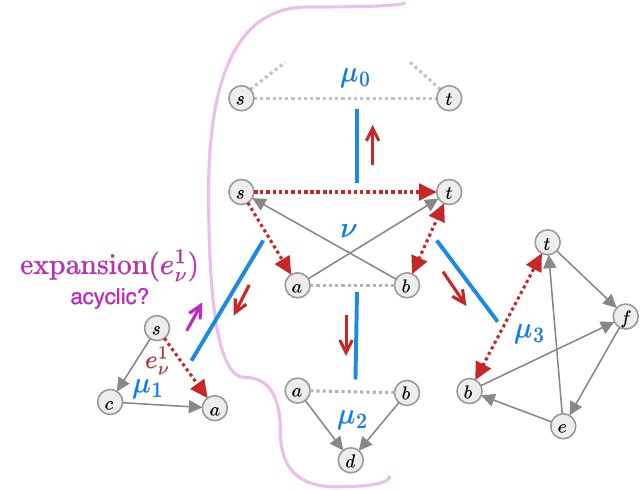}
    \caption{Phase~2 of the superbubble finding algorithm (BFS from the root).
    The algorithm processes node~$\nu$ and determines for each neighbor~$\mu_i$ whether $\expansion(e_\nu^i)$ is acyclic, for $\iink$.
    The red arrows represent states already computed (due to Phase 1 and the BFS traversal order in Phase 2).
    The purple region shows $\expansion(e_\nu^1)$.
    To handle all these tree nodes/states in linear-time in the size of $\skel(\nu)$, the feedback-arcs of $\dirskel(\nu)$ are computed.}
    \label{fig:phase2}
\end{figure}

This procedure identifies the superbubbles for each $2$-separator in one direction---in the separated component without the vertices of $r$---but not in the other. For that direction, we instead need to know that if we were to remove the subgraph corresponding to the virtual edge, would the remaining subgraph be a superbubble.
The main idea here is that if the subgraph corresponding to a virtual edge between $u$ and $v$ is acyclic and has no sources and sinks, then we can ``collapse'' it into a single arc whose direction is determined by whether all walks in the subgraph go from $u$ to $v$ or from $v$ to $u$. If we then collapse all the virtual edges of the node, then identifying the remaining superbubbles reduces to finding the feedback arcs among the collapsed virtual edges, that is, arcs whose removal makes the graph acyclic. Such arcs are computable in linear time in the number of arcs \citep{garey1978linear}, and SPQR trees contain only linearly many virtual edges in the size of the input. The process is illustrated in \Cref{fig:phase2}. Consequently, we get a linear-time algorithm for identifying all superbubbles.


\begin{theorem}
    The superbubbles of a directed graph $G$ can be computed in time $O(|V(G)| + |E(G)|)$.
\end{theorem}

\subsection{Snarls}

Because snarls only require separability and minimality, identifying them with the SPQR tree should intuitively be more straightforward than identifying ultrabubbles. On the other hand, the lack of structural requirements makes it possible for there to be quadratically many of them in the size of the input; this occurs for example in a clique of tips.
To solve this issue, we provide a novel characterization of snarls and then provide a concise representation of all snarls based on that, whose size is only linear in the size of the input.

We start by identifying a subset of cutvertices of the graph such that the extremities of a snarl cannot be in distinct connected components after \emph{splitting} of any of them. These cutvertices have a property which we call \emph{sign-consistency}: a sign-consistent vertex becomes a tip in each component that is created after being split (not necessarily with the same sign in each component). By splitting each of these vertices we obtain a set of disjoint graphs that we call \emph{sign-cut graphs}, which preserve the set of all original snarls and where every snarl has its extremities in a single sign-cut graph.

We then observe that the extremities of each snarl are either (i) a pair of tips or (ii) a pair of non-tips, within a (unique) sign-cut graph.
For the snarls of type (i), we compile a list of tips for each sign-cut graph, enabling us to capture a quadratic number of snarls with linear-sized lists. Crucially, there are only linearly many snarls of type (ii), which we find by analyzing the nodes of the SPQR tree. Ultimately, we obtain the next result.

\begin{restatable}{theorem}{maintheoremsnarls}
    \label{thm:main}
    Given a bidirected graph $G$, there exists a representation of the set of all snarls of $G$ consisting of sets $T_1, T_2, \dots, T_k$ and $S_1, S_2, \dots, S_\ell$, where
        \begin{enumerate}[nosep]
            \item each $T_i$ is a set of vertex-sides of $G$, and any pair of vertex-sides in $T_i$ identifies a snarl of $G$;
            \item each $S_i$ is a pair of vertex-sides $\{u\alpha, v\beta\}$ identifying a snarl of $G$;
            \item $\sum_{i=1}^{k} |T_i|=O(|V(G)|)$ and $\sum_{i=1}^{\ell} |S_i| = O(|V(G)|)$.
        \end{enumerate}
    Moreover, this representation can be computed in time $O(|V(G)| + |E(G)|)$.
\end{restatable}

For a concrete example, cutvertex $b$ is sign-consistent in \Cref{fig:bubbles}~(b). After splitting it, we would get a sign-cut graph with tips $a$, $b$, and $c$ and another sign-cut graph with tips~$b$ and $e$. The remaining snarls $\{d+, f-\}$ and $\{d-, f+\}$ correspond to a pair of non-tip vertices~$d$ and~$f$.

\subsection{Ultrabubbles}
Ultrabubbles were introduced as a canonical generalization of superbubbles to bidirected graphs. Their similarities were exploited in~\cite{Harviainen2026.03.28.714704} in order to get a linear-time algorithm for finding ultrabubbles. This algorithm works by converting the input bidirected graph into a directed graph such that an ultrabubble in the original instance corresponds to a ``weak''\footnote{We do not need the exact definition of \emph{weak superbubble} in this paper. The original definition can be found in~\citep{gartner2019direct}.} superbubble in the transformed instance, and vice-versa. As a corollary, they showed that ultrabubbles are ``directable'', i.e., they can be cast on a directed graph.
The approach of~\cite{Harviainen2026.03.28.714704} has the limitation of requiring the input graph to contain a tip or a cutvertex.

Although ultrabubbles are originally defined in a way where their resemblance with superbubbles is not entirely clear, both these objects share the following key reachability property: every vertex in the superbubble/ultrabubble lies in some directed/bidirected path between the extremities of the bubble. In fact, many if not essentially all the structural results we present in detail for superbubbles can be adapted to ultrabubbles in a straightforward manner. As directed graphs are more common in the literature, we first present our SPQR tree approach to find superbubbles in linear time and then go on to show how to adapt this algorithm to find ultrabubbles also in linear time.

To achieve this within our SPQR-tree-methodologies we must be able to perform the following two procedures in linear time: testing whether a bidirected graph has cycloids and finding its \emph{bidirected feedback edges}, that is, edges whose deletion destroys all cycloids (i.e., the resulting bidirected graph is acyclic).
We observe this to be impossible in general under the $k$-\textsc{Clique Conjecture} (see Conjecture 10 of~\cite{Kunnemann24}), which asserts in particular that a triangle---a clique of $3$ vertices---cannot be found in an undirected graph in time $O(|V|^{\omega - \epsilon})$ for the matrix multiplication exponent~$\omega$ and any $\epsilon > 0$. The result follows by a relatively simple reduction, where we essentially associate appropriate vertex-sides to each undirected edge of the graph such that any cycloid is an undirected triangle and vice versa; see \Cref{fig:triangle} for an example.

\begin{figure}[H]
    \centering
    \includegraphics[width=0.6\linewidth]{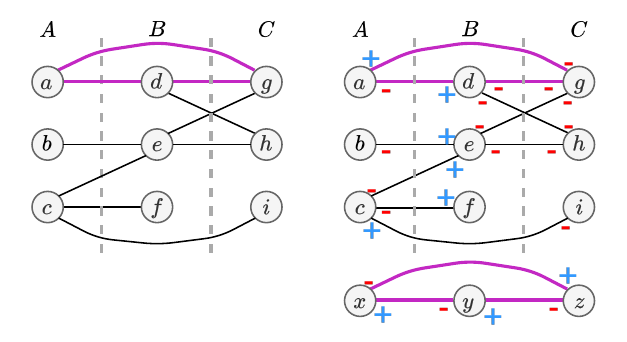}
    \caption{The reduction from finding triangles in tripartite graphs (left) to finding bidirected feedback edges, or finding cycloids, in bidirected graphs (right), where the vertices $x$, $y$, and $z$ are omitted in the latter case. Roughly speaking, vertex-sides are associated to each undirected edge based on the tripartite sets of the endpoints of the edges such that triangles correspond to closed walks. If a triangle exists then there are no bidirected feedback edges because there are two disjoint closed walks (in violet).}\label{fig:triangle}
\end{figure}

\begin{theorem}
    Let $G$ be a bidirected graph. Under the $k$-\textsc{Clique Conjecture}, we cannot decide if a bidirected graph is acyclic (i.e. it has no cycloid) or it has at least one bidirected feedback edge, in time $O(|V(G)|^{\omega - \epsilon})$ for any $\epsilon > 0$.
\end{theorem}

Fortunately, we are able to exploit the special structure of ultrabubbles to solve these problems in linear time. We observe the problems to be linear-time solvable in graphs without tips, which is a property of ultrabubbles.
Our solution involves an ear-addition procedure exploiting the bidirected reachability properties of the graph being constructed. If the construction succeeds then the procedure outputs a strongly connected directed graph with the same set of cycles as the original bidirected graph, and so we can use the linear-time algorithm of~\cite{garey1978linear} to find all feedback edges. If the procedure halts before having built the whole graph then we can correctly output that no feedback edge exists.

\begin{theorem}
    Let $G$ be a bidirected graph without tips. Then the set of feedback edges of $G$ can be computed in time $O(|V(G)| + |E(G)|)$. Further, we can decide whether $G$ has a cycloid in time $O(|V(G)| + |E(G)|)$.
\end{theorem}

\begin{theorem}
    The ultrabubbles of a bidirected graph $G$ can be computed in time $O(|V(G)| + |E(G)|)$.
\end{theorem}

\section{Superbubbles}
\label{sec:superbubbles}


In this section we develop a linear-time algorithm to find superbubbles in directed graphs.


\paragraph*{Basic notions}

Onodera et al. introduced in~\citep{onodera2013detecting} the notion of superbubble and motivate it as being a natural generalization of bubbles - a structurally simple graph motif appearing in graphs built from biological data.
A superbubble $(s,t)$ is a minimal acyclic vertex-induced subgraph by the set of vertices reachable from $s$ without going through $t$, which must coincide with the set of vertices that reach $t$ without going through $s$. This property is called the \emph{matching} property of superbubbles, so in particular every vertex in the superbubble lies in some $s$-$t$ path (see~\citep{onodera2013detecting} for the formal definition). In the same paper Onodera et al. showed that every vertex is the entry (and exit) of at most one superbubble.

A useful relaxation of superbubbles is that of superbubbloids, introduced by~\cite{gartner-revisited, gartner2019direct}; superbubbloids are superbubbles without the minimality constraint. Moreover,~\citep{gartner-revisited, gartner2019direct} defines (and proves equivalence with) these constructs in a way that is more suitable for our SPQR tree approach, as it exposes better the visual intuition that superbubbles are attached to the rest of the graph only by its defining vertices.

\begin{definition}[Superbubbloid~\citep{gartner2019direct}]\label{def:superbubbloid-gartner}
    Let $G$ be a directed graph, $X \subseteq V(G)$, and $s,t \in X$. The pair $(s,t)$ is a \emph{superbubbloid} with \emph{superbubbloid graph} $G[X]$ if:
    \begin{enumerate}[nosep]
        \item every $u \in X$ is reachable from $s$,
        \item $t$ is reachable from every $u \in X$,
        \item if $u \in X$ and $w \in V(G) \setminus X$, then every $w$-$u$ directed path contains $s$,
        \item if $u \in X$ and $w \in V(G) \setminus X$, then every $u$-$w$ directed path contains $t$,
        \item if $uv$ is an edge in $G[X]$, then every $v$-$u$ directed path in $G$ contains both $t$ and $s$, and
        \item $G$ does not contain the edge $ts$.
    \end{enumerate}
\end{definition}

Let $(s,t)$ be a superbubbloid with graph $B$. Vertex $s$ is the \emph{entry} and vertex $t$ is the \emph{exit} of the superbubbloid. The \emph{interior} of $(s,t)$ is the set $V(B) \setminus\{s,t\}$; notice that the interior of a superbubbloid does not contain sources or sinks.
Since the pair $(s,t)$ uniquely defines the superbubbloid graph $B$ and the superbubbloid graph $B$ uniquely defines the pair $(s,t)$, we refer to both the graph and the pair of vertices simply as ``superbubbloid''.

A \emph{superbubble} $(s,t)$ is a superbubbloid where no $s'\in V(B) \setminus \{s\}$ is such that $(s',t)$ is a superbubbloid. An edge $st$ with $N^+_G(s)=\{t\}$, $N^-_G(t)=\{s\}$, and $ts \notin E(G)$ is a \emph{trivial superbubble} (in fact, the original notion of ``bubbles'' essentially coincides with that of trivial superbubbles).

Next we give some results on the relation between cutvertices, superbubbloids, and superbubbles. Importantly, we show that cutvertices are not in the interior of any superbubble. 

\begin{lemma}
\label{lem:superbubbloid-superbubble-cutvertex}
    Let $G$ be a directed graph and let $(s,t)$ be a superbubbloid of $G$ with graph $B$. Then $(s,t)$ is a superbubble if and only if no vertex in the interior of $B$ is an $s$-$t$ cutvertex in $B$.
\end{lemma}
\begin{proof}
    The statement holds if $B$ is trivial, so suppose $B$ has at least three vertices.

    $(\Rightarrow)$ (Trivial.)
    
    $(\Leftarrow)$ Suppose for a contradiction that $B$ has a vertex $v \neq s,t$ violating the minimality of $(s,t)$.
    Notice that $B$ has an $s$-$t$ directed path avoiding $v$ for otherwise $v$ is an $s$-$t$ cutvertex.
    So $s$ reaches $t$ without $v$ and so $s$ is contained in the superbubbloid graph of $(v,t)$, and consequently $v$ reaches $s$ without $t$. But $s$ reaches $v$ without $t$ because $(s,t)$ is a superbubbloid, and thus $B$ has a cycle, a contradiction.
\end{proof}

\begin{corollary}
\label{cor:superbubbles-biconnected}
    Superbubbles are biconnected.
\end{corollary}
\begin{proof}
    Suppose there is a vertex $v$ in a superbubble $(s,t)$ with graph $B$ such that $B-v$ is disconnected with components $C'_1,\dots,C'_\ell$ $(\ell \geq 2)$.
    Let $C_i:=B[V(C'_i)+v]$ for each $i\in\{1,\dots,\ell\}$.
    Observe that there is exactly one component $C_i$ where both $s$ and $t$ appear, since if $v \neq s,t$ then $v$ does not separate $s$ and $t$ because \Cref{lem:superbubbloid-superbubble-cutvertex} gives us two internally vertex-disjoint $s$-$t$ paths in $B$ (and if $v=s$ or $v=t$ this follows trivially).
    Moreover, any component $C_j$ but the one containing $s$ and $t$ contains at most one source or sink, which is $v$, since $B$ has no sources or sinks besides $s$ and $t$ due to conditions 1 and 2 of superbubbles (see \Cref{def:superbubbloid-gartner}). Therefore $C_j \subseteq B$ has a cycle, a contradiction to the acyclicity of superbubbles.
\end{proof}

\begin{lemma}[Superbubbles and cutvertices]\label{lem:bubbles-cutvertices}
    Let $G$ be a directed graph and let $(s,t)$ be a superbubble of $G$ with graph $B$. Then no vertex in the interior of $B$ is a cutvertex of $G$.
\end{lemma}
\begin{proof}
    We can assume that $B$ contains at least three vertices.
    Suppose for a contradiction that the interior of $B$ contains a cutvertex of $G$. 
    There are two cases to analyze.

    \textbf{No block contains both $s$ and $t$:}
    Since no block contains both $s$ and $t$ there is a vertex $v$ whose removal disconnects $s$ from $t$ in $G$. Since $(s,t)$ is a superbubble and every $s$-$t$ path in $U(G)$ contains $v$, every $s$-$t$ directed path in $G$ also contains $v$. Thus $v$ is reachable from $s$ without $t$ and so $v \in V(B)$. Since $B \subseteq G$, vertex $v$ is an $s$-$t$ cutvertex in $B$, which is a contradiction by \Cref{lem:superbubbloid-superbubble-cutvertex}.

    \textbf{There is a block containing $s$ and $t$:}
    Let $v$ be a cutvertex of $G$ in the interior of $B$.
    Removing $v$ from $G$ results in a disconnected graph where at least one component does not contain both $s$ and $t$ since one block of $G$ contains both $s$ and $t$.
    Thus, let $w$ be a vertex in such a component adjacent to $v$ in $G$.
    Notice that regardless of whether $vw \in E(G)$ or $wv \in E(G)$, we have $w \in V(B)$ because $v \in V(B)$.
    So, in $G$, $w$ is reachable from $s$ without $t$ and reaches $t$ without $s$, but any two directed paths witnessing these reachabilities contain $v$, and so there is a cycle in $B$ through $v$, a contradiction.
\end{proof}

By \Cref{lem:bubbles-cutvertices} a cutvertex of $G$ can only be the entry or the exit of a superbubble. Therefore, superbubble graphs are confined to the blocks of $G$, and moreover there is a unique block that contains both the entrance and exit of of any given superbubble.
Then the task of computing superbubbles in a directed graph $G$ reduces to that of computing superbubbles in each block of $G$. Since block-cut trees can be built in linear time, if we can find superbubbles inside a block in linear time then we can find every superbubble of an arbitrary graph also in linear time.

The next result explains why (interesting) superbubbles induce 2-separators of the underlying undirected graph.

\begin{theorem}[Superbubbles and split pairs]\label{thm:bubbles-split-pairs}
    Let $G$ be a directed graph and let $(s,t)$ be a superbubble of $G$ with graph $B$. Let $H_1,\dots,H_{\ell}$ be the blocks of $G$ $(\ell \geq 1)$. Then $\{s,t\}$ is a split pair of some $H_i$ or $V(B)=V(H_i)$.
\end{theorem}
\begin{proof}
    It follows from~\Cref{lem:bubbles-cutvertices} that $V(B)$ is contained in a block of $G$, say, without loss of generality, of $H_1$.
    Suppose that $|V(B)| \geq 3$, $V(B) \neq V(H_1)$. Then it suffices to show that $\{s,t\}$ disconnects $H_1$.
    Let $u \in V(B) \setminus \{s,t\}$ and let $v \in V(H_1) \setminus V(B)$.
    By (3) and (4) of \Cref{def:superbubbloid-gartner}, every $u$-$v$ path in $G$ contains $t$ and every $v$-$u$ path in $G$ contains $s$. But then every $u$-$v$ path in $H_1$ contains $s$ or $t$, and therefore $\{s,t\}$ is a separation pair of $H_1$.
\end{proof}

The interesting case of~\Cref{thm:bubbles-split-pairs} is when the vertices identifying a superbubble form a separation pair of a block, as the other two cases can be dealt with a linear-time preprocessing step (with a single graph traversal it is possible to decide if the whole block is a superbubble, and with a simple predicate it is possible to decide whether an edge is a trivial superbubble).
By~\Cref{lem:spqr-tree-contains-split-pairs} we know that every separation pair of a graph is encoded as the endpoints of a virtual edge of some node of the SPQR tree, or as nonadjacent vertices in an S-node (but these cannot form superbubbles unless its superbubble graph is the whole block, as the next result shows; see~\Cref{fig:nonadj-snode}).
Conversely, the endpoints of any virtual edge in a node of the SPQR tree forms a separation pair.
Therefore, by correct examination of the virtual edges of the SPQR tree we can obtain the complete set of superbubbles of $G$ (see \Cref{fig:spqr-tree-bubbles} for an example).

\begin{proposition}
\label{prop:nonadjacent-S-node}
    Let $G$ be a directed graph, $H$ a maximal 2-connected subgraph of $G$, $T$ the SPQR tree of $H$, and $\mu \in V(T)$ an S-node with $u,v \in V(\skel(\mu))$. If $\{u,v\} \notin E(\skel(\mu))$ then $(u,v)$ and $(v,u)$ are not superbubbles, unless the corresponding graph is $H$.
\end{proposition}
\begin{proof}
    Suppose for a contradiction and without loss of generality that $(u,v)$ is a superbubble with graph $B$. Assume $B \neq H$.
    Let $e_1,\dots,e_\ell$ and $f_1,\dots,f_k$ denote the sequence of edges in $\skel(\mu)$ in the two $u$-$v$ paths of this S-node ($\ell,k \geq 2$ because $\{u,v\} \notin E(\skel(\mu))$).
    Since superbubbles are contained in blocks and $u,v\in V(H)$ we have $B \subseteq H$ and, due to the structure of the S-nodes, it is not hard to see that any $u$-$v$ directed path contains the vertices which are endpoints of the edges $e_i$ or of the edges $f_j$, for $i \in \{1,\dots,\ell\}$ and $j \in \{1,\dots,k\}$.
    Thus the set of vertices of $B$ contains at least the endpoints of the edges $e_i$ or those of $f_j$ (possibly of both). Suppose without loss of generality that $B$ contains those of the edges $e_i$. We claim that $V(\bigcup_{i=1}^\ell \expansion(e_i)) \subseteq V(B)$.
    
    Suppose otherwise and let $w \in V(\bigcup_{i=1}^\ell \expansion(e_i)) \setminus V(B)$ and let $e_i=\{x,y\}$ be the edge such that $w \in V(\expansion(e_i))$, and suppose that $x$ denotes the vertex closest to $u$ in $\skel(\mu)$ (possibly $y=v$). By \Cref{lem:reaches-in-expansion} $\expansion(e_i)$ has an $x$-$w$ path $p$ avoiding $y$, which starts in a vertex in $B$ and ends in a vertex not in $B$. Let $x'$ denote the last vertex in $p$ that is in $B$ (possibly $x'=x$). So $x'$ has a successor $y'$ in $p$ which is not contained in $B$ (notice that $y' \neq v$ since $p$ avoids $y$).
    Thus $\expansion(e_i)$ has a directed edge $x'y'$ or $y'x'$.
    Since $x' \in V(B)$, $B$ has a $u$-$x'$ path avoiding $v$ and an $x'$-$t$ path avoiding $u$.
    Therefore, if $x'y' \in E(\expansion(e_i))$ then $\expansion(e_i)$ has a $u$-$y'$ path avoiding $v$, thus $y' \in V(B)$, a contradiction, and similarly a contradiction is obtained if $y'x' \in E(\expansion(e_i))$. Hence $V(\bigcup_{i=1}^\ell \expansion(e_i)) \subseteq V(B)$.

    If $B$ has a vertex not in $V(\bigcup_{i=1}^\ell \expansion(e_i))$ then such a vertex can be chosen so that it is an endpoint of an edge $f_j$. Thus the same argument as above can be made to conclude that $V(\bigcup_{i=1}^k \expansion(f_i)) \subseteq V(B)$.
    But notice that either $V(\bigcup_{i=1}^\ell \expansion(e_i)) \subseteq V(B)$ or $V(\bigcup_{i=1}^k \expansion(f_i)) \subseteq V(B)$ hold, for if both hold then $V(H) \subseteq V(B)$ and since superbubbles are contained in blocks we get $V(B)=V(H)$ and hence $B=H$, a contradiction.
    So $V(B) \subseteq V(\bigcup_{i=1}^\ell \expansion(e_i))$ without loss of generality and thus $V(B)=V(\bigcup_{i=1}^\ell \expansion(e_i))$.
    Due to the structure of S-nodes and the fact that $\ell \geq 2$ we can conclude that $B$ has a $u$-$v$ cutvertex, a contradiction to \Cref{lem:superbubbloid-superbubble-cutvertex}.

\end{proof}

\begin{figure}[t]
    \centering
    \includegraphics[width=\linewidth]{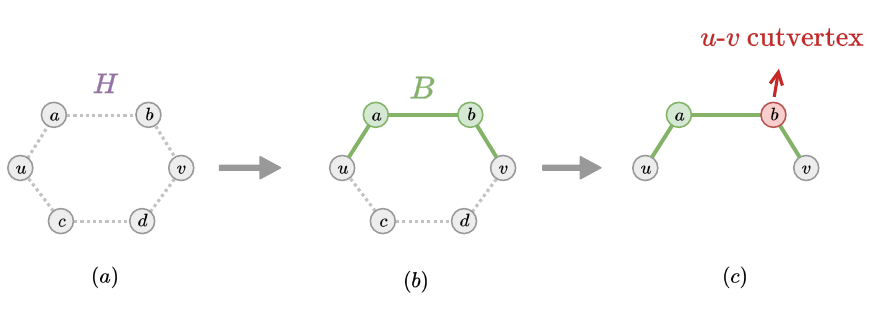}
    \caption{Why nonadjacent vertices in an S-node do not form superbubbles (\Cref{prop:nonadjacent-S-node}). (a)~S-node skeleton of a block $H$ with $u$ and $v$ nonadjacent. There are two $u$-$v$ paths in the cycle. (b)~The superbubble graph $B$ must contain one entire side. (c)~Since this side has at least two edges, there is a $u$-$v$ cutvertex in $B$ (red), a contradiction. If both sides lie in $B$ then $B=H$, a contradiction.}
    \label{fig:nonadj-snode}
\end{figure}

The next result is merely technical and will be used later on to show correctness of the superbubble finding algorithm.

\begin{lemma}[Unique orientation at poles of acyclic components]\label{lem:unique-poles-from-acyclicity}
    Let $G$ be a directed graph and let $C \subseteq V(G)$, $|C| 
    \geq 2$, be such that $G[C]$ is connected and $G[C]$ is acyclic. Moreover, let $s,t \in C$ be such that for all other vertices $v \in C \setminus \{s,t\}$ there is no edge in $G$ between $v$ and some $v' \in V(G) \setminus C$. If no vertex in $C \setminus \{s,t\}$ is a source or sink of $G$, then one vertex among $\{s,t\}$ is the unique source of $G[C]$ and the other vertex is the unique sink of $G[C]$.
\end{lemma}
\begin{proof}
    Notice that since $G[C]$ is acyclic, it has at least one source (relative to $G[C]$), say $v$. If $v \in C \setminus \{s,t\}$, then it is also a source of $G$, since by the hypothesis there is no edge in $G$ between $v$ and some $v' \in V(G) \setminus C$. This contradicts the assumption that no vertex in $C \setminus \{s,t\}$ is a source of $G$. Therefore, any source of $G[C]$ belongs to $\{s,t\}$.

    Symmetrically, we have that any sink of $G[C]$ belongs to $\{s,t\}$. Observe that neither $s$ nor $t$ can be both a source and a sink of $G[C]$, because otherwise it would be an isolated vertex with no edges in $G[C]$, which contradicts the assumption that $G[C]$ is connected. Therefore, since $G[C]$ has at least one source and at least one sink (being acyclic), one vertex among $\{s,t\}$ is the unique source of $G[C]$ and the other vertex is the unique sink in $G[C]$.
\end{proof}

\subsection{Setup}
\label{subsec:superbubble-finding-algorithm}

Here, we focus on superbubbles that induce a separation pair of a block $H$. The next result will be applied frequently later on.

\begin{lemma}
\label{lem:separation-pair-is-superbubbloid}
    Let $G$ be a directed graph, let $\{s,t\}$ be a separation pair of a maximal 2-connected subgraph of $G$, and let $K$ be the union of a nonempty subset of the split components of $\{s,t\}$. If there are no extremities of $G$ in $V(K)\setminus\{s,t\}$, $K$ is acyclic, $N^+_G(s) \subseteq V(K)$, $N^-_G(t) \subseteq V(K)$, and $ts \notin E(G)$, then $(s,t)$ is a superbubbloid of $G$ with graph $K$.
\end{lemma}
\begin{proof}
    Let $K_1,\dots,K_k$ denote the set of split components of $\{s,t\}$ whose union is $K$ $(k\geq1)$.    
    We show that $K$ fulfills each condition of \Cref{def:superbubbloid-gartner}.

    \begin{enumerate}[nosep]
        \item Let $u \in V(K_i)$ and let $p=v_1,\dots,v_\ell$ be a maximal directed path in $K_i$ containing $u$ $(\ell \geq 1)$. Suppose that $v_1$ has in-neighbors in $K_i$.
        If $v_1$ has an in-neighbor in $p$ then this edge is in $K_i$ by maximality of split component, and so there is a cycle in $K_i \subseteq K$, a contradiction.
        If $v_1$ has an in-neighbor in $K_i$ that is not in $p$ then $p$ can be prolonged, contradicting its maximality. Therefore $v_1$ does not have in-neighbors in $K_i$ and thus $v_1=s$ since $s$ is the unique source of $K_i$. Analogously, we have $v_\ell=t$. As $\bigcup_{i=1}^kV(K_i)=V(K)$ it follows that every vertex in $K$ lies in some $s$-$t$ directed path, and thus it is reachable from $s$ and reaches $t$.
        
        \item (Proved in the previous item.)
        
        \item Let $u\in V(K)$ and $w\in V(G) \setminus V(K)$. If $u=s$ we are done. Otherwise, suppose for a contradiction that $G$ has a $w$-$u$ directed path avoiding $s$. Then, since $K$ consists of the union of split components of $\{s,t\}$, this path contains $t$, a contradiction to the fact that every in-neighbor of $t$ is contained in $K$. Therefore every $w$-$u$ path in $G$ contains $s$.

        \item (Analogous to the previous item.)
        
        \item Let $uv \in E(K)$ and let $K_i$ denote the split component that contains the edge $uv$. Suppose for a contradiction that $G$ has a $v$-$u$ path avoiding, say, $s$. Then this path is contained in $K_i$ because $K_i$ is a split component (if the path leaves $K_i$ it must do it through $t$, but then it cannot reenter since paths do not repeat vertices). This path together with the edge $uv$ forms a cycle in $K_i$, a contradiction. Therefore every $v$-$u$ path in $G$ contains both $t$ and $s$.
        
        \item Direct by assumption.
    \end{enumerate}
\end{proof}

We assume in what follows that the SPQR tree of $H$ does not consist of a single node, since otherwise any superbubble is either $H$ or a single edge by \Cref{thm:bubbles-split-pairs} and \Cref{prop:nonadjacent-S-node}. 

Let $T$ be the SPQR tree of $H$.
Let $\{\nu,\mu\} \in E(T)$ such that $\nu$ is the parent of $\mu$, and let $e_\nu$ be the virtual edge in $\skel(\mu)$ pertaining to $\nu$ and define $e_\mu$ analogously; let $\{s,t\}$ denote the endpoints of these virtual edges. In the edge $\{\nu,\mu\}$ we store two pieces of information: the state corresponding to the subgraph $\expansion(e_\mu)$ as $\state{\nu}{\mu}$ and the state corresponding to $\expansion(e_\nu)$ as $\state{\mu}{\nu}$. We say that $\state{\nu}{\mu}$ \emph{leaves} $\nu$ and that it \emph{enters} $\mu$. (This can be seen as a directed edge in the tree pointing from $\nu$ to $\mu$.) Let $X=\expansion(e_\mu)$. Notice that any state uniquely identifies a virtual edge, in this case $e_\mu$. In $\mathsf{State_{\nu,\mu}}$ we store the following information:

\begin{itemize}
    \item $\noextremity{\nu}{\mu} := \true$ iff no vertex in $V(X)\setminus\{s,t\}$ is an extremity of $G$.
    
    \item $\acyclic{\nu}{\mu} := 
    \begin{cases}
        \Null, & \text{if $\noextremity{\nu}{\mu}$ is false,}\\
        \true, & \text{otherwise, if $X$ is acyclic,}\\
        \false, & \text{otherwise.}
    \end{cases}$\\
    \item $\reachesst{\nu}{\mu} := 
    \begin{cases}
        \Null, & \text{if $\acyclic{\nu}{\mu}$ is $\false$ or $\Null$,}\\ 
        \true, & \text{otherwise, if $s$ reaches $t$ in $X$,}\\
        \false, & \text{otherwise.}
    \end{cases}
    $\\
    \item $\reachests{\nu}{\mu} := 
    \begin{cases}
        \Null, & \text{if $\acyclic{\nu}{\mu}$ is $\false$ or $\Null$,}\\ 
        \true, & \text{otherwise, if $t$ reaches $s$ in $X$,}\\
        \false, & \text{otherwise.}
    \end{cases}
    $
\end{itemize}

We define $\state{\mu}{\nu}$ in the same way, but where $X$ is the graph $\expansion(e_\nu)$ (i.e., this state now ``points'' from $\mu$ to $\nu$).
With this information we can ``almost'' decide if a separation pair $\{s,t\}$ identifies a superbubble $(s,t)$, since if $\noextremity{\nu}{\mu}$ and $\acyclic{\nu}{\mu}$ are $\true$, $N^+_G(s),N^-_G(t) \subseteq V(\expansion(e_\mu))$, and  $ts\notin E(G)$, then $(s,t)$ is a superbubbloid with graph $\expansion(e_\mu)$ by~\Cref{lem:separation-pair-is-superbubbloid}. This fact also hints that most of our effort is in the computation of $\state{\nu}{\mu}$ for all edges $\{\nu,\mu\}$ of $T$, as the other conditions are easy to check.

The algorithm consists of three phases.

\begin{itemize}
    \item \emph{Phase 1.} Process the edges $\{\nu,\mu\}$ of $T$ (with $\nu$ the parent of $\mu$) with a DFS traversal starting in the root, and compute all $\state{\nu}{\mu}$.
    \item \emph{Phase 2.} Process the nodes $\nu$ of $T$ with a BFS traversal starting in the root. For every child $\mu$ of $\nu$, compute all $\state{\mu}{\nu}$.
    \item \emph{Phase 3.} Examine the separation pairs $\{s,t\}$ of $H$ in $T$ and use the information computed in the previous phases to decide whether $(s,t)$ or $(t,s)$ are superbubbles.
\end{itemize}

\subsection{Algorithm - Phase 1}
Phase 1 is a dynamic program over the edges of $T$.
Let $\nu$ be the parent of $\mu$ in $T$ and let $\{s,t\}$ denote the endpoints of $e_\nu \in \skel(\mu)$ and of $e_\mu \in \skel(\nu)$.
If $\mu$ has no children then the edges of its skeleton but $e_\nu$ are all real edges, and hence the problem of updating $\state{\nu}{\mu}$ is trivial: with one DFS on $\dirskel(\mu)-st-ts$ we can decide $\noextremity{\nu}{\mu}$, $\acyclic{\nu}{\mu}$, $\reachesst{\nu}{\mu}$, and $\reachests{\nu}{\mu}$.
Otherwise $\mu$ has at least one virtual edge besides $e_\nu$. Let us denote the children of $\mu$ by $\mu_1,\dots,\mu_k$ $(k\geq 1)$ and denote the endpoints of the corresponding virtual edges $e_i$ in $\skel(\mu)$ as $\{s_i,t_i\}$ for all $i \in \{1,\dots,k\}$.
Assume recursively that $\state{\mu}{\mu_i}$ is solved and let $X = \expansion(e_\mu)$, $X_i = \expansion(e_i)$ for all $i \in \{1,\dots,k\}$, and let $K=\dirskel(\mu)-st-ts$.

We now describe how to compute the states $\state{\nu}{\mu}$.
\begin{description}
    \item[$\noextremity{\nu}{\mu}$:] We set $\noextremity{\nu}{\mu}$ to $\true$ if and only if no vertex in $V(K) \setminus \{s,t\}$ is an extremity and $\noextremity{\mu}{\mu_i}$ is $\true$ for all $i \in \{1,\dots,k\}$. To see this is correct we prove both implications.
    
    ($\Rightarrow$) Suppose no vertex in $V(X) \setminus \{s,t\}$ is an extremity. Then indeed, no vertex in $V(K) \setminus \{s,t\}$ is an extremity because $V(K) \subseteq V(X)$.
    Moreover, $\noextremity{\mu}{\mu_i}$ must be $\true$ for all $i \in \{1,\dots,k\}$, for otherwise an extremity $x$ in $X_i$ is different from both $s_i$ and $t_i$ and thus also different from $s$ and $t$, as it does not belong to $\skel(\mu)$ since $\{s_i,t_i\}$ is a separation pair.

    $(\Leftarrow)$ Suppose no vertex in $V(K) \setminus \{s,t\}$ is an extremity and $\noextremity{\mu}{\mu_i}$ is $\true$ for all $i \in \{1,\dots,k\}$. For a contradiction, assume that some $x \in V(X) \setminus \{s,t\}$ is an extremity. By the initial assumption, we have that $x$ cannot belong to $V(K)$. Thus, $x$ is also different from $s_i,t_i$, for all $i \in \{1,\dots,k\}$. Since $x \in V(X) \setminus \{s,t\}$, $x$ is a vertex of $X_i$ for some $i \in \{1,\dots,k\}$. Therefore, it is an extremity for it, since it is different from $s_i$ and $t_i$. This contradicts the initial assumption that $\noextremity{\mu}{\mu_i}$ is $\true$.

    \item[$\acyclic{\nu}{\mu}$:] If $\noextremity{\nu}{\mu}$ is $\false$ then we set $\acyclic{\nu}{\mu}$ to $\Null$, which is correct by definition. Thus, in the following we assume that $\noextremity{\nu}{\mu}$ is $\true$.
    
    If for some $\iink$ $\acyclic{\mu}{\mu_i}$ is $\Null$, then by definition $\noextremity{\mu}{\mu_i}$ is $\false$. Let thus $x$ be an extremity in $V(X_i) \setminus \{s_i,t_i\}$. Since $\{s_i,t_i\}$ is a separation pair, we have that $x \notin \{s,t\}$. Thus, $x \in V(X) \setminus \{s,t\}$, which contradicts the fact $\noextremity{\nu}{\mu}$ is $\true$. Thus $\acyclic{\mu}{\mu_i}$ is $\true$ or $\false$ for each $\iink$.
    
    If for some $\iink$ $\acyclic{\mu}{\mu_i}$ is $\false$, then $X_i$ has a cycle, which implies that also $X$ contains a cycle because $X_i$ is a subgraph of $X$. Since $\noextremity{\nu}{\mu}$ is $\true$, then by definition we can set $\acyclic{\nu}{\mu}$ to $\false$.

    \begin{sloppypar}
    Finally, we are in the case where for every $\iink$, $\acyclic{\mu}{\mu_i}$ is $\true$, and therefore $\reachesst{\mu}{\mu_i}$ and $\reachests{\mu}{\mu_i}$ are $\true$ or $\false$. In other words, each $X_i$ is acyclic, and, importantly, the reachability in $X_i$ between the endpoints of each virtual edge $e_i$ are known.
    
    Then $K$ can be built explicitly and we can set $\acyclic{\nu}{\mu}$ to $\true$ if $K$ is acyclic and to $\false$ otherwise. To see this is correct, notice that any cycle $C$ in $X$ can be mapped to a cycle in $K$: whenever $C$ uses edges of some $X_i$, it passes through $s_i$ (or $t_i$), and since $X_i$ is acyclic, it must return to $t_i$ (or $s_i$). This path of $C$ in $X_i$ between $s_i$ and $t_i$ (or between $t_i$ and $s_i$) can be mapped to the edge of $K$ that was introduced from $s_i$ to $t_i$, if $\reachesst{\mu}{\mu_i}$ is $\true$ (or from $t_i$ to $s_i$, if $\reachests{\mu}{\mu_i}$ is $\true$). Viceversa, every cycle $C$ in $K$ can be symmetrically mapped to a cycle in $X$ such that whenever $C$ uses some edge $s_it_i$ (or $t_is_i$) in $K$, we expand this edge into a path from $s_i$ to $t_i$ (or from $t_i$ to $s_i$) in $X_i$.
    \end{sloppypar}
    
    \item[$\reachesst{\nu}{\mu}$, $\reachests{\nu}{\mu}$:] At this point we have $\acyclic{\nu}{\mu}$ computed. If it is $\false$ or $\Null$ we set $\reachesst{\nu}{\mu}$ and $\reachests{\nu}{\mu}$ to $\Null$, which is correct by definition. Otherwise $X$ is acyclic, $V(X) \setminus \{s,t\}$ has no cutvertex, no sink and no source of $G$, and there is no edge between a vertex not in $V(X)$ and a vertex in $V(X) \setminus \{s,t\}$ since $\{s,t\}$ is a separation pair; moreover $X$ is clearly connected. We can thus apply \Cref{lem:unique-poles-from-acyclicity} to conclude that one vertex between $s$ and $t$ is a source of $X$ and the other is a sink. In the former case we can set $\reachesst{\nu}{\mu}$ to $\true$ and $\reachests{\nu}{\mu}$ to $\false$, and in the latter case we can set $\reachests{\nu}{\mu}$ to $\true$ and $\reachesst{\nu}{\mu}$ to $\false$.
\end{description}

Let $\mu$ be a node of $T$ and let $\nu$ be its parent.
Since $T$ is a tree, each of its edges is processed exactly once during the DFS, thus every state of the form $\state{\nu}{\mu}$ is updated during this phase. Moreover, since every node $\mu$ has a unique parent in $T$ and $\dirskel(\mu)$ is built only when the edge $\{\nu,\mu\}$ is processed, each directed skeleton $\dirskel(\mu)$ is built only once during the algorithm. Since the computational work per edge $\{\nu,\mu\}$ is linear in the size of $\dirskel(\mu)$ and the total size of all skeletons is linear in the size of the current block $H$ (recall \Cref{lem:spqr-total-size}), Phase 1 runs in time $O( |V(H)| + |E(H)|)$.

\begin{algorithm}[t]
\small
\caption{Superbubble finding -- Phase 1}
\label{alg:phase1}
\KwIn{Directed graph $G$, SPQR tree $T$ having at least two nodes}

\SetKwFunction{ProcessEdge}{ProcessEdge}
\SetKwProg{Fn}{Function}{:}{}
\Fn{\ProcessEdge{$\nu,\mu$}}{
    \tcp{$\nu$ is the parent of $\mu$}
    \For{$\mu_i \in \mathsf{children}_T(\mu)$}{
        \ProcessEdge{$\mu,\mu_i$}\;
    }
    \If{$\nu = \text{null}$}{
        \Return{}\;
    }
    $\{s,t\} \gets e_\mu$ (where $e_\mu \in \skel(\nu)$)\;
    \textit{noExtr} $\gets \text{false}$ iff $V(\skel(\mu)) \setminus \{s,t\}$ has an extremity of $G$\;

    Let $e_1,\dots,e_k$ denote the virtual edges of $\mu$ with pertaining nodes $\mu_1,\dots,\mu_k$ $(k \ge 0)$\;
    \textit{ThereIsNoExtremityBelow} $\gets \bigwedge_{i=1}^k \noextremity{\mu}{\mu_i}$\;
    \textit{ThereIsNoCycleBelow} $\gets \bigwedge_{i=1}^k \acyclic{\mu}{\mu_i}$\;
    $\noextremity{\nu}{\mu} \gets$ \textit{noExtr} $\land$ \textit{ThereIsNoExtremityBelow}\;

    \If{$\noextremity{\nu}{\mu}$ is $\false$}{
        $\acyclic{\nu}{\mu} \gets \Null$\;
    }
    \Else{
        \If{$\neg$ ThereIsNoCycleBelow}{
            $\acyclic{\nu}{\mu} \gets \false$\;
        }
        \Else{
            \tcp{We are in conditions to build $\dirskel(\mu) - st - ts$}
            $K \gets \dirskel(\mu) - st - ts$\;
            $\acyclic{\nu}{\mu} \gets \true$ iff $K$ is acyclic\; \tcp{Run DFS or BFS on $K$}
        }
    }

    \If{$\acyclic{\nu}{\mu}$ is $\true$}{
        \If{$N^+_K(s) \cap V(K) \neq \emptyset$}{
            \tcp{See \Cref{lem:unique-poles-from-acyclicity}}
            $\reachesst{\nu}{\mu} \gets \true$;
            $\reachests{\nu}{\mu} \gets \false$\;
        }
        \Else{
            $\reachesst{\nu}{\mu} \gets \false$;
            $\reachests{\nu}{\mu} \gets \true$\;
        }
    }
    \Else{
        $\reachesst{\nu}{\mu} \gets \Null$;
        $\reachests{\nu}{\mu} \gets \Null$\;
    }
}

$\rho \gets$ the root of $T$\;
\ProcessEdge{$\text{null}, \rho$}\;

\end{algorithm}

\subsection{Algorithm - Phase 2} In Phase 2 we compute the states $\state{\mu}{\nu}$ with $\nu$ the parent of $\mu$ by processing the \emph{nodes} of $T$ via Breadth-First Search, i.e., we compute the states ``pointing'' towards the root.
Notice that the dependencies between states behave differently from Phase 1. Now the relevant states for $\state{\mu}{\nu}$ are those leaving $\nu$ to its children except $\mu$, and the state leaving $\nu$ to its parent whenever $\nu$ is different from the root of $T$; the former states are known from Phase 1 and the latter state is known due to the breadth-first traversal order. We remark that following the same strategy of computation as in Phase 1 may cause the algorithm to have a worst-case quadratic running time. For example, if $T$ consists only of node $\rho$ with children $\mu_1,\dots,\mu_k$, then in order to update $\acyclic{\mu_i}{\nu}$ we would have to build $\dirskel(\nu)-s_it_i-t_is_i$ for each $i=1,\dots,k$, which would have a quadratic running time in the size of the graph whenever, e.g., $|V(\skel(\nu))| \geq |V(G)|/2$. To overcome this issue we examine the states $\acyclic{\mu_i}{\nu}$ for $i=1,\dots,k$ all ``at once''.

Let $\nu$ be a node of $T$. Let $\mu_1,\dots,\mu_k$ denote the children of $\nu$ and denote the endpoints of the corresponding virtual edges in $\skel(\nu)$ as $e_i = \{s_i,t_i\}$ for $i \in \{1,\dots,k\}$ $(k \geq 1)$. To distinguish the reference edges $e_\nu$ belonging to each node $\mu_i$, we write $e_\nu^i$ for the edge $e_\nu$ in node $\mu_i$. Assume from the breadth-first traversal order that the states leaving $\nu$ to its parent are known and, for convenience, denote by $e_0=\{s_0,t_0\}$ the reference edge of $\nu$ and by $\mu_0$ the parent of $\nu$, so the neighbours of $\nu$ in $T$ are the nodes $\mu_0,\mu_1\dots,\mu_k$ (if $\nu$ is the root of $T$ then $\mu_0$ and $e_0$ can be ignored during the following discussion).
Let $X_i = \expansion(e_\nu^i)$, $K=\dirskel(\nu)$, and $K_i=K-s_it_i-t_is_i$ for every $\iink$.

First we compute $\noextremity{\mu_i}{\nu}$ for each $\iink$ similarly to Phase 1.

\begin{description}
    \item[$\noextremity{\mu_i}{\nu}$:] We set $\noextremity{\mu_i}{\nu}$ to $\true$ if and only if no vertex in $V(K_i) \setminus \{s_i,t_i\}$ is an extremity and $\noextremity{\nu}{\mu_j}$ is $\true$ for every $\jinkz\setminus \{i\}$. To see this is correct we prove both implications.
    
    ($\Rightarrow$) Suppose no vertex in $V(X_i) \setminus \{s_i,t_i\}$ is an extremity. Then indeed, no vertex in $V(K_i) \setminus \{s_i,t_i\}$ is an extremity.
    Moreover, $\noextremity{\nu}{\mu_j}$ must be $\true$ for each $\jinkz$ distinct from $i$, for otherwise an extremity in $\expansion(e_j)$ is different from both $s_j$ and $t_j$ and thus also different from $s_i$ and $t_i$, as it does not belong to $\skel(\nu)$ since $\{s_j,t_j\}$ is a separation pair.

    $(\Leftarrow)$ Suppose no vertex in $V(K_i) \setminus \{s_i,t_i\}$ is an extremity and that $\noextremity{\nu}{\mu_j}$ is $\true$ for all $\jinkz\setminus\{i\}$. For a contradiction, assume that a vertex $x \in V(X_i) \setminus \{s_i,t_i\}$ is an extremity. By the initial assumption, we have that $x$ cannot belong to $V(K_i)$. Thus, $x$ is also different from $s_j,t_j$ for each $\jinkz\setminus\{i\}$, and so $x$ is an extremity in $\expansion(e_j)$ for some $\jinkz \setminus \{i\}$. Therefore it is an extremity for it, since it is different from $s_j$ and $t_j$, contradicting the assumption that $\noextremity{\nu}{\mu_j}$ is $\true$.
\end{description}

Then we compute the states $\acyclic{\mu_i}{\nu}$ for all $\iink$.
Notice that at this point the states $\reachesst{\nu}{\mu_i}$ and $\reachests{\nu}{\mu_i}$ are known for all $\iinkz$. We proceed by cases on the values of these states.

\begin{itemize}

    \item If there is an $\iinkz$ such that $\reachesst{\nu}{\mu_i}$ or $\reachests{\nu}{\mu_i}$ is $\Null$, then by definition $\acyclic{\nu}{\mu_i}$ is $\Null$ or $\false$. Then we proceed by cases.

    \begin{itemize}
        \item If $\acyclic{\nu}{\mu_i}$ is $\Null$ then $\noextremity{\nu}{\mu_i}$ is $\false$ by definition, and so there is an extremity $x \in V(\expansion(e_i)) \setminus \{s_i,t_i\}$. So for every $\jink$ distinct from $i$, vertex $x$ is an extremity also for $X_j$: $x \in V(X_j)$ because $\expansion(e_i)$ is a subgraph of $X_j$ and $x$ is different from $s_j,t_j$ since $x$ is different from $s_i,t_i$ and $\{s_i,t_i\}$ is a separation pair; thus each state $\acyclic{\mu_j}{\nu}$ is $\Null$.
        
        For the remaining state $\acyclic{\mu_i}{\nu}$ we proceed by cases. First, if $\noextremity{\mu_i}{\nu}$ is $\false$ then $\acyclic{\mu_i}{\nu}$ is $\Null$. We can thus assume that $\noextremity{\mu_i}{\nu}$ is $\true$, which implies that $\acyclic{\nu}{\mu_j}$ is $\true$ or $\false$ for each $\jinkz$ distinct from $i$.
        If some $\acyclic{\nu}{\mu_j}$ is $\false$ then $\expansion(e_j)$ has a cycle, and hence so does $X_i$ as it is a supergraph of $\expansion(e_j)$; therefore $\acyclic{\mu_i}{\nu}$ is $\false$. Otherwise every $\acyclic{\nu}{\mu_j}$ is $\true$ and thus we are in conditions to build $K_i$ since the states $\reachesst{\nu}{\mu_j}$ and $\reachests{\nu}{\mu_j}$ are not $\Null$ by definition. Then $\acyclic{\mu_i}{\nu}$ is $\false$ if and only if $K_i$ has a cycle because any cycle in $X_i$ can be mapped to a cycle in $K_i$ (similarly to the acyclicity update rule discussed in Phase 1).
        \item Otherwise $\acyclic{\nu}{\mu_i}$ is $\false$. So $\expansion(e_i)$ contains a cycle $C$. Then for every $\jink$ with $j\neq i$, $\acyclic{\mu_j}{\nu}$ is $\false$ since $C \subseteq \expansion(e_i) \subseteq X_j$. For the remaining state $\acyclic{\mu_i}{\nu}$ we proceed identically as in the case above.
    \end{itemize}

    \item Otherwise $\reachesst{\nu}{\mu_i}$ and $\reachests{\nu}{\mu_i}$ are either $\true$ or $\false$ for all $\iinkz$. Therefore we are in conditions to build $K$. Moreover, the fact that the reachability states are all either $\true$ or $\false$ implies, by definition, that $\acyclic{\nu}{\mu_i}$ is $\true$ for all $\iinkz$; in particular, there is no cycle in $K$ of the form $s_it_is_i$. So $\acyclic{\mu_i}{\nu}$ is $\true$ if and only if $K_i$ is acyclic (the correctness of this argument was established in Phase 1), and hence it is enough to examine the acyclicity of $K_i$ in order to determine $\acyclic{\mu_i}{\nu}$, for every $\iink$. However, we do not test the acyclicity of each $K_i$ individually.
    
    First, notice that if $K$ is acyclic then so is $K_i$ because $K_i$ is a subgraph of $K$, in which case $\acyclic{\mu_i}{\nu}$ is $\true$ for all $\iink$.
    Otherwise $K$ has a cycle.

    Let $\iink$. Suppose that an edge $e_i \in E(K)$ intersects every cycle of $K$. Since $K_i=K-e_i$ it follows that $K_i$ is acyclic, and therefore $\acyclic{\mu_i}{\nu}$ is $\true$. Conversely, if $e_i$ does not intersect every cycle of $K$ then $K_i$ has a cycle, and therefore $\acyclic{\mu_i}{\nu}$ is $\false$.

    So in order to keep the algorithm linear-time, it suffices to identify the edges that intersect every cycle of the directed skeleton in time proportional to its size. This is essentially the feedback arc set problem for the restricted case where every feedback set contains just one arc\footnote{In its generality, the feedback arc set problem is an NP-hard problem which asks if a directed graph $G$ has a subset of at most $k$ edges intersecting every cycle of $G$. Here, we are interested in enumerating all feedback-arcs.} (``arc'' and ``edge'' mean the same thing).
    Our subroutine to find feedback-arcs works as follows. We start by testing if the graph is acyclic. If it is we are done. Otherwise we compute the strongly connected components (SCCs) of $G$. If there are multiple non-trivial SCCs then there are two disjoint cycles and no solution exists. Thus, the last case is when there is a single non-trivial SCC, where we then have to find feedback-arcs. For that, we use the linear-time algorithm of Garey and Tarjan~\citep{garey1978linear} for finding feedback vertices, and use a standard linear-time reduction from the feedback problem on arcs to the feedback problem on vertices described in, e.g., Even et al.~\citep{Even98}). We briefly describe how the reduction works. Subdivide each arc $uv$ of $G$ into two arcs $uw$ and $wv$, thus obtaining a graph $G'$ with $|V(G)|+|E(G)|$ vertices and $2|E(G)|$ edges. If an arc $uv$ is a feedback arc of $G$ then $w$ is a feedback vertex of $G'$ (deleting $w$ from $G'$ corresponds to deleting the arcs $uw$ and $wv$ in $G$), and the converse also holds. Notice, however, that $G'$ has feedback vertices that do not correspond to arcs of $G$, but those can be safely ignored.
\end{itemize}

The states $\reachesst{\mu_i}{\mu},\reachests{\mu_i}{\mu}$ get updated for $\iink$ as in Phase 1.

\begin{description} 
    \item[$\reachesst{\mu_i}{\nu}$, $\reachests{\mu_i}{\nu}$:] At this point $\acyclic{\mu_i}{\nu}$ is known. If $\acyclic{\mu_i}{\nu}$ is $\false$ or $\Null$ then $\reachesst{\mu_i}{\nu}$ and $\reachests{\mu_i}{\nu}$ are $\Null$ by definition. Otherwise $\acyclic{\mu_i}{\nu}$ is $\true$, and thus so is $\noextremity{\mu_i}{\nu}$ by definition. Therefore, $X_i$ is acyclic, $V(X_i) \setminus \{s_i,t_i\}$ has no cutvertex, no sink and no source of $G$, and there is no edge between a vertex not in $V(X_i)$ and a vertex in $V(X_i) \setminus \{s_i,t_i\}$ since $\{s_i,t_i\}$ is a separation pair; moreover $X_i$ is clearly connected. We can thus apply \Cref{lem:unique-poles-from-acyclicity} to conclude that one vertex between $s_i$ and $t_i$ is a source of $X_i$ and the other is a sink. In the former case we can set $\reachesst{\mu_i}{\nu}$ to $\true$ and $\reachests{\mu_i}{\nu}$ to $\false$, and in the latter case we can set $\reachests{\mu_i}{\nu}$ to $\true$ and $\reachesst{\mu_i}{\nu}$ to $\false$.
\end{description}

Notice that each node $\nu$ of $T$ is processed exactly once during this phase by BFS properties. Moreover, this also implies that every state pointing from a node to its parent gets updated, as desired.
As the work done in $\nu$ is linear in the size of $\dirskel(\nu)$ (see \Cref{alg:phase2}), with \Cref{lem:spqr-total-size} we can conclude that Phase 2 runs in time $O(|V(H)|+|E(H)|)$.

\begin{lemma}\label{lem:phases-correct}
    \Cref{alg:phase1} and \Cref{alg:phase2} correctly compute the states $\state{\nu}{\mu}$ and $\state{\mu}{\nu}$ for every edge $\{\nu,\mu\}$ of $T$ and run in time $O(|V(H)|+|E(H)|)$ where $H$ is the input graph.
\end{lemma}

\begin{algorithm}[t!]
\footnotesize
\caption{Superbubble finding -- Phase 2}
\label{alg:phase2}
\KwIn{Directed graph $G$, SPQR tree $T$ having at least two nodes}

$\rho \gets$ root of $T$, $\mathsf{Q} \gets \mathsf{queue()}$, $\mathsf{Q.push}(\rho)$\;

\While{$\mathsf{Q}$ is not empty}{
    $\nu \gets \mathsf{Q.pop()}$\;
    \lIf{$\nu$ has no children in $T$}{
        \textbf{continue}
    }

    Let $\mu_1,\dots,\mu_k$ be the children of $\nu$ with pertaining virtual edges $\{s_i,t_i\} = e_i$ $\in E(\skel(\nu))$, and $\mu_0$ the parent of $\nu$ with pertaining virtual edge $e_0 \in E(\skel(\nu))$\;
    \tcp{$\mu_0$ can be ignored if $\nu=\rho$}
    $\mathsf{Q.push}(\mu_i)$ \hspace{0.5em} $\forall i \in [1,k]$\;
    \textit{AllNodeExtremities} $\gets $ a set containing all the extremities of $G$ in $V(\skel(\nu))$\;
    \textit{AllEdgeExtremities} $\gets$ a set containing all the virtual edges $e_i \in E(\skel(\nu))$ such that $\noextremity{\nu}{\mu_i}$ is $\false$\;
    \For{$i \in [1,k]$}{
        \textit{noExtr} $\gets \true$ iff \textit{AllNodeExtremities} $\setminus \{s_i,t_i\} = \emptyset$\;
        \tcp{Notice that it suffices to store (up to) three extremities of $V(\skel(\mu))$ in order to update \textit{noExtr} in constant time}
        $\noextremity{\mu_i}{\nu} \gets $ (\textit{AllEdgeExtremities} $ \setminus \{\{s_i,t_i\}\} = \emptyset$) $\land$ \textit{noExtr}\;
        \tcp{Similarly, it suffices to store (up to) two virtual edges of $\skel(\nu)$ with corresponding extremity states leaving $\nu$ set to $\false$}
        \If{$\noextremity{\mu_i}{\nu}$ is $\false$}{
            $\acyclic{\mu_i}{\nu} \gets \Null$\;
        }
    }
    \If{at least two states among $\{ \acyclic{\nu}{\mu_1}, \dots ,\acyclic{\nu}{\mu_k} \}$ evaluate to $\Null$}{
        $\acyclic{\mu_i}{\nu} \gets \Null \; \forall \iink$\;
    }
    \If{exactly one state among $\{ \acyclic{\nu}{\mu_1}, \dots ,\acyclic{\nu}{\mu_k} \}$ evaluates to $\Null$}{
        Let $\jink$ be such that $\acyclic{\nu}{\mu_j}=\Null$\;
        $Y \gets \{1,\dots,k\} \setminus \{j\}$\;
        $\acyclic{\mu_i}{\nu} \gets \Null, \; \forall i\in Y$\;
        \textit{AcyclicOutside} $\gets$ true iff $\bigwedge_{i\in Y \cup \{0\}} \acyclic{\nu}{\mu_i}$ is $\true$\;
        \lIf{$\neg$ AcyclicOutside}{
            $\acyclic{\mu_j}{\nu} \gets \false$
        }
        \Else{
            \tcp{We are in conditions to build $\dirskel(\nu) - s_jt_j - t_js_j$}
            $K \gets \dirskel(\nu) - s_jt_j - t_js_j$\;
            $\acyclic{\mu_j}{\nu} \gets \true$ iff $K$ is acyclic\;
        }
    }
    \If{no state among $\{ \acyclic{\nu}{\mu_1}, \dots ,\acyclic{\nu}{\mu_k} \}$ evaluate to $\Null$}{

        \If{at least two states among $\{ \acyclic{\nu}{\mu_1}, \dots ,\acyclic{\nu}{\mu_k} \}$ evaluate to $\false$}{
            $\acyclic{\mu_i}{\nu} \gets \false \; \forall \iink$\;
        }
        \If{exactly one state among $\{ \acyclic{\nu}{\mu_1}, \dots ,\acyclic{\nu}{\mu_k} \}$ evaluates to $\false$}{
            Let $\jink$ be such that $\acyclic{\nu}{\mu_j}=\false$\;
            $Y \gets \{1,\dots,k\} \setminus \{j\}$\;
            $\acyclic{\mu_i}{\nu} \gets \false, \; \forall i\in Y$\;
            \tcp{We are in conditions to build $\dirskel(\nu) - s_jt_j - t_js_j$}
            $K \gets \dirskel(\nu) - s_jt_j - t_js_j$\;
            $\acyclic{\mu_j}{\nu} \gets \true$ iff $K$ is acyclic\;
        }
        \If{every state among $\{ \acyclic{\nu}{\mu_1}, \dots ,\acyclic{\nu}{\mu_k} \}$ evaluates to $\true$}{
            \tcp{We are in conditions to build $\dirskel(\mu)$}
            $A \gets \mathsf{FeedbackArcs(\dirskel(\nu))} \cap \{e_1,\dots,e_k\}$\;
            \tcp{$A$ contains those virtual edges of $\dirskel(\nu)$ which are feedback arcs in $\dirskel(\nu)$}
            $\acyclic{\mu_i}{\nu} \gets \true  \; \forall e_i \in A$\;
            $\acyclic{\mu_i}{\nu} \gets \false \; \forall e_i \notin A$\;
        }
    }
}
\end{algorithm}

\subsection{Algorithm - Phase 3}

In Phase 3 the pairs $(s,t)$, $(t,s)$ such that $\{s,t\}$ is a separation pair of $H$ are reported. (Recall that these are necessarily endpoints of virtual edges due to \Cref{lem:spqr-tree-contains-split-pairs} and \Cref{prop:nonadjacent-S-node}).
Further, if a pair of vertices are adjacent in the skeleton of an S-node and these identify a superbubble then the corresponding superbubble graph is not within the S-node, as we show next.

\begin{proposition}[Superbubbles and S-nodes]
\label{prop:superbubbles-in-S-nodes}
    Let $G$ be a directed graph, let $(s,t)$ be a superbubble of $G$ with graph $B$, let $T$ be the SPQR tree of a maximal 2-connected subgraph of $G$, and let $e=\{s,t\}$ be a virtual edge of a node of $T$. If the pertaining node of $e$ is an S-node then $B \not\subseteq \expansion(e)$.
\end{proposition}
\begin{proof}
    Suppose for a contradiction that $B \subseteq \expansion(e)$.
    By definition of S-node, the graph $\expansion(e)$ is a split component of the split pair $\{s,t\}$ and moreover contains a vertex $y$ separating $s$ from $t$, so $y$ is an $s$-$t$ cutvertex in $B$ since $B \subseteq \expansion(e)$. 
    The result now follows from \Cref{lem:superbubbloid-superbubble-cutvertex}.
\end{proof}

Thus, if $(s,t)$ is a superbubble and $\{s,t\}$ is a separation pair of $H$, then there is a P-node of $T$ with vertex-set $\{s,t\}$, or there is an R-node of $T$ with a virtual edge $\{s,t\}$. We discuss informally the two cases.

SPQR trees encode not only every separation pair of the graph but also the respective sets of split components. The way in which these split components are put together to form the skeletons of the nodes is what defines the different types of nodes, S, P, and R, as well as the SPQR tree itself. For the application of finding superbubbles, examining only the natural separations (and split components thereof) encoded in the SPQR tree is not enough to ensure completeness. Consider, for instance, a P-node $\mu$ with $k \geq 4$ split components. The separations encoded in each of the $k$ tree-edges incident to $\mu$ implicitly group $k-1$ expansions of the edges of $\skel(\mu)$ and puts the vertices therein in one side of the separation, and the vertices on the expansion of remaining virtual edge is put on the remaining side of the separation. However, it may be that the graph of a superbubble could match, e.g., the union of the expansions of two virtual edges of $\skel(\mu)$.
To overcome this, we first iterate over all the P-nodes of $T$ (say, with vertex set $\{s,t\}$) and group the virtual edges containing out-neighbours of $s$ and group the virtual edges containing in-neighbours of $t$; these have to match, otherwise $(s,t)$ is not a superbubble. Further, for each virtual edge in these (matching) sets, we check if the respective state leaving this P-node has the acyclicity and absence-of-extremities fields set to true. Finally, if all the out-neighbors (resp. in-neighbors) of $s$ (resp. $t$) are contained in candidate superbubble graph given by the matching sets, and $ts \notin E(G)$, then $(s,t)$ is a superbubbloid by \Cref{lem:separation-pair-is-superbubbloid}. Minimality follows from the structure of P-nodes as follows. If the resulting matching sets have more than one edge then minimality follows from \Cref{lem:superbubbloid-superbubble-cutvertex}. Otherwise, if the pertaining node of the unique edge in the set is an S-node then \Cref{prop:superbubbles-in-S-nodes} tells us that the expansion of that edge is not a superbubble, and so in this case S-nodes can be ignored; hence the pertaining node in question is an R-node, in which case minimality follows easily due to the connectivity of R-nodes.
We formalize this discussion with the next two results.

\begin{proposition}[Superbubbloids and P-nodes, see \Cref{fig:superbubbloid-pnode}]\label{prop:P-node}
    Let $G$ be a directed graph, $H$ be a maximal 2-connected subgraph of $G$, and $\mu$ be a P-node of the SPQR tree of $H$.
    Let $e_1,\dots,e_k$ denote the edges of $\skel(\mu))$ with endpoints $\{s,t\}$ $(k\geq3)$.  
    Let $E^+_s = \{ e_i : V(\expansion(e_i))\cap N^+(s) \neq \emptyset \}$, $E^-_t = \{ e_i : V(\expansion(e_i))\cap N^-(t) \neq \emptyset \}$, and $K=\bigcup_{e\in E^+_s}\expansion(e)$.
    Then $(s,t)$ identifies a superbubbloid of $G$ with graph $K$ if and only if $E^+_s \neq \emptyset$, $E^+_s = E^-_t$, $N^+_G(s) \subseteq V(K)$, $N^-_G(t) \subseteq V(K)$, $ts \notin E(G)$, and for each $e \in E^+_s$ the graph $\expansion(e)$ is acyclic and does not contain extremities of $G$ except $\{s,t\}$.
\end{proposition}
\begin{proof}
    $(\Rightarrow)$ Let $(s,t)$ be a superbubbloid of $G$ with graph $B$ and let $\mu$ be a P-node whose skeleton has vertex set $\{s,t\}$. Since superbubbles are contained in the blocks of $G$ by \Cref{lem:bubbles-cutvertices} and $s,t \in V(H)$ it follows that $V(B) \subseteq V(H)$. Further, since $B$ contains all the out-neighbors of $s$ and $V(B) \subseteq V(H)$, it follows that $E^+_s \neq \emptyset$ (analogously, $E^-_t \neq \emptyset$).
    We show that $K = B$.
    
    We show that $K \subseteq B$. Notice that $K$ is an induced subgraph since each expansion is induced and there are no edges across expansions. Since $B$ is also induced by definition of superbubbloid, it is enough to show that any vertex in $K$ is also in $B$.
    Notice that if $u \in V(K)$ then $u\in \expansion(e)$ for some $e\in E^+_s$. As established in the proof of (1) of \Cref{lem:separation-pair-is-superbubbloid}, $\expansion(e)$ has a directed path from $s$ to $t$ through $u$ since it is acyclic and has no extremities except $\{s,t\}$, and since $(s,t)$ is a superbubbloid, we have $u \in B$.
    Now we show that $B \subseteq K$. Suppose for a contradiction that $B \not\subseteq K$. Since $K$ and $B$ are induced subgraphs there is a vertex $v \in V(B)\setminus V(K)$. So in particular, $s$ reaches $v$ without $t$ via some directed path. 
    Due to the structure of P-nodes, this path is contained in $\expansion(e)$ for some $e\in E(\skel(\mu))$. Thus, the first vertex following  $s$ in this path is also in $\expansion(e)$ and hence $\expansion(e)$ has an out-neighbour of $s$. Therefore $e\in E^+_s$ and hence $v \in V(K)$, a contradiction.

    The conditions $N^+_G(s) \subseteq V(K)$ and $N^-_G(t) \subseteq V(K)$ follow trivially since $B=K$, and $ts \notin E(G)$ because $(s,t)$ is a superbubbloid. 
    Further, for each $e \in E^+_s$ the graph $\expansion(e)$ is acyclic and does not contain extremities of $G$ except $\{s,t\}$, since a cycle or extremity except $\{s,t\}$ in some expansion would be a cycle or extremity in $K$, each contradicting the fact that $K$ is a superbubbloid graph.
    The equality $E^+_s = E^-_t$ follows at once from the fact that $B=K=\bigcup_{e\in E^+_s}\expansion(e)$ and $N^+_G(s) \subseteq V(K)$ and $N^-_G(t) \subseteq V(K)$.

    $(\Leftarrow)$ First, notice that if $\{s,t\}$ is not a separation pair then $\skel(\mu)$ has exactly three edges, two of which are the real edges $st$ and $ts$ (recall that $G$ has no parallel edges and that $|E(\skel(\mu))| \geq 3$ by definition), a contradiction to the assumption that $ts \notin E(G)$.
    
    So $\{s,t\}$ is a separation pair and $K$ consists of a union of a nonempty subset of split components of $\{s,t\}$ since $E^+_s\neq\emptyset$.
    Moreover, $K$ has no extremities of $G$ except $\{s,t\}$ because $\expansion(e)$ has no extremities of $G$ except $\{s,t\}$ for every $e\in E^+_s$ .
    Further, since $\expansion(e)$ is acyclic and has no sources or sinks except $\{s,t\}$ for all $e\in E^+_s$, it follows that one vertex in $\{s,t\}$ is the unique source and the other is the unique sink of $\expansion(e)$ (see \Cref{lem:unique-poles-from-acyclicity}); due to the neighborhood constraints, it follows that $s$ is the source and $t$ is the sink of $\expansion(e)$.
    Also, since $\expansion(e)$ is acyclic for each $e\in E^+_s$, any cycle in $K$ contains vertices from different split components. So a cycle in $K$ contains both $s$ and $t$, but since $s$ is a source (and $t$ is a sink) in $\expansion(e)$ for any $e\in E^+_s$, it follows that $K$ is acyclic.
    So we are in conditions of applying \Cref{lem:separation-pair-is-superbubbloid} and conclude that $(s,t)$ is a superbubbloid of $G$ with graph $K$.
\end{proof}

\begin{figure}[t]
    \centering
    \includegraphics[width=0.75\linewidth]{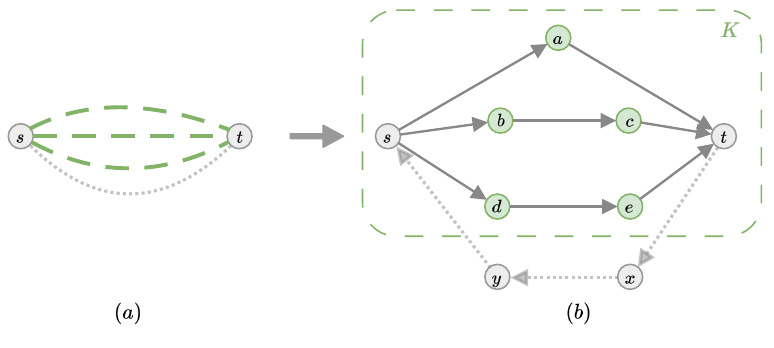}
    \caption{Superbubbloids in a P-node (\Cref{prop:P-node}). (a)~P-node skeleton with four parallel split components between $s$ and $t$. The green edges are those whose expansion contains out-neighbors of $s$ and in-neighbors of $t$ ($E^+_s = E^-_t$). (b)~ The three green paths go from $s$ to $t$ and their union $K$ forms the superbubbloid graph. The gray path goes from $t$ to $s$ and is not part of $K$.}
    \label{fig:superbubbloid-pnode}
\end{figure}

\begin{proposition}[Superbubbles and R-nodes]\label{prop:R-node}
    Let $G$ be a directed graph, $H$ be a maximal 2-connected subgraph of $G$, and $\nu,\mu$ be nodes of the SPQR tree of $H$.
    Let $e_\mu=\{s,t\} \in E(\skel(\nu))$ be the virtual edge pertaining to $\mu$ and $e_\nu \in E(\skel(\mu))$ be the virtual edge pertaining to $\nu$.
    If $\nu$ is an R-node, $N^+_G(s),N^-_G(t) \subseteq V(\expansion(e_\nu))$, $\expansion(e_\nu)$ is acyclic and has no extremities except $\{s,t\}$, $ts \notin E(G)$, then $(s,t)$ is a superbubble with graph $\expansion(e_\nu)$.
\end{proposition}
\begin{proof}
    Let $K=\expansion(e_\nu)$.
    Notice that $\{s,t\}$ is a separation pair of $H$ and that $K$ is a split component with respect to $\{s,t\}$.
    So we are in conditions of applying \Cref{lem:separation-pair-is-superbubbloid}, which implies that $(s,t)$ is a superbubbloid with graph $K$. Next we argue on the minimality.
    
    Notice that $\skel(\nu)$ has three internally vertex-disjoint $s$-$t$ paths since it is 3-connected. Then $\skel(\nu)$ without the edge $\{s,t\}$ has two internally vertex-disjoint $s$-$t$ undirected paths and hence so does $K$ (split components are connected, so a path through the edges of $\skel(\nu)$ can be mapped to a path in $K$). Therefore $(s,t)$ is a superbubble by \Cref{lem:superbubbloid-superbubble-cutvertex}.
\end{proof}

\begin{algorithm}[t!]
\small
\caption{Superbubble finding -- Phase 3}
\label{alg:phase3}
\KwIn{Directed graph $G$, SPQR tree $T$ having at least two nodes}

\For{every P-node $\mu$ of $T$}{
    Build the sets $E^+_s, E^-_t$ of $\mu$ as described in \Cref{prop:P-node}\;
    Build the sets $E^-_s, E^+_t$ analogously\;
    \If{$E^+_s = E^-_t$}{
        \tcp{Equivalently, $E^-_s = E^+_t$}
        Let $\{s,t\}$ be the vertex set of $\skel(\mu)$\;
        Let $e_1,\dots,e_k$ denote the edges in $\skel(\mu)$ $(k\ge3)$\;
        Let $\mu_1,\dots,\mu_\ell$ denote the pertaining nodes of edges in $E^+_s$ $(\ell\ge0)$\;
        Let $\mu'_1,\dots,\mu'_{\ell'}$ denote the pertaining nodes of edges in $E^-_s$ $(\ell'\ge0)$\;
        $\mathsf{assert}(\ell' = k - \ell)$\;
        
        \If{$\acyclic{\mu}{\mu_i}$ and $\noextremity{\mu}{\mu_i}$ are $\true$ for all $i \in \{1,\dots,\ell\}$}{
            \If{$N^+_G(s), N^-_G(t) \subseteq V(\bigcup_{e \in E^+_s} \expansion(e))$ and $ts \notin E(G)$}{
                \If{$\ell = 1$}{
                    \tcp{See \Cref{prop:superbubbles-in-S-nodes}}
                    Report $(s,t)$ if the pertaining node of the edge in $E^+_s$ is not an S-node\; \label{line:superbubbles-P-st-1}
                }
                \Else{
                    Report $(s,t)$\; \label{line:superbubbles-P-st-2}
                }
            }
        }

        \If{$\acyclic{\mu}{\mu'_i}$ and $\noextremity{\mu}{\mu'_i}$ are $\true$ for all $i \in \{1,\dots,\ell'\}$}{
            \If{$N^-_G(s), N^+_G(t) \subseteq V(\bigcup_{e \in E^-_s} \expansion(e))$ and $st \notin E(G)$}{
                \If{$\ell' = 1$}{
                    \tcp{See \Cref{prop:superbubbles-in-S-nodes}}
                    Report $(t,s)$ if the pertaining node of the edge in $E^-_s$ is not an S-node\; \label{line:superbubbles-P-ts-1}
                }
                \Else{
                    Report $(t,s)$\; \label{line:superbubbles-P-ts-2}
                }
            }
        }
    }
}

\For{every R-node $\mu$ of $T$}{
    \For{every neighbour $\nu$ of $\mu$ in $T$ that is not a P-node}{
        Let $\{s,t\}=e_\mu \in \skel(\nu)$ be the virtual edge pertaining to $\mu$\;
        Let $X = \expansion(e_\mu)$\;
        \If{$\acyclic{\nu}{\mu}$ and $\noextremity{\nu}{\mu}$ are $\true$}{
            \If{$N^+_G(s), N^-_G(t) \subseteq V(X)$ and $ts \notin E(G)$}{
                \tcp{And hence $N^-_G(s), N^+_G(t) \subseteq \overline{V(X)}$}
                Report $(s,t)$\; \label{line:superbubbles-R-st}
            }
            \If{$N^+_G(t), N^-_G(s) \subseteq V(X)$ and $st \notin E(G)$}{
                \tcp{And hence $N^-_G(t), N^+_G(s) \subseteq \overline{V(X)}$}
                Report $(t,s)$\; \label{line:superbubbles-R-ts}
            }
        }
    }
}
\end{algorithm}

\begin{algorithm}[t]
\small
\caption{Superbubble finding algorithm}
\label{alg:superbubbles-main}
\KwIn{Directed graph $G$}

Let $\mathcal{B}$ and $C \subseteq V(G)$ be the list of blocks and cutvertices of $U(G)$, respectively $(k \ge 1)$\;

\For{$H \in \mathcal{B}$}{
    \For{$e=\{s,t\} \in E(H)$}{
    \If{$N^+_s(G) = \{t\}$ and $N^-_t(G) = \{s\}$ and $ts\notin E(G)$}{
        Report $(s,t)$\; \label{line:superbubbles-trivial-1}
        }
    \If{$N^+_t(G) = \{s\}$ and $N^-_s(G) = \{t\}$ and $st\notin E(G)$}{
        Report $(t,s)$\;
        \label{line:superbubbles-trivial-2}
        }
    }
   \If{$H$ is a multi-bridge}{
       \textbf{continue}
   }\Else{
       \tcp{$H$ is 2-connected}
       \If{$H$ has exactly one source $s$ and one sink $t$ w.r.t.\ $H$ and $ts \notin E(G)$ \label{line:superbubbles-rn1}}{
           \If{$C \cap V(H)\setminus\{s,t\} = \emptyset$ and $N^+_G(s), N^-_G(t) \subseteq V(H)$ and $H$ is acyclic \label{line:superbubbles-rn2}}{
               Report $(s,t)$\; \label{line:superbubbles-block}
           }
       }
       $T \gets \mathsf{BuildSPQR}(H)$\;
       $\mathsf{Phase1}(G,T)$\;
       $\mathsf{Phase2}(G,T)$\;
       $\mathsf{Phase3}(G,T)$\;
   }
}
\end{algorithm}

\subsection{The superbubble finding algorithm}

We are in conditions to prove the correctness of the superbubble finding algorithm directly.

\paragraph{Correctness and runtime}

\begin{theorem}\label{thm:superbubbles}
    Let $G$ be a directed graph. The algorithm computing superbubbles (\Cref{alg:superbubbles-main}) is correct, that is, it finds every superbubble of $G$ and only its superbubbles, and it can be implemented in time $O(|V(G)| + |E(G)|)$.
\end{theorem}
\begin{proof}
     
     
    \textbf{(Completeness.)}
    We argue that every superbubble $(s,t)$ of $G$ is reported by the algorithm. Let $B$ denote the superbubble graph of $(s,t)$ and $H$ the block containing $B$. 
    
    If $(s,t)$ is a trivial superbubble then $N^+_G(s)=\{t\}$, $N^-_G(t)=\{s\}$, and $ts\notin E(G)$ by definition. These are exactly the conditions tested in Line~\ref{line:superbubbles-trivial-1} and~\ref{line:superbubbles-trivial-2}, and so $(s,t)$ is reported by the algorithm.
    Otherwise, if $V(B)=V(H)$, then $(s,t)$ is reported by the algorithm in Line~\ref{line:superbubbles-block}: clearly, $B$ has at most one source $s$ and at most one sink $t$ of $G$, no vertex in the interior of $B$ is a cutvertex of $G$ by \Cref{lem:bubbles-cutvertices}, $B$ is acyclic, $ts\notin E(G)$, and the out-neighbors of $s$ and the in-neighbors of $t$ are all contained in $H$. These conditions altogether are enough to report the pair $(s,t)$.
    
    Otherwise $\{s,t\}$ is a separation pair of $H$ by \Cref{thm:bubbles-split-pairs} (so $H$ is a maximal 2-connected subgraph of $G$). 
    Let $T$ denote the SPQR tree of $H$. Since no pair of nonadjacent vertices in an S-node identifies a superbubble by \Cref{prop:nonadjacent-S-node}, it follows by \Cref{lem:spqr-tree-contains-split-pairs} that $\{s,t\}$ are endpoints of a virtual edge of a node $\mu$ of $T$. This virtual edge is associated with a tree edge $\{\nu,\mu\}$. Let $e_\mu$ be the virtual edge in $\nu$ pertaining to $\mu$ and let $e_\nu$ the virtual edge in $\mu$ pertaining to $\nu$.

    If $\mu$ is an S-node then \Cref{prop:superbubbles-in-S-nodes} implies that $B \not\subseteq \expansion(e_\mu)$.
    (Essentially $\state{\nu}{\mu}$ can be ignored). Symmetrically, $\state{\mu}{\nu}$ can be ignored whenever $\nu$ is an S-node.
    If $\mu$ is a P-node with vertex set $\{s,t\}$ then $B$ can be expressed as the union of the expansions of the virtual edges of $\mu$ as described in \Cref{prop:P-node}. Symmetrically, the same is done whenever $\nu$ is a P-node.
    Hence, the remaining virtual edges that encode superbubbles are those contained in the R-nodes. Moreover, if the pertaining node of such a virtual edge is a P-node then $(s,t)$ is processed when analyzing P-nodes. So it suffices to analyze P-nodes individually followed by the tree-edges $\{\nu,\mu\}$ such that $\nu$ is not a P-node and $\mu$ is an R-node. We argue on the two cases separately.
    
    \begin{itemize}
        \item \textbf{$\mu$ is a P-node:}
        Let $e_1,\dots,e_k$ be the edges in $\skel(\mu)$ whose endpoints are $\{s,t\}$ $(k \geq 3)$.
        Since $(s,t)$ is a superbubble, $(s,t)$ is also a superbubbloid and thus \Cref{prop:P-node} implies that
        $B$ can be expressed, without loss of generality, as $\bigcup_{i=1}^{k'} \expansion(e_i)$ for some $k' < k$ ($k\neq k'$ since otherwise $V(B)=V(H)$); further, it implies that $\expansion(e_i)$ is acyclic and has no extremities except $\{s,t\}$ for each $i=1,\dots,k'$, $E^+_s=E^-_t$, $N^+_G(s), N^-_G(t) \subseteq V(B)$ and $ts \notin E(G)$.
        If $k' \neq 1$ then these conditions are enough to report $(s,t)$ as a superbubble (Line~\ref{line:superbubbles-P-st-2} or Line~\ref{line:superbubbles-P-ts-2}). 
        Otherwise we have $k'=1$. If $e_1$ is a real edge then it was reported when analyzing the trivial superbubbles (notice also that, in this case, the conditions given by \Cref{prop:P-node} match those characterizing a trivial superbubble). Otherwise $e_1$ is virtual and thus it has a pertaining node in $T$. Suppose for a contradiction that the pertaining node of $e_1$ is an S-node. Since $(s,t)$ is a superbubble and the out-neighbors of $s$ are all contained in $\expansion(e_1)$, it implies that $B \subseteq \expansion(e_1)$, from where \Cref{prop:superbubbles-in-S-nodes} gives a contradiction. Therefore the pertaining node of $e_1$ is not an S-node and $(s,t)$ is reported in Line~\ref{line:superbubbles-P-st-1} or Line~\ref{line:superbubbles-P-ts-1}.
        
        \item \textbf{$\mu$ is an R-node and $\nu$ is not a P-node:}
        We have that $s$ and $t$ are the endpoints of $e_\mu$ and $e_\nu$.
        Suppose that $s$ has out-neighbors both in $\expansion(e_\nu)$ and $\expansion(e_\mu)$.
        Then we claim that $V(H)=V(B)$, a contradiction to the fact that we are under the assumption $V(H)\neq V(B)$.
        First we show that $V(\expansion(e_\nu)) \subseteq V(B)$.

        Suppose for a contradiction that $V(\expansion(e_\nu)) \not\subseteq V(B)$. Let $x \in V(\expansion(e_\nu)) \setminus V(B)$. We claim that $U(\expansion(e_\nu))$ has an $s$-$x$ path $p$ avoiding $t$.
        Let $e_x=\{x',y'\}$ be an edge in $\skel(\nu)-e_\mu$ whose expansion contains $x$ with $x' \neq t$ (possibly $x'=s$ or $y'=t$, but not both equalities hold). First we argue that $\skel(\nu)-e_\mu$ has an $s$-$x'$ path avoiding $t$ and then we argue that $\expansion(e_x)$ has an $x'$-$x$ path avoiding $t$. The concatenation of these two paths produces an undirected $s$-$x$ walk in $U(\expansion(e_\nu))$ avoiding $t$, which can be simplified into the desired path.
        
        If $\nu$ is an S-node then $\skel(\nu)-e_\mu$ consists of a path between $s$ and $t$ with at least three vertices. Since $x' \neq t$ the graph $\skel(\nu)-e_\mu$ has an $s$-$x'$ path avoiding $t$.
        If $\nu$ is an R-node then $\skel(\nu)$ has three internally vertex-disjoint $s$-$x'$ paths at most one of which contains the edge $e_\nu$. Thus $\skel(\nu)-e_\mu$ has an $s$-$x'$ path avoiding $t$. Notice that this path can be mapped to a path in $\expansion(e_\nu)$ since split components are connected (while still avoiding $t$).
        Finally, applying \Cref{lem:reaches-in-expansion} gives an $x'$-$x$ undirected path in $\expansion(e_x)$ avoiding $y'$, so this path does not contain $t$.

        The undirected path $p$ starts in a vertex contained in $B$ and ends in a vertex not contained in $B$.
        Let $a$ denote the last vertex in $p$ that is contained in $B$ (such a vertex exists since $a=s$ at the earliest). Then $a$ has a successor in $p$, say $b$, which is not contained in $B$.
        Thus $U(\expansion(e_\nu))$ has an edge $\{a,b\}$ and hence $\expansion(e_\nu)$ has an edge $ab$ or $ba$.
        Since $B$ is a superbubble graph and $a \in V(B)$, $H$ has an $s$-$a$ directed path avoiding $t$ and an $a$-$t$ directed path avoiding $s$. 
        So if $ab \in E(\expansion(e_\nu))$ then $\expansion(e_\nu)$ has an $s$-$b$ directed path avoiding $t$ and thus $b \in V(B)$, a contradiction, and if $ba \in E(\expansion(e_\nu))$ then $\expansion(e_\nu)$ has a $b$-$t$ path avoiding $s$ and thus $b \in V(B)$, a contradiction. Therefore $V(\expansion(e_\nu)) \subseteq V(B)$.
        
        Symmetrically we have $V(\expansion(e_\mu)) \subseteq V(B)$: we can apply the argument described above for the case when $\nu$ is an R-node since $\mu$ is an R-node.
        Since $V(\expansion(e_\mu)) \cup V(\expansion(e_\nu)) = V(H)$ we have $V(H) \subseteq V(B)$, and since superbubbles live within blocks we get $V(B)=V(H)$, as desired.

        So $s$ has out-neighbors only in one expansion between $\nu$ and $\mu$ and therefore the superbubble is contained in that expansion. By symmetry, $t$ has in-neighbors in only one expansion, and it is not hard to see that these expansions have to match.
        If the out-neighbors of $s$ are contained in $\expansion(e_\nu)$ and $\nu$ is an S-node then \Cref{prop:superbubbles-in-S-nodes} gives a contradiction, so $\nu$ is an R-node. Further, we have that $B$ is acyclic, has no extremities of $G$ except $\{s,t\}$, and $ts \notin E(G)$, since $(s,t)$ is a superbubble. These conditions altogether are enough to report $(s,t)$ in Line~\ref{line:superbubbles-R-st} (when iterating over node $\mu$ if $B \subseteq \expansion(e_\mu)$ and over node $\nu$ if $B \subseteq \expansion(e_\nu)$). 
    \end{itemize}

    \textbf{(Soundness.)}
    Let $(s,t)$ be a pair of vertices reported by the algorithm. We show that $(s,t)$ is a superbubble of $G$.

    If the pair $(s,t)$ is reported in Line~\ref{line:superbubbles-trivial-1} or~\ref{line:superbubbles-trivial-2} then $(s,t)$ is a trivial superbubble by definition.
    If the pair $(s,t)$ is reported by virtue of Line~\ref{line:superbubbles-block} then $H$ has exactly one source $s$ and exactly one sink $t$ (with respect to $H$), no vertex in $H$ except $\{s,t\}$ is a cutvertex of $G$, $H$ is acyclic, $N^+_G(s),N^-_G(t) \subseteq V(H)$, and $ts\notin E(G)$. It is not hard to see that under these conditions the pair $(s,t)$ is a superbubbloid with graph $H$ (a similar proof to that of \Cref{lem:separation-pair-is-superbubbloid} is possible and we omit it for the sake of brevity).
    Since $H$ is 2-connected it has two internally vertex-disjoint $s$-$t$ undirected paths, so \Cref{lem:superbubbloid-superbubble-cutvertex} implies that $(s,t)$ is a superbubble.
    

    Now we discuss the case when $\{s,t\}$ is a separation pair of a block $H$ of $G$. By symmetry it suffices to show that the pairs reported in Lines~\ref{line:superbubbles-P-st-1},~\ref{line:superbubbles-P-st-2}, and~\ref{line:superbubbles-R-st} are superbubbles.
    If $(s,t)$ is reported in Line~\ref{line:superbubbles-P-st-2} then \Cref{prop:P-node} implies that $(s,t)$ is a superbubbloid with graph $K$; moreover, since $K=\bigcup_{e\in E^+_s}\expansion(e)$ consists of the union of $\ell \geq 2$ split components of $\{s,t\}$, $K$ has two internally vertex-disjoint $s$-$t$ undirected paths, and hence \Cref{lem:superbubbloid-superbubble-cutvertex} gives that $(s,t)$ is a superbubble.
    If $(s,t)$ is reported in Line~\ref{line:superbubbles-R-st} then the fact that $(s,t)$ is a superbubble follows at once by \Cref{prop:R-node}.
    If $(s,t)$ is reported in Line~\ref{line:superbubbles-P-st-1} then the pertaining node of the unique edge in $E^+_s$ is an R-node (as no two P-nodes are adjacent in $T$), so we are conditions of applying \Cref{prop:R-node} and conclude that $(s,t)$ is a superbubble.

    \textbf{(Running time.)} Block-cut trees can be built in linear time \citep{Hopcroft73blockcut} and the total size of the blocks is linear in $|V(G)|+|E(G)|$. The case when a block is a multi-bridge is trivial, so suppose that we are analyzing a block $H$ that is 2-connected. Let $|H|:=|V(H)|+|E(H)|$. We show that the rest of the algorithm runs in time $O(|H|)$, thus proving the desired bound.
    
    The conditions on Lines~\ref{line:superbubbles-rn1} and~\ref{line:superbubbles-rn2} are trivial and require $O(|H|)$ time altogether.
    The SPQR tree $T$ can be built in $O(|H|)$ time~\citep{gutwenger2001linear}. Phases 1 and 2 take $O(|H|)$ time by \Cref{lem:phases-correct}.
    For Phase 3, recall first that $T$ has $O(|H|)$ P-nodes as well as tree-edges by \Cref{lem:spqr-total-size}. Further, notice that the work done in each P-node and in each tree-edge entering an R-node takes constant-time with exception of the inclusion-neighborhood queries of $s$ and $t$.
    To handle this type of queries, we can proceed as follows.

    In order to decide inclusion-neighborhood queries, e.g., of vertices $u$ and $v$, we process all edges of $T$ with a DFS traversal starting in the root. Let $\nu$ be the parent of $\mu$ in $T$ and let $\{u,v\}$ denote the endpoints of $e_\nu$ and $e_\mu$.
    We store at $u$ and $v$ the number of their out- and in-neighbors in $\expansion(e_\mu)$. Assume that we have already computed this information (via the DFS order) for all tree edges to children of $\mu$ in $T$ (if $\mu$ is not a leaf). For all such tree edges to children of $\mu$ in which $u$ is present, we increment the respective counts of $u$ by these values, and same for $v$. Moreover, we scan every real edge in $\skel(\mu)$ and use the neighborhoods induced by the edge to correspondingly increment the respective counts for $u$ and $v$. Doing this, we process every real edge once because every edge of the input graph is a real edge in exactly one skeleton. Having the correct out- and in-neighborhood counts for $u$ and $v$ in $\expansion(e_\mu)$, we can obtain their counts in $\expansion(e_\nu)$ by subtracting from the total number of out-neighbors of $u$ the value of out-neighbor counter of $u$ in $\expansion(e_\mu)$ (and same for $v$). This can again be obtained by paying only constant time per edge.
    
    To conclude, for each P-node $\mu$ the algorithm spends $O(|E(\skel(\mu))|)$ time to build the sets described in \Cref{prop:P-node}, and for tree-edges entering R-nodes the algorithm spends a constant amount of time. The latter thus requires $O(|V(H)|)$ time altogether because $T$ has $O(|V(H)|)$ R-nodes at most, and the former requires $O(|H|)$ time altogether since the total number of edges in the skeletons of the nodes of $T$ is $O(|E(H)|)$ and $T$ has $O(|V(H)|)$ P-nodes (recall \Cref{lem:spqr-total-size}).
\end{proof}

\section{Snarls}
\label{sec:snarls}

\subsection{Setup}

We assume, without loss of generality, that $G=(V,E)$ is a \emph{connected} bidirected graph. To give our equivalent snarl characterization we introduce more terminology. The \emph{splitting} operation takes a bidirected graph $G=(V,E)$ and a vertex-side $u\alpha$ and produces a new bidirected graph $G'=(V',E')$ with $V' := V(G) \cup \{u'\}$ and $E' := E(G) \setminus \{\{u\hat{\alpha},v\beta\} \mid \{u\hat\alpha,v\beta\} \in E(G)\} \cup \{ \{u'\hat{\alpha},v\beta\} \mid \{u\hat\alpha,v\beta\} \in E(G) \}$. As a result, all edges incident to $u$ with sign $\hat{\alpha}$ will be incident to $u'$ instead.

\paragraph{Remark:} In the remainder of this section we discuss two equivalent definitions of ``snarls'' and for the sake of our results this section can be skipped, only \Cref{def:snarl} is required for the rest of the paper.\\



\cite{paten2018ultrabubbles} define snarls via the \emph{biedged graph}, an undirected graph $B(G) = (V_B, E_B)$ constructed from a bidirected graph $G = (V,E)$ as follows.
We first \emph{split} every vertex $v\in V$ into two nodes $v+$ and $v-$ (one per vertex-side), so that $V_B=\{v+,v- \mid v\in V\}$.
Then, for each $v\in V$, we add an undirected \emph{inner} edge $\{v+,v-\}$. Finally, for each bidirected edge $\{u\alpha, v\beta\}\in E$ (with $\alpha,\beta\in\signs$), we add an undirected \emph{outer} edge $\{u\alpha, v\beta\}$ between the corresponding split nodes. This construction is illustrated in \Cref{fig:bubbles}(b).
We call the endpoints of an inner edge \emph{opposites} and denote by $\hat{x}$ the opposite of a node $x\in V_B$.
If an inner edge has a parallel outer edge, we keep both edges (as distinct parallel edges).
Otherwise, we assume without loss of generality that there are no parallel outer edges, since they do not affect snarls.
Throughout, we write $u\alpha$ for a vertex-side (with $\alpha\in\signs$) and $u\hat{\alpha}$ for its opposite vertex-side.
In the biedged graph, Paten~et~al. define snarls as follows.

\begin{definition}[Snarls in biedged graphs]
    \label{def:snarl-biedged}
    An unordered pair of distinct, non-opposite nodes $\{a, b\}$ is a \emph{snarl} if
    \begin{enumerate}[nosep]
        \item[(a)] \emph{separable:} the removal of the inner edges incident with $a$ and $b$ (i.e., $\{a,\hat{a}\}$ and $\{b,\hat{b}\}$) disconnects the graph, creating a connected component $X$ that contains $a$ and $b$ but neither $\hat{a}$ nor $\hat{b}$.
        We call $X$ the \emph{snarl component} of $\{a, b\}$.
        \item[(b)] \emph{minimal:} no pair of opposites $\{z, \hat{z}\}$ in $X$ different from $a$ and $b$ exists such that $\{a, z\}$ and $\{b, \hat{z}\}$ are separable.
    \end{enumerate}
\end{definition}

To avoid using the biedged graph in our algorithm, we propose an equivalent definition of snarls in bidirected graphs.

\begin{definition}[Snarl, Snarl component]\label{def:snarl}
    A pair of vertex-sides $\{x\alpha, y\beta\}$ with $x \neq y$ is a \emph{snarl} in a bidirected graph $G$:
    \begin{enumerate}[nosep]
        \item[(a)] \emph{separability:}
        the graph created by splitting the vertex-sides $x\alpha$ and $y\beta$ in $G$ has a separate component $X$ containing $x$ and $y$ but not the vertices $x'$ and $y'$ created by the split operation.
        We call $X$ the \emph{snarl component} of $\{x\alpha, y\beta\}$.
        
        \item[(b)] \emph{minimality:}
        $X$ has no vertex-sides $z\gamma$ and $z\hat{\gamma}$ such that $\{x\alpha, z\gamma\}$ and $\{z\hat{\gamma}, y\beta\}$ are separable in $G$.
    \end{enumerate}
\end{definition}

The equivalence of these two definitions should be clear.

\begin{lemma}[Equivalence of snarl definitions]
    \label{lem:snarl-equivalence}
    Let $G$ be a bidirected graph and let $B(G)$ be its biedged graph. Let $\{u\alpha, v\beta\}$ be an unordered pair of vertex-sides of $G$ (equivalently, an unordered pair of nodes of $B(G)$), with $\alpha,\beta\in\signs$.  Then $\{u\alpha, v\beta\}$ is separable (resp.\ minimal, resp.\ a snarl) in $G$ by \Cref{def:snarl} if and only if it is separable (resp.\ minimal, resp.\ a snarl) in $B(G)$ by \Cref{def:snarl-biedged}.
\end{lemma}
\begin{proof} \textit{(Sketch)}
    Deleting the inner edge $\{u\alpha,u\hat{\alpha}\}$ in $B(G)$ separates the node $u\alpha$ from its opposite $u\hat{\alpha}$ while leaving all outer edges intact. This has the same effect on connectivity as splitting the vertex-side $u\alpha$ in $G$ (which detaches all edges incident to $u\hat{\alpha}$ from $u$ by moving them to the new vertex $u'$). Applying the same argument to $v\beta$ yields the equivalence of separability. The minimality clauses translate verbatim, since ``opposites'' in $G$ correspond exactly to the endpoints of an inner edge in $B(G)$.
\end{proof}

\subsection{Sign-cut graphs and dangling blocks}

Let $x$ be a cutvertex of $G$ and let $C_1,\dots,C_\ell$ be the components of $G-x$.
Then $x$ is \emph{sign-consistent} if 
$x$ is a tip in $G[V(C_i) + x]$ for $i \in \{1,\dots,\ell\}$.

\begin{definition}[Sign-cut graphs, see \Cref{fig:sign-cut-graphs}]
\label{def:sign-cut-graphs}
    Let $G$ be a bidirected graph and let $Y \subseteq V(G)$ be the set of sign-consistent vertices of $G$. The \emph{sign-cut graphs} of $G$ are the graphs resulting from splitting each vertex $y\in Y$ with any sign in $\{+,-\}$ and relabeling each new vertex $y'$ as $y$.
\end{definition}

\begin{figure}[t]

    \centering

    \includegraphics[width=\linewidth]{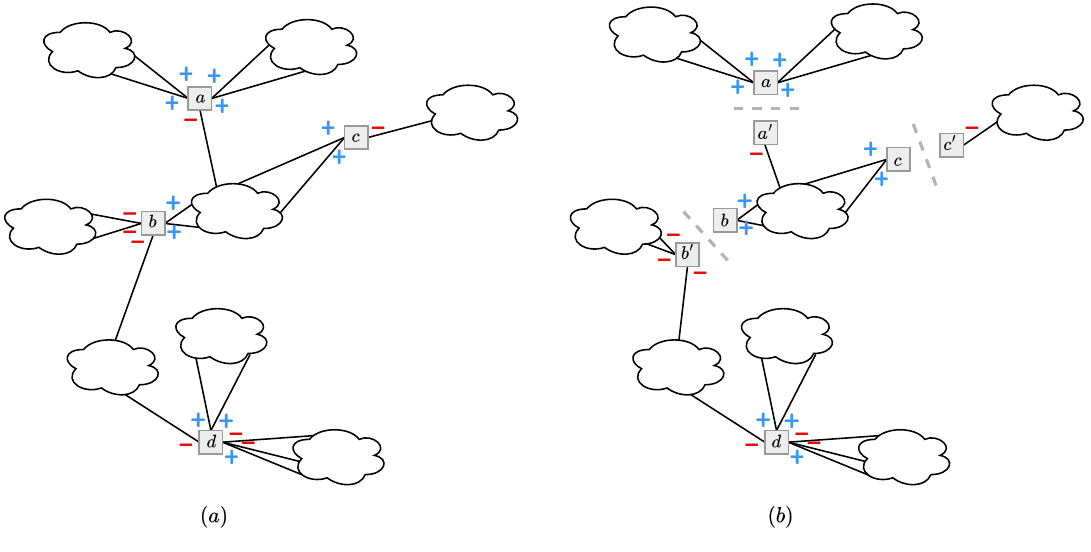}

    \caption{Illustration of sign-cut graphs (\Cref{def:sign-cut-graphs}). (a)~A bidirected graph $G$ with three sign-consistent cutvertices $a$, $b$, $c$ and one non-sign-consistent cutvertex $d$ (which has both $+$ and $-$ vertex-sides towards the same block). (b)~The sign-cut graphs of $G$ after splitting each sign-consistent vertex. Each split vertex appears as a tip in both copies. Note also that each pair of tips of a resulting a sign-cut graph of $G$ forms a snarl.}

    \label{fig:sign-cut-graphs}

\end{figure}

Notice that if $u\alpha$ and $v\beta$ with $u \neq v$ are vertex-sides, then splitting $u\alpha$ and then $v\beta$ yields the same graph as splitting $v\beta$ and then $u\alpha$, and so sign-cut graphs are well defined.

A vertex is contained in two sign-cut graphs if and only if it is sign-consistent, and moreover every sign-consistent vertex becomes a tip in both sign-cut graphs it appears in (one with its positive vertex-sides and the other with its negative vertex-sides).
Sign-cut graphs partition the edges (and thus the vertex-sides) of $G$ and the blocks of $G$ coincide with the blocks of its sign-cut graphs.
With this, we can already show a simple property of snarls.

\begin{lemma}\label{lem:snarls-inside-sign-cut-tree}
    Let $G$ be a bidirected graph, let $F_1$ and $F_2$ be distinct sign-cut graphs of $G$, and let $u\alpha$ be a vertex-side of $F_1$ and $v\beta$ a vertex-side of $F_2$. Then $\{u\alpha,v\beta\}$ is not a snarl of $G$.
\end{lemma}
\begin{proof}
    We can assume that $u \neq v$ for otherwise $\{u\alpha,v\beta\}$ is not a snarl by definition. There is a sign-consistent vertex $x \in V(G)$ that puts $u\alpha$ and $v\beta$ in distinct sign-cut graphs (possibly $x=u$ or $x=v$). Moreover, for some $\gamma\in\signs$, the edges of $G$ incident to $x$ containing a vertex-side $x\gamma$ are all in $F_1$ and those with a vertex-side $x\hat \gamma$ are all in $F_2$, since $x$ is sign-consistent.
    As such, splitting $x\gamma$ in $G$ results in a graph, say $G_{x\gamma}$, with two components: one containing $x$ and $u$ and the other containing $x'$ and $v$.
    
    Suppose that $x$ is distinct from $u$ and $v$ and suppose for a contradiction that $\{u\alpha,v\beta\}$ is a snarl with component $X$. We argue that $\{u\alpha, x\gamma\}$ is separable, and the fact that $\{v\beta, x\hat\gamma\}$ is separable follows symmetrically, thus contradicting the minimality of $\{u\alpha, v\beta\}$.
    Notice that splitting $u\alpha$ in $G_{x\gamma}$ does not separate $u$ and $x$, otherwise splitting $u\alpha$ in $G$ separates $u$ from $v$ because every $u$-$v$ path in $G$ contains $x$. Similarly, since $u$ and $v$ are in different components of $G_{x\gamma}$ (or $G-x$), splitting $u\alpha$ separates $u$ from $u'$. So $u$ and $x$ remain connected and become separated from $u'$ and $x'$.

    Suppose now that $x$ is equal to $u$ or $v$; say $x=u$ without loss of generality, so $x\gamma=u\alpha$. Then $G_{x\gamma}$ has one component containing vertex $u$ and another containing the vertices $u'$ and $v$.
    Further splitting $v\beta$ does not create paths between $u$ and $v$ (although it may separate $u'$ and $v$), so the pair $\{u\alpha, v\beta\}$ is not separable and so it is not a snarl.
\end{proof}

Cutvertices that are not sign-consistent require additional care, as we show in the next result.
Let $v$ be a vertex in a block $H$. If $H'$ is a block distinct from $H$ that contains $v$ and has vertex-sides of opposite signs at $v$, then $H'$ is a \emph{dangling block} of $v$ with respect to $H$. For instance, non cutvertices have no dangling blocks.

\begin{proposition}[Dangling blocks, see \Cref{fig:dangling-block}]
\label{prop:mixed-dangling-blocks}
    Let $G$ be a bidirected graph, $H$ be a block of $G$, and $u,v \in V(H)$ be vertices. If $u$ or $v$ has dangling blocks with respect to $H$ then $\{u\alpha,v\beta\}$ is not separable for any signs $\alpha,\beta \in \signs$.
\end{proposition}
\begin{proof}
    \sloppypar
    Without loss of generality suppose that $u$ has a dangling block $H'$ with respect to $H$. So $H'$ has edges $\{u+,x\gamma\}$ and $\{u-,y\delta\}$. Notice that $u$ is a cutvertex of $G$, since otherwise there is exactly one block containing $u$ (which is $H$).
    Since $H'$ is a block, it has an $x$-$y$ path avoiding $u$. Thus splitting $u\alpha$ results in a graph containing a $u$-$u'$ path whose internal vertices are contained in $H'$.
    Since $u,v \in V(H)$ and no two blocks contain the same two vertices, $v$ is not in $H'$ and so splitting $v\beta$ does not affect the path previously constructed.
    Therefore $u$ and $u'$ remain connected in the graph resulting from splitting $u\alpha$ and $v\beta$, in other words, $\{u\alpha,v\beta\}$ is not separable.
\end{proof}

\begin{figure}[t]
    \centering
    \includegraphics[width=0.8\linewidth]{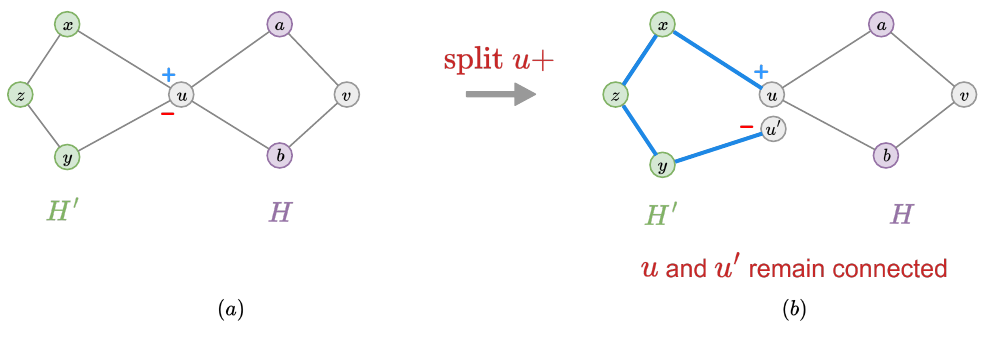}
    \caption{A dangling block $H'$ with respect to $u$ and $H$ (\Cref{prop:mixed-dangling-blocks}). (a)~Block $H'$ (green) has both $+$ and $-$ vertex-sides at $u$, while $H$ (purple) has only $+$. (b)~After splitting $u{+}$, the path through $H'$ (blue) connects $u$ to $u'$, so $\{u\alpha, v\beta\}$ is not separable.}
    \label{fig:dangling-block}
\end{figure}

Our goal is to show an equivalence between the snarls of $G$ and the snarls of its sign-cut graphs (\Cref{lem:snarls-G-F}), from where pinpointing all the snarls becomes easy (\Cref{thm:where-are-snarls-after-cutting}).

\begin{restatable}{lemma}{snarlsGF}
\label{lem:snarls-G-F}
    Let $G$ be a bidirected graph, let $\{u\alpha,v\beta\}$ be a pair of vertex-sides. Then $\{u\alpha,v\beta\}$ is a snarl of $G$ if and only if there is a sign-cut graph $F$ of $G$ such that $\{u\alpha,v\beta\}$ is a snarl of $F$.
\end{restatable}

\begin{restatable}{theorem}{snarlsaftercutting}
\label{thm:where-are-snarls-after-cutting}
    Let $G$ be a bidirected graph and let $F$ be a sign-cut graph of $G$.
    \begin{enumerate}[nosep]
        \item If $u$ and $v$ are distinct tips in $F$ with signs $\alpha,\beta \in\signs$, respectively, then $\{u\alpha,v\beta\}$ is a snarl of $G$.
        \item If $\{u\alpha,v\beta\}$ is a snarl of $G$ and $u$ and $v$ are non-tips in $F$ then there is a unique block of $F$ where $u$ and $v$ are non-tips and where $\{u,v\}$ is a split pair.
    \end{enumerate}
\end{restatable}

\subsection{Properties of snarls}

In this section we give a series of technical results on snarls in order to prove \Cref{lem:snarls-G-F} and \Cref{thm:where-are-snarls-after-cutting}.


\begin{proposition}
\label{prop:mixed-signs-inside-block}
    Let $G$ be a bidirected graph, $F$ be a sign-cut graph of $G$, and $u \in V(F)$ be a vertex. If $u$ is a non-tip in $F$ then there is a block of $F$ with vertex-sides of opposite signs in $u$.
\end{proposition}
\begin{proof}
    Suppose for a contradiction that, for every block of $F$, the vertex-sides of $u$ have the same sign.
    If $u$ is not a cutvertex in $F$ then $u$ is in a unique block of $F$ and thus $u$ is a tip, a contradiction.
    Otherwise $u$ is a sign-consistent non-tip cutvertex of $G$, a contradiction to the fact that $F$ is a sign-cut graph of $G$.
\end{proof}

Unlike superbubbles and cutvertices, we cannot argue that snarls do not contain sign-consistent vertices in their ``interior'' (e.g., the component of a snarl formed by two tips is the whole graph, and thus contains all sign-consistent vertices). Nonetheless, sign-cut graphs are useful because they give us a way to efficiently encode all the snarls (this will become clear later on when the snarl finding algorithm is given).

\begin{lemma}
\label{lem:snarls-both-tips-component}
    Let $G$ be a connected bidirected graph and let $\{u\alpha,v\beta\}$ be a snarl of $G$ with component $X$. Then $X=G$ if and only if $u$ and $v$ are tips in $G$ with signs $\alpha,\beta\in\signs$, respectively.
\end{lemma}
\begin{proof}
    $(\Leftarrow)$
    If $u$ and $v$ are tips in $G$ with signs $\alpha$ and $\beta$ then splitting $u\alpha$ and $v\beta$ results in a graph consisting of $G$ plus two isolated vertices, $u'$ and $v'$. Since $G$ is connected it follows that $X=G$.

    $(\Rightarrow)$
    We show that if $u$ or $v$ are non-tips in $G$ then $X \neq G$ (i.e., there is an edge or a vertex of $G$ that is not part of $X$, as $X \subseteq G$ by definition of snarl and snarl component).
    Suppose without loss of generality that $\alpha=+$.
    Since $u$ is a non-tip in $G$, $G$ has an edge $\{u-,w\gamma\}$. Let $G'$ denote the graph resulting from splitting $u+$ and $v\beta$.
    
    Suppose that $w \neq v$. If $w \in V(X)$ then $X$ has a $u$-$w$ path in $X$ because components are connected. This path can be extended with the edge $\{w\gamma,u'-\} \in E(G')$, and so there is a $u$-$u'$ path in $G'$ and we have $u' \in V(X)$, contradicting the fact that $\{u\alpha,v\beta\}$ is separable. Thus $w \notin V(X)$ and so $V(X) \neq V(G)$.

    Suppose that $w = v$. Then $\gamma=\hat{\beta}$, otherwise $\{u\alpha,v\beta\}$ is not separable since $u$ and $u'$ would be connected in $G'$ via $v$. Then the edge $\{u-,v\hat{\beta}\}$ of $G$ (which becomes the edge $\{u'-,v'\hat{\beta}\}$ in $G'$) is not contained in $X$ since $u',v'\notin V(X)$, and therefore $E(X) \neq E(G)$.
\end{proof}

\begin{proposition}\label{prop:tip-nontip-nonseparable}
    Let $G$ be a bidirected graph and $F$ be a sign-cut graph of $G$ with vertices $u$ and $v$. With respect to $F$, if $u$ is a tip with sign $\alpha$ and $v$ is a non-tip then $\{u\alpha,v\beta\}$ is not separable for any $\beta\in\signs$.
\end{proposition}
\begin{proof}
    Since $v$ is a non-tip in $F$ we can apply \Cref{prop:mixed-signs-inside-block} to get a block $H$ of $F$ containing edges $\{v+, z\gamma\}$ and $\{v-, w\delta\}$. So $H$ has a $z$-$w$ path $p$ avoiding $v$ (if $H$ is 2-connected this follows from 2-connectivity, and if $H$ is a multi-bridge then $z=w$ and the path is trivial).
    Thus, splitting $v\beta$ for $\beta\in\signs$ in $F$ results in a graph where $v$ and $v'$ are connected by $p$ and the two edges incident to $z$ and $w$ (where $v$ is appropriately changed to $v'$ in one of the edges).
    Splitting $u\alpha$ has no effect in this path as it only creates an isolated vertex $u'$, $u$ being a tip.
\end{proof}

We are now ready to prove the two desired results.

\snarlsGF*
\begin{proof}
    $(\Rightarrow, separability)$ Let $\{u\alpha,v\beta\}$ be a separable pair of vertex-sides of $G$. \Cref{lem:snarls-inside-sign-cut-tree} implies that the vertex-sides $u\alpha$ and $v\beta$ are contained in the same sign-cut graph of $G$, say $F$. Since $F\subseteq G$, it follows that $\{u\alpha,v\beta\}$ is separable in $F$.

    $(\Leftarrow, separability)$ Let $\{u\alpha,v\beta\}$ be a separable pair of vertex-sides of $F$.
    If $u$ and $v$ are not sign-consistent vertices of $G$ then the set of edges incident to $u$ and $v$ in $G$ are all contained in $F$, and thus the separability of $\{u\alpha,v\beta\}$ in $F$ clearly carries over to $G$.
    
    Otherwise $u$ or $v$ is a sign-consistent vertex of $G$, say $u$ without loss of generality. By separability, splitting $u\alpha$ and $v\beta$ in $F$ leaves $u$ and $v$ connected by a path (which is contained in $F$), and since $F\subseteq G$ splitting $u\alpha$ and $v\beta$ in $G$ also leaves $u$ and $v$ connected at least by that same path. 
    It remains to show that $u$ and $v$ become separated from $u'$ and $v'$ (in $G$).

    Since $u$ is sign-consistent for $G$ it is a tip in $F$ with sign $\alpha$, and thus \Cref{prop:tip-nontip-nonseparable} implies that $v$ is a tip in $F$.
    Moreover, the edges with vertex-side $u\hat \alpha$ (possibly none if $u$ is also a tip in $G$) are all contained in another sign-cut graph.
    Thus, splitting $u\alpha$ in $G$ amounts to disconnecting $G$ into two components, one containing the edges with vertex-sides $u \alpha$ and the other containing the edges with vertex-sides $u'\hat \alpha$, so $u$ and $v$ become disconnected from $u'$.
    Since $v$ is a tip, $v$ is either a sign-consistent vertex of $G$ or a non-cutvertex tip of $G$. The latter case is trivial, and in the former we can argue as we did for $u$ and conclude that splitting $v\beta$ results in a graph where $v$ and $u$ are disconnected from $v'$.

    $(\Leftarrow, minimality)$ Let $\{u\alpha,v\beta\}$ be a snarl in $F$. Suppose for a contradiction that $G$ has vertex-sides $w\gamma$ and $w\hat\gamma$ with $w\neq u,v$ such that $\{u\alpha,w\gamma\}$ and $\{v\beta, w\hat\gamma\}$ are separable in $G$. Then $w\gamma$ is in $F$ by \Cref{lem:snarls-inside-sign-cut-tree} since otherwise $\{u\alpha,w\gamma\}$ is not separable in $G$. Moreover, by the separability result for the $(\Rightarrow)$ direction given above, $\{u\alpha,w\gamma\}$ and $\{v\beta, w\hat\gamma\}$ are also separable in $F$, a contradiction to the minimality of $\{u\alpha,v\beta\}$. So $\{u\alpha,v\beta\}$ is a snarl in $G$.

    $(\Rightarrow, minimality)$ Let $\{u\alpha,v\beta\}$ be a snarl in $G$. If $F$ has vertex-sides $w\gamma$ and $w\hat\gamma$ with $w\neq u,v$ such that $\{u\alpha,w\gamma\}$ and $\{v\beta, w\hat\gamma\}$ are separable in $F$, then by separability result for the $(\Leftarrow)$ direction given above, $\{u\alpha,w\gamma\}$ and $\{v\beta, w\hat\gamma\}$ are also separable in $G$, a contradiction to the minimality of $\{u\alpha,v\beta\}$, and so $\{u\alpha,v\beta\}$ is a snarl in $F$.
\end{proof}

\snarlsaftercutting*
\begin{proof}
    We prove the two items separately.
    \begin{enumerate}[nosep]
        \item Let $F'$ denote the graph resulting from splitting $u\alpha$ and $v\beta$ in $F$.
        Since $u$ and $v$ are tips in $F$ with signs $\alpha$ and $\beta$, respectively, $F'$ consists of three components, which are $F$ and the two isolated vertices $u'$ and $v'$. Hence, $\{u\alpha,v\beta\}$ is separable.

        To see minimality, suppose for a contradiction that $F$ has vertex-sides $z\gamma$ and $z\hat\gamma$ with $z\notin\{u,v\}$ such that $\{u\alpha,z\gamma\}$ and $\{v\beta,z\hat{\gamma}\}$ are separable. If $z$ is not a tip then neither $\{u\alpha,z\gamma\}$ nor $\{v\beta,z\hat{\gamma}\}$ are separable by \Cref{prop:tip-nontip-nonseparable} (since $u$ and $v$ are tips by assumption).
        So $z$ is a tip in $F$, without loss of generality, with sign $\gamma$. Splitting $z\hat{\gamma}$ results in $z$ being isolated and further splitting $v\beta$ does not create new paths. Hence there is no $v$-$z$ path in the graph resulting from these two splits and thus $\{v\beta,z\hat{\gamma}\}$ is not separable, a contradiction.
        
        Hence $\{u\alpha,v\beta\}$ is a snarl of $F$, and by \Cref{lem:snarls-G-F}, of $G$ too.

        \item By \Cref{lem:snarls-G-F} there is a sign-cut graph of $G$ where $\{u\alpha,v\beta\}$ is a snarl. Notice that $u$ and $v$ are not sign-consistent in $G$ because they are non-tips in $F$. Therefore the only sign-cut graph containing these vertices is exactly $F$ and thus $\{u\alpha,v\beta\}$ is a snarl in $F$.
        
        First we show that there is a block of $F$ containing the vertices $u$ and $v$. Suppose otherwise.       
        Because $u$ is a non-tip in $F$ we can apply \Cref{prop:mixed-signs-inside-block} to conclude that there are edges $\{u+,z\gamma\}$ and $\{u-,w\delta\}$ within a block of $F$.
        So there is a $z$-$w$ path $p$ in this block that avoids $u$.
        After splitting $u\alpha$, one of $\{u+,z\gamma\}$ and $\{u-,w\delta\}$ remains incident to $u$ and the other becomes incident to $u'$, so $u$ and $u'$ remain connected via $p$. Since $v$ is in a different block than $u$ by assumption, vertex $v$ is not in $p$ and so splitting $v\beta$ does not separate $u$ from $u'$, contradicting the separability of $\{u\alpha,v\beta\}$. Hence $u$ and $v$ are in the same block of $F$, say $H$.
        
        Since $\{u\alpha,v\beta\}$ is separable and $u,v \in V(H)$, applying \Cref{prop:mixed-dangling-blocks} to $F,H,u,v$ implies that $u$ and $v$ have no dangling blocks with respect to $H$.
        Since $u$ and $v$ are non-tips in $F$ by assumption and $F$ is a sign-cut graph of $G$, it follows that $H$ is the unique block of $F$ that contains vertex-sides of different signs at $u$ and $v$, in other words, $u$ and $v$ are non-tips only in $H$.

        We are left to show that $\{u,v\}$ is a split pair of $H$.
        If $\{u,v\}$ is an edge we are done, so suppose otherwise.
        Graph $H$ has vertices $x \in N^+_H(u)$ and $y \in N^-_H(u)$ because $u$ is a non-tip in $H$. Furthermore, since $\{u,v\}$ is not an edge, both $x$ and $y$ are distinct from $v$. Clearly $x$ is contained in the snarl component of $\{u\alpha,v\beta\}$ and $y$ is not.
        So if $\{u,v\}$ is not a separation pair of $H$ then $H-\{u,v\}$ has an $x$-$y$ path and therefore the graph resulting from splitting $u\alpha$ and $v\beta$ in $F$ has a $u$-$u'$ path, contradicting the separability of $\{u\alpha,v\beta\}$.
    \end{enumerate}
\end{proof}

\Cref{lem:snarls-G-F} has a convenient consequence for arguing on the minimality of separable vertex-sides. If $\{u\alpha,v\beta\}$ is separable in $F$ then in order to show that it is a snarl in $G$ it is enough that no vertex-side in $F$ violates minimality, even if the component of $\{u\alpha,v\beta\}$ in $G$ spans over $F$. This may the case if at least one vertex between $u$ and $v$ is a tip in $F$, for instance, if $F$ has three sign-consistent vertices of $G$ then any two of these vertices (together with the obvious signs) form a snarl whose component spans over $F$. Further, minimality is trivial to check in this case as it is shown in the proof of (1) of \Cref{thm:where-are-snarls-after-cutting}.
On the other hand it is not hard to see that in the other case, i.e., non-tip with non-tip snarls, the component is contained in $F$. For these we give a result showing that the vertex-sides are in the block $H$ where both $u$ and $v$ appear, i.e., they are not contained in the blocks ``attached'' to $H$ by $u$ or $v$.

\subsection{The snarl finding algorithm}

In this section we develop an algorithm to find snarls whose vertices form a separation pair of a block of $G$. Since snarls are only defined by their separability and minimality, it is not required to maintain any partial information during the algorithm. In fact, most of our effort is devoted to understanding where and how snarls arise in the different nodes of the SPQR tree.

We start by giving a useful result for showing minimality of separable pairs of vertex-sides. Then we show results giving (mostly) sufficient conditions for a snarl to exist within the different types of nodes of the SPQR tree. Finally we present our algorithm and give a correctness proof.

\begin{proposition}\label{prop:snarls-disjoint-paths}
    Let $G$ be a bidirected graph and let $u\alpha$ and $v\beta$ be vertex-sides of $G$ such that $\{u\alpha,v\beta\}$ is separable with component $X$. If $X$ has two internally vertex-disjoint $u$-$v$ paths then $\{u\alpha,v\beta\}$ is a snarl.
\end{proposition}
\begin{proof}
    Suppose for a contradiction that $X$ has vertex-sides $w\gamma$ and $w\hat\gamma$ such that $\{u\alpha, w\gamma\}$ and $\{v\beta, w\hat{\gamma}\}$ are separable.
    Since $w\neq u,v$, at least one of $p_1,p_2$ avoids $w$; assume without loss of generality that $p_2$ avoids $w$, and let $C$ be the connected component of $X-w$ containing $u$ and $v$.

    Since $\{v\beta, w\hat{\gamma}\}$ is separable, the graph obtained by splitting $v\beta$ and $w\hat{\gamma}$ has a $v$-$w$ path.
    Let $\{x\delta, w\hat{\gamma}\}$ be the last edge of such a path (so $x$ is the predecessor of $w$ in the path). Then $x$ is reachable from $v$ in $G-w$ and, since the edge $\{x\delta, w\hat{\gamma}\}$ is also present in the graph obtained by splitting $u\alpha$ and $v\beta$, we have $x\in V(X)$.
    
    If $x\notin V(C)$ then $x$ lies in a component of $X-w$ disjoint from $\{u,v\}$, whose vertices (being in $X\setminus\{u,v\}$) have no neighbors outside $X$; thus $x$ would not be reachable from $v$ in $G-w$, a contradiction. Therefore $x \in V(C)$.

    Now split $u\alpha$ and $w\gamma$, and let $w'$ denote the vertex created by splitting $w\gamma$ (so $w'$ is incident to the edges originally incident to $w\hat{\gamma}$).
    The component $C$ contains a $u$-$x$ path avoiding $w$; since it is contained in $X$, it starts at $u$ with a $u\alpha$ vertex-side and is preserved by the split of $u\alpha$.
    Further, the edge $\{x\delta, w\hat{\gamma}\}$ becomes incident to $w'$.
    Hence $u$ reaches $w'$, contradicting the separability of $\{u\alpha, w\gamma\}$.

    Therefore no such $w$ exists and $\{u\alpha,v\beta\}$ is minimal, and so is a snarl.
\end{proof}

\paragraph*{S-, P-, and R-nodes}

The technique to find snarls within S-nodes is similar to the technique used to define sign-cut graphs. Let $\mu$ be an S-node and consider a fixed cyclical order of the edges of $\skel(\mu)$. Say that $v \in \skel(\mu)$ is \emph{good} if the vertex-sides at $v$ in the expansion of the edge to the left of $v$ all have sign $\alpha\in\signs$ and those in the expansion of the edge to its right have sign $\hat{\alpha}$ at $v$, and there are no dangling blocks with respect to $H$ at $v$. We show that the consecutive pairs formed by these vertex-sides in the obvious way form snarls.

\begin{proposition}[Snarls and S-nodes, see \Cref{fig:snarls-spr-nodes}(a)]
\label{prop:S-node-snarls}
    Let $G$ be a bidirected graph, $H$ be a 2-connected subgraph of $G$, $T$ be the SPQR tree of $H$, and $\mu$ be an S-node of $T$.
    Suppose that $\skel(\mu)$ has vertices $v_0,\dots,v_{k-1}$ and edges $e_0,\dots,e_{k-1}$ with the endpoints of $e_i$ being $v_i$ and $v_{(i+1 \mod k)}$ $(k \geq 3)$.
    Let $v_{i_1},\dots,v_{i_q}$ be the good vertices of $\mu$ with $q \geq 2$ listed in the order such that $i_1 < \dots < i_q$ $(i_1,\dots,i_q\subseteq\{0,\dots,k-1\})$, and let $\hat{\alpha}_{i_j},\alpha_{i_j} \in\signs$ denote the signs of the vertex-sides at $v_{i_j}$ in $\expansion(e_{ (i_j-1)\mod k})$ and $\expansion(e_{i_j})$, respectively, with $j \in \{1,\dots,q\}$. Then $\{v_{i_1}\alpha_{i_1}, v_{i_2}\hat{\alpha}_{i_2}\},\dots, \{v_{i_q}\alpha_{i_q}, v_{i_1}\hat{\alpha}_{i_1}\}$ are snarls.
\end{proposition}
\begin{proof}
    For conciseness, let $u = v_{i_j}$, $\alpha=\alpha_{i_j}$, $v = v_{i_{(j + 1)\,\mathrm{mod}\,q}}$, and $\beta=\beta_{i_{(j + 1)\,\mathrm{mod}\,q}}$, for an arbitrary $j \in \{1, 2, \dots, q\}$.
    For separability, first note that after obtaining graph $G'$ by splitting $u\alpha$ and $v\hat\beta$ in $G$, there remains a path from $u$ to $v$ through the edges of
    \[ E\left(\expansion(e_{i_j})\right) \cup E\left(\expansion(e_{(i_j+1)\,\mathrm{mod}\, k})\right) \cup \dots \cup E\left(\expansion\left(e_{\left(-1 + i_{(j+1)\,\mathrm{mod}\,q}\right)\,\mathrm{mod}\,k}\right)\right)\,. \]
    
    It remains to show that $u$ does not reach $u'$ and $v$ does not reach $v'$ in $G'$. All vertex-sides of $u'$ are in $E\left(\expansion(e_{(i_j - 1)\,\mathrm{mod}\,k})\right)$ and all the vertex-sides of $v'$ are in $E\left(\expansion(e_{i_{(j+1)\,\mathrm{mod}\,q}})\right)$. Further, there are no $u'\alpha$ or $v'\hat\beta$ vertex-sides because of splitting.
    Without loss of generality, assume for contradiction that there is a $u$-$u'$ path $p$ in $U(G')$. There then has to exist a last vertex $a$ on the path $p$ and its (not necessarily immediate) successor $b$ with the property that $a$ is $\{u', v'\}$ or in
    \[ V\left(\expansion(e_{i_{(j+1)\,\mathrm{mod}\,q}})\right) \cup V\left(\expansion(e_{(1 + i_{(j+1)\,\mathrm{mod}\,q})\,\mathrm{mod}\,k})\right) \cup \dots \cup V\left(\expansion\left(e_{\left(i_j - 1\right)\,\mathrm{mod}\,k}\right)\right) \]
    with $a \not\in \{u, v\}$, and $b$ is in
    \[ V\left(\expansion(e_{i_j})\right) \cup  V\left(\expansion(e_{(i_j+1)\,\mathrm{mod}\,k})\right) \cup \dots \cup V\left(\expansion\left(e_{\left(-1 + i_{(j+1)\,\mathrm{mod}\,q}\right)\,\mathrm{mod}\,k}\right)\right)\,. \]
    By the definition of splitting, it cannot be the case that $a = v'$ and $b = v$, since there are no edges between $v$ and $v'$ and we have no dangling blocks.
    Similarly, it cannot be that $a = u'$ and $b = u$. On the other hand, we must have either $a = v'$ and $b = v$ or $a = u'$ and $b = u$, because we are operating within an S-node. By the contradiction, $u$ does not reach $u'$ and $v$ does not reach $v'$ in $G'$.
    
    For minimality, suppose for a contradiction that there is are vertex-sides $w\gamma$ and $w\hat\gamma$ with $w \neq u,v$ in the component of $\{u\alpha, v\hat\beta\}$ such that $\{u\alpha, w\gamma\}$ and $\{w\hat\gamma, v\hat\beta\}$ are separable.
    Clearly, $w$ does not have dangling blocks with respect to $H$ by \Cref{prop:mixed-dangling-blocks}. It also must be that $w$ is a vertex of the S-node, since a necessary condition for separability is that there cannot be a path that starts with $w+$ vertex-side and ends with $w-$ vertex-side and does not pass through $u$ or $v$. Let thus $w = v_l$. Since we assumed that $w$ is not a good vertex and there are no dangling blocks, either $\expansion(e_l)$ or $\expansion(e_{(l - 1)\,\mathrm{mod}\,k})$ has both $w+$ and $w-$ vertex-sides. Without loss of generality, assume this to be $e_{l}$. Then, the non-separability of $\{u\alpha, w\hat\gamma\}$ follows by there being a path from $w$ and $w'$ to $v$ if we split $u\alpha$ and $w\hat\gamma$.
\end{proof}

\begin{algorithm}[ht]
\small
\caption{$\mathsf{FindSnarlsInSnodes}(F,H,T)$\label{alg:find-snarls-in-S}}
\KwIn{Sign-cut graph $F$ of a bidirected graph, maximal 2-connected subgraph $H$ of $G$, SPQR tree $T$ of $H$}
\For{each S-node $\mu$ of $T$}{
    $v_0,\dots,v_{k-1} \gets$ ordered sequence of the vertices of $\skel(\mu)$\;
    $e_0,\dots,e_{k-1} \gets$ ordered sequence of the edges of $\skel(\mu)$ such that $v_i$ is an endpoint of $e_{( (i+1) \bmod k)}$ and $e_i$\;
    $W \gets [\;]$\; 
    \For{$i \in [0,k-1]$}{
        \If{$v_i$ is a tip in $F$ or $\mathsf{HasDangling}(F,H,v_i)$}{
            \textbf{continue}\;
        }
        $L, R \gets \expansion(e_i), \expansion(e_{(i+1 \bmod k)})$\;
        \If{$(N^+_H(v_i) \subseteq V(L)$ or $N^+_H(v_i) \subseteq V(R))$ and $(N^-_H(v_i) \subseteq V(L)$ or $N^-_H(v_i) \subseteq V(R))$}{
            \tcp{$v_i$ is good}
            $\alpha \gets +$ if $(N^+_H(v_i) \subseteq V(R))$ else $-$\;
            $W.\mathsf{append}(v_i\hat{\alpha})$\;
            $W.\mathsf{append}(v_i\alpha)$\;
        }
    }
    Report the pairs formed by the vertex-sides in $W$ in consecutive positions starting from the second element, and lastly pair the last vertex-side with the first vertex-side of $W$\; \label{line:S-node-snarls}
}
\end{algorithm}

For P-nodes we can give a characterization of separability, similarly to \Cref{prop:P-node} and superbubbloids.

\begin{proposition}[Snarls and P-nodes, , see \Cref{fig:snarls-spr-nodes}(b)]
\label{prop:P-node-snarls}
    Let $G$ be a bidirected graph, $F$ be a sign-cut graph of $G$, and $H$ be a 2-connected subgraph of $F$ with SPQR tree $T$. Let $\mu$ be a P-node of $T$ whose skeleton has edges $e_1,\dots,e_k$ with endpoints $\{u,v\}$ $(k\geq3)$.
    Let $\alpha,\beta\in\signs$.
    Let $E^{\alpha}_u = \{ e_i : V(\expansion(e_i))\cap N^{\alpha}_H(u) \neq \emptyset \}$ and $E^{\beta}_v = \{e_i : V(\expansion(e_i))\cap N^{\beta}_H(v) \neq \emptyset\}$.
    Then $\{u\alpha,v\beta\}$ is separable in $F$ if and only if $E^{\alpha}_u \neq \emptyset$, $E^{\alpha}_u \cap E^{\hat{\alpha}}_u = \emptyset$, $E^{\beta}_v \cap E^{\hat{\beta}}_v = \emptyset$, $E^{\alpha}_u = E^{\beta}_v$, and $u$ and $v$ have no dangling blocks with respect to $H$.
\end{proposition}
\begin{proof}
    $(\Rightarrow)$ Suppose that $\{u\alpha,v\beta\}$ is separable in $F$ with component $X$. We show that each condition described in the statement holds.

    If $u$ and $v$ are tips in $F$ with signs $\alpha$ and $\beta$, respectively, then each condition clearly holds (notice that we do not impose $E^{\hat{\alpha}}_u\neq\emptyset$ in the conditions of the statement). By \Cref{prop:tip-nontip-nonseparable} a tip and a non-tip do not form a separable pair, so we can assume that $u$ and $v$ are both non-tips in $F$.
    By \Cref{prop:mixed-dangling-blocks} it follows that $u$ and $v$ have no dangling blocks since $\{u\alpha,v\beta\}$ is separable. We show that the remaining conditions hold.

    If $E_u^{\alpha} = \emptyset$ then $H$ has no vertex-sides of $u$ with sign $\alpha$. Since $F$ is a sign-cut graph, there is a block of $F$ containing opposite vertex-sides of $u$, for otherwise $u$ is a non-tip and a sign-consistent vertex in $F$, contradicting the fact that $F$ is a sign-cut graph of $G$. In other words, $u$ has a dangling block with respect to $H$, contradicting that $u$ has no dangling blocks. Therefore $E_u^{\alpha} \neq \emptyset$.

    If $e\in E_u^{\alpha} \cap E_u^{\hat{\alpha}}$ then $\expansion(e)$ has edges $\{u\alpha,x\gamma\}$ and $\{u\hat{\alpha},y\delta\}$. Notice that $x,y \neq v$ by construction of the P-nodes: a real edge with endpoints $u$ and $v$ would constitute a split component of $\{u,v\}$ and thus would be represented alone by an edge in $\skel(\mu)$. Now, $\expansion(e)$ has an $x$-$y$ path avoiding $u$ and $v$ since otherwise $x$ and $y$ are in different split components of $\mu$, contradicting the fact that $x,y \in V(\expansion(e))$. So the graph resulting from splitting $u\alpha$ and $v\beta$ has a $u$-$u'$ path, a contradiction. Thus $E_u^{\alpha} \cap E_u^{\hat{\alpha}} = \emptyset$, and symmetrically we can deduce $E_v^{\beta} \cap E_v^{\hat{\beta}} = \emptyset$.

    If $E_u^{\alpha} \neq E_v^{\beta}$ then, without loss of generality, there is an edge $e \in E_u^{\alpha} \setminus E_v^{\beta}$. Since $E_u^{\alpha} \cap E_u^{\hat{\alpha}} = \emptyset$ and $E_v^{\beta} \cap E_v^{\hat{\beta}} = \emptyset$, it follows that every vertex-side of $\expansion(e)$ in $v$ has sign $\hat{\beta}$. Since $\expansion(e)$ is connected, it has a $u\alpha$-$v\hat\beta$ path $p$ (which avoid the vertex-sides $u\hat\alpha$ and $v\beta$, since it is a path). By the separability of $\{u\alpha,v\beta\}$, its component contains $u$ and $v$ and does not contain $u'$ and $v'$, but $p$ connects $u$ and $v'$, and thus $v$ and $v'$ are connected, a contradiction. Therefore $E_u^\alpha=E_v^\beta$.

    $(\Leftarrow)$
    If $E^{\alpha}_u \neq \emptyset$, $E^{\alpha}_u \cap E^{\hat{\alpha}}_u = \emptyset$, $E^{\beta}_v \cap E^{\hat{\beta}}_v = \emptyset$, $E^{\alpha}_u = E^{\beta}_v$, and $u$ and $v$ have no dangling blocks with respect to $H$, then the separability of $\{u\alpha,v\beta\}$ follows at once.
\end{proof}

\begin{figure}[H]
    \centering
    \includegraphics[width=\linewidth]{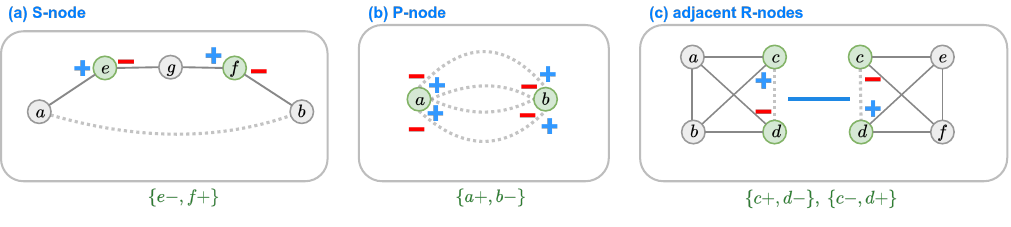}
    \caption{Finding nontip-nontip snarls in the three types of SPQR tree nodes. Real edges are solid gray and virtual edges are dashed. (a)~S-node: vertices $e$ and $f$ (green) have opposite signs on their two adjacent edges and are the only ``good'' vertices (see~\Cref{prop:S-node-snarls}), yielding the snarl $\{e{-},f{+}\}$. (b)~P-node: the signs at $a$ and $b$ partition across the parallel edges, yielding $\{a{+},b{-}\}$. (c)~Adjacent R-nodes: the signs at $c$ and $d$ partition across the boundary, yielding $\{c{+},d{-}\}$ and $\{c{-},d{+}\}$. In all cases, the reversed sign counterpart is also a snarl.}
    \label{fig:snarls-spr-nodes}
\end{figure}

\begin{algorithm}[ht]
\small
\caption{$\mathsf{FindSnarlsInPnodes}(F,H,T)$\label{alg:find-snarls-in-P}}
\KwIn{Sign-cut graph $F$ of a bidirected graph, maximal 2-connected subgraph $H$ of $G$, SPQR tree $T$ of $H$}
\For{each P-node $\mu$ of $T$}{
    $u,v \gets$ the vertices of $\skel(\mu)$\;
    $e_1,\dots,e_k \gets$ the edges of $\skel(\mu)$\;
    $X_1,\dots,X_k \gets \expansion(e_1),\dots,\expansion(e_k) \; (k\geq 3)$\;
    \If{$\mathsf{HasDangling}(F,H,u)$ or $\mathsf{HasDangling}(F,H,v)$ or $u$ is a tip in $F$ or $v$ is a tip in $F$}{
        \textbf{continue}\;
    }
    Build the sets $E^+_u, E^-_u, E^+_v, E^-_v$ as described in \Cref{prop:P-node-snarls}\;
    \tcp{Since $u$ and $v$ are non-tips in $F$ and have no dangling blocks with respect to $H$, all of the above are non-empty}
    \For{$\alpha,\beta \in\signs$}{
        \If{$E^{\alpha}_u \neq \emptyset$, $E^{\alpha}_u \cap E^{\hat\alpha}_u = \emptyset$, $E^{\beta}_v \cap E^{\hat\beta}_v = \emptyset$, $E^{\alpha}_u = E^{\beta}_v$}{
            \tcp{Equivalently, $E^{\hat\alpha}_u \neq \emptyset$, $E^{\alpha}_u \cap E^{\hat\alpha}_u = \emptyset$, $E^{\beta}_v \cap E^{\hat\beta}_v = \emptyset$, $E^{\hat\alpha}_u = E^{\hat\beta}_v$}
            \If{$|E^{\alpha}_u| = 1$ and the pertaining node of the edge in $E^{\alpha}_u$ is an S-node}{
                \tcp{$u,v$ are adjacent in the skeleton of this S-node and are good, so if $\{u\alpha,v\beta\}$ is a snarl then it is reported when S-nodes are examined}
                \textbf{continue}\;
            }
            \Else{
                Report $\{u\alpha,v\beta\}$\; \label{line:P-node-snarls-1-2}
            }
            \If{$|E^{\hat\alpha}_u| = 1$ and the pertaining node of the edge in $E^{\hat\alpha}_u$ is an S-node}{
                \textbf{continue}\;
            }
            \Else{
                Report $\{u\hat\alpha,v\hat\beta\}$\; \label{line:P-node-snarls-2-2}
            }
        }
    }
}
\end{algorithm}

\begin{proposition}[Snarls and R-nodes, see \Cref{fig:snarls-spr-nodes}(c)]
\label{prop:R-node-snarls}
    Let $G$ be a bidirected graph, $H$ be a 2-connected subgraph of $G$, $T$ be the SPQR tree of $H$, and $\nu,\mu$ be adjacent nodes in $T$.
    Let $e_\mu=\{u,v\}\in E(\skel(\nu))$ be the virtual edge pertaining to $\mu$ and $e_\nu=\{u,v\}\in E(\skel(\mu))$ be the virtual edge pertaining to $\nu$.
    Let $\alpha,\beta\in\signs$.
    If $\nu$ is an R-node and all vertex-sides at $u$ and $v$ in $\expansion(e_\nu)$ have signs $\alpha$ and $\beta$, respectively, and all vertex-sides at $u$ and $v$ in $\expansion(e_\mu)$ have signs $\hat{\alpha}$ and $\hat{\beta}$, respectively, and $u$ and $v$ have no dangling blocks with respect to $H$, then $\{u\alpha,v\beta\}$ is a snarl.
\end{proposition}
\begin{proof}
    Let $G'$ denote the graph after splitting $u\alpha$ and $v\beta$.    
    Since $u$ and $v$ have no dangling blocks with respect to $H$, for each block $H'\neq H$ that intersects $u$, the vertex-side of $H'$ at $u$ all have the same sign, and the same for $v$. So the blocks in $G'$ containing $u\alpha$ vertex-sides remain attached to $u$ and those with $u\hat{\alpha}$ vertex-sides are reattached to $u'$, and the same for $v$. Thus $u$ and $u'$ are not connected in $G'$ via any of these blocks, and the same for $v$ and $v'$.
    Now, the fact that $u$ and $v$ are separated from $u'$ and $v'$ in $G'$ follows from the fact that the tree-edge $\{\nu,\mu\}$ encodes a separation of $H$ where one side contains all vertex-sides of $G$ at $u$ and $v$ of signs $\alpha$ and $\beta$, respectively, and the other side contains those with signs $\hat\alpha$ and $v\hat\beta$. Finally, the fact that $u$ and $v$ remain connected follows from the fact that $\expansion(e_\nu)$ is connected. Therefore $\{u\alpha,v\beta\}$ is separable.
    
    For minimality notice that $\expansion(e_\nu)$ has two internally vertex-disjoint $u\alpha$-$v\beta$ paths since $\skel(\nu)$ is 3-connected. So the component of $\{u\alpha,v\beta\}$ in $G'$ containing $u$ and $v$ also contains these two paths and thus we can apply \Cref{prop:snarls-disjoint-paths} to conclude that $\{u\alpha,v\beta\}$ is a snarl.
\end{proof}

\begin{algorithm}[ht]
\small
\caption{$\mathsf{FindSnarlsBetweenRRnodes}(F,H,T)$\label{alg:find-snarls-in-RR}}
\KwIn{Sign-cut graph $F$ of a bidirected graph, maximal 2-connected subgraph $H$ of $F$, SPQR tree $T$ of $H$}
\For{$\{\nu,\mu\} \in E(T)$}{
    \If{$\nu$ and $\mu$ are R-nodes}{
        $e_{\mu} \gets$ the virtual edge in $\skel(\nu)$ pertaining to $\mu$\;
        $e_{\nu} \gets$ the virtual edge in $\skel(\mu)$ pertaining to $\nu$\;
        $u,v \gets$ the endpoints of $e_\nu,e_\mu$\;
        $X_\mu, X_\nu \gets \expansion(e_\mu), \expansion(e_\nu)$\;
        \If{$\mathsf{HasDangling}(F,H,u)$ or $\mathsf{HasDangling}(F,H,v)$ or $u$ is a tip in $F$ or $v$ is a tip in $F$}{
            \textbf{continue}\;
        }
        \For{$\alpha,\beta \in\signs$}{
            \If{$N^{\alpha}_H(u) \subseteq V(X_\mu)$ and $N^{\hat{\alpha}}_H(u) \subseteq V(X_\nu)$ and $N^{\beta}_H(v) \subseteq V(X_\mu)$ and $N^{\hat{\beta}}_H(v) \subseteq V(X_\nu)$}{
                Report $\{u\alpha, v\beta\}, \{u\hat{\alpha}, v\hat{\beta}\}$\; \label{line:R-node-snarls}
            }
        }
    }
}
\end{algorithm}

\paragraph*{Correctness and running time}

The next results are merely technical tools used in the completeness part of \Cref{thm:snarls-correct} to exclude mixed-sign configurations incident to R-nodes, and for an edge-case arising in the proof. 

\begin{proposition}
\label{prop:R-node-mixed}
    Let $G$ be a bidirected graph, $H$ be a 2-connected subgraph of $G$, and $T$ be the SPQR tree of $H$. Let $\nu$ be a node of $T$ such that $\skel(\nu)$ has a virtual edge $e_\mu=\{u,v\}$ pertaining to an R-node $\mu$. If $\expansion(e_\mu)$ has opposite vertex-sides at $u$, then $\{u\alpha,v\beta\}$ is not separable for any $\alpha,\beta \in \signs$.
\end{proposition}
\begin{proof}
    By assumption, $\expansion(e_\mu)$ has edges $\{u\alpha,x\gamma\}$ and $\{u\hat{\alpha},y\delta\}$. Notice that vertex $x$ is contained in the expansion of a virtual edge $e_x \in\skel(\mu)$ where $u$ is an endpoint, and analogously for $y$ with edge $e_y$ where $u$ is also an endpoint.
    Since $u$ and $v$ are not both the endpoints of these virtual edges as $x,y\in\expansion(e_\mu)$, we have that the other endpoint of $e_x$, say $x'$, is distinct from $v$; similarly, the other endpoint of $e_y$, say $y'$, is distinct from $v$.
    Now notice that $\expansion(e_x)$ has an $x-x'$ path avoiding $u$ by \Cref{lem:reaches-in-expansion}, and analogously $\expansion(e_y)$ has a $y$-$y'$ path avoiding $u$.
    Since these paths are contained in $\expansion(e_x)$ and $\expansion(e_y)$, respectively, and $v\notin V(\expansion(e_x)),V(\expansion(e_y))$, they also avoid $v$. Moreover, $\skel(\mu)$ has an $x'$-$y'$ path avoiding $u$ and $v$ because the skeleton of R-nodes is 3-connected. Therefore, $\expansion(e_\mu)$ has an $x$-$y$ path avoiding $u$ and $v$ and thus the graph resulting from splitting $u\alpha$ and $v\beta$ has a $u$-$u'$ path, and therefore $\{u\alpha,v\beta\}$ is not separable.
\end{proof}

\begin{proposition}\label{prop:snarl-edge-case}
    Let $G$ be a bidirected graph, $F$ be a sign-cut graph of $G$, $H$ be a block of $F$ containing the vertices $u$ and $v$, and let $\{u\alpha,v\beta\}$ be a snarl of $F$.
    If $u$ and $v$ are non-tips in $H$, $\{u,v\}$ is an edge of $H$, and $\{u,v\}$ is not a separation pair of $H$, then either $e=\{u\alpha,v\beta\} \in E(H)$ and all vertex-sides at $u$ and $v$ except for those in $e$ have signs $\hat{\alpha}$ and $\hat{\beta}$, respectively, or $e=\{u\hat{\alpha},v\hat{\beta}\} \in E(H)$ and all vertex-sides at $u$ and $v$ except for those in $e$ have signs $\alpha$ and $\beta$, respectively.
\end{proposition}
\begin{proof}
    First notice that $\{u\alpha, v\hat{\beta}\}, \{u\hat{\alpha}, v\beta\} \not\in E(H)$, since otherwise splitting $u\alpha$ and $v\beta$ results in a graph with an edge $\{u,v'\}$ or $\{u',v\}$, contradicting the separability of $\{u\alpha,v\beta\}$.
    
    If $e=\{u\alpha,v\beta\} \in E(H)$ then we can pick an edge $\{u\hat{\alpha}, x\gamma\} \in E(H)$ since $u$ is a non-tip in $H$. Suppose for a contradiction that $H$ has an edge $\{u\alpha, y\delta\} \neq e$. Since $\{u,v\}$ is not a separation pair of $H$, $H$ has an $x$-$y$ path avoiding $u$ and $v$.
    This path remains after splitting $u\alpha$ and $v\beta$, thus violating separability between $u$ and $u'$, a contradiction. Therefore all vertex-sides at $u$ except for the one in $e$ have signs $\hat{\alpha}$, and we can argue symmetrically to conclude the same for $v$ and $\hat{\beta}$.
    
    Otherwise we have $e=\{u\hat\alpha,v\hat\beta\} \in E(H)$ and we proceed identically as above to conclude that all vertex-sides at $u$ and $v$ except for those in $e$ have signs $\alpha$ and $\beta$, respectively.
\end{proof}

We can now present the correctness proof of \Cref{alg:snarls}.

\begin{theorem}
\label{thm:snarls-correct}
    Let $G$ be a bidirected graph. The algorithm identifying snarls (\Cref{alg:snarls}) is correct, that is, it identifies all snarls of $G$ and only its snarls.
\end{theorem}
\begin{proof}

    \textbf{(Completeness.)}
    We argue that every snarl of $G$ is reported by the algorithm.

    Let $\{u\alpha,v\beta\}$ be a snarl of $G$. By \Cref{lem:snarls-G-F} there is a sign-cut graph $F$ of $G$ where $\{u\alpha,v\beta\}$ is also a snarl.
    
    By \Cref{prop:tip-nontip-nonseparable} it follows that either $u$ and $v$ are both tips or both non-tips in $F$.
    If $u$ and $v$ are both tips then the snarl in question is encoded in the list described in Line~\ref{line:tip-tip-snarls} of \Cref{alg:snarls} (in $\mathcal{T}_i$ every two vertex-sides form a snarl).
    Otherwise $u$ and $v$ are both non-tips in $F$. By (2) of \Cref{thm:where-are-snarls-after-cutting} it follows that there is a unique block of $F$, say $H$, where $u$ and $v$ are both non-tips and where $\{u,v\}$ is a split pair. 
    Let $T$ denote the SPQR tree of $H$.

    If $\{u,v\}$ is a separation pair of $H$ then \Cref{lem:spqr-tree-contains-split-pairs} implies that $T$ has an edge $\{\nu,\mu\}$ corresponding to the separation pair $\{u,v\}$ or $T$ has an S-node where $u$ and $v$ are nonadjacent in the skeleton. Therefore it is enough to analyze all the S- and P-nodes individually, and the tree-edges between R-nodes. Importantly, we remark that the current assumptions do not exclude the possibility that $\{u,v\}$ is an edge of $H$. We discuss each of the possible cases.

    \begin{enumerate}[nosep]
        \item Suppose that $\mu$ is an S-node of $T$ such that $\{u,v\}\subseteq V(\skel(\mu))$. If $u$ and $v$ are good in $\mu$ and consecutive in the (circular) list of good vertices then $\{u\alpha,v\beta\}$ is reported in Line~\ref{line:S-node-snarls}. If $u$ and $v$ are good in $\mu$ and not consecutive in the (circular) list of good vertices, then there is a good vertex $w$ in between $u$ and $v$ such that $w\gamma$ and $w\hat\gamma$ are vertex-sides violating minimality of $\{u\alpha,v\beta\}$ (see the proof of \Cref{prop:S-node-snarls}), a contradiction to the fact that $\{u\alpha,v\beta\}$ is a snarl. If $u$ or $v$ is not good then we require a careful argument for which we do case analysis. Suppose without loss of generality that $u$ is not good.
        \begin{enumerate}[nosep]
            \item Suppose that there is an edge $e=\{u,w\}\in E(\skel(\mu))$ with $w \neq v$ such that $\expansion(e)$ has opposite vertex-sides at $u$. Then $\expansion(e)$ has edges $\{u\alpha,a\gamma\}$ and $\{u\hat{\alpha},b\delta\}$, and notice that $a,b\neq v$ because $w \neq v$. Let $w$ be the first vertex in $\skel(\mu)$ on an $a$-$v$ path in $\expansion(e)$ that avoids $u$ (such a path exists by \Cref{lem:reaches-in-expansion}). Due to the structure of S-nodes, doing the same reasoning for $b$ also yields vertex $w$. Then $H$ has an $a$-$b$ path avoiding $v$, and also avoiding $u$ by construction. So the graph resulting from splitting $u\alpha$ and $v\beta$ has an $a$-$b$ path and thus it has a $u$-$u'$ path, a contradiction.
            
            \item Otherwise the only edge witnessing the fact that $u$ is not good is the edge $e=\{u,v\} \in \skel(\mu)$. If $v$ is good then it is not hard to see that $\{u\alpha,v\beta\}$ is not separable, again, by two applications of \Cref{lem:reaches-in-expansion} to the out- and in-neighbor of $u$ in $\expansion(e)$ (where the vertex to be avoided is $u$). If $v$ is not good and there is an edge distinct from $e$ in $\skel(\mu)$ witnessing this fact, then we can argue for $v$ as we did in item (1a) for $u$ and contradict the separability of $\{u\alpha,v\beta\}$.
            The last case is thus when both $u$ and $v$ are not good and $e$ is the only edge such that $\expansion(e)$ has opposite vertex-sides at $u$ and $v$, and the other two edges of $\skel(\mu)$ incident to $u$ and $v$ are such that their expansions have only vertex-sides of the same sign at $u$ and $v$.
            If the pertaining node of $e$ is an R-node then \Cref{prop:R-node-mixed} gives a contradiction to the separability of $\{u\alpha,v\beta\}$. Since no two S-nodes are adjacent in $T$ the pertaining node of $e$ is a P-node and the snarl is reported once P-nodes are analyzed as shown next in item (2).
        \end{enumerate}

        \item Suppose that $\mu$ is a P-node of $T$ such that $V(\skel(\mu))=\{u,v\}$. Since $\{u\alpha,v\beta\}$ is separable, \Cref{prop:P-node-snarls} implies that the conditions described in the statement hold.
        Suppose that $|E^{\alpha}_u| = 1$ and let $e$ be the unique edge in $E^{\alpha}_u$.
        If $e$ pertains to an S-node then notice that $u$ and $v$ are good and that no other vertex is good, since otherwise a contradiction to the minimality of $\{u\alpha,v\beta\}$ follows (to see in detail why, see the proof of \Cref{prop:S-node-snarls}). Thus $\{u\alpha,v\beta\}$ is reported when S-nodes are analyzed.
        If $e$ does not pertain to an S-node then $e$ is a real edge or it pertains to an R-node. So the snarl is reported in Line~\ref{line:P-node-snarls-1-2} (symmetrically, Line~\ref{line:P-node-snarls-2-2}) and the same if $|E^{\alpha}_u| > 1$.
        
        Notice the following observation concerning vertices that do not form a separation pair but form an edge of $H$.

        \begin{observation}\label{obs:edge-case}
            If $\skel(\mu)$ has two real edges $\{u\alpha,v\beta\}$ and $\{u\hat{\alpha},v\hat{\beta}\}$ and one virtual edge pertaining to an S-node whose expansion contains only vertex-sides at $u$ and $v$ with signs $\alpha$ and $\beta$, respectively, then the snarl $\{u\alpha,v\beta\}$ is reported in Line~\ref{line:P-node-snarls-1-2} of \Cref{alg:find-snarls-in-P}.
        \end{observation}
        
        \item Suppose that $\mu$ is an R-node of $T$ with a virtual edge $\{u,v\}$. If the pertaining node $\nu$ of this virtual edge is an S- or a P-node then $\{u\alpha,v\beta\}$ was reported before. So $\nu$ is an R-node. We claim that the vertex-sides in $\expansion(e_\nu)$ all have the same sign in $u$, and those in $\expansion(e_\mu)$ all have the opposite sign in $u$, and the same for $v$. Suppose for a contradiction (and without loss of generality) that $\expansion(e_\mu)$ has vertex-sides of opposite signs at $u$. Then \Cref{prop:R-node-mixed} gives a contradiction to the fact that $\{u\alpha,v\beta\}$ is separable.
        Now notice that the conditions just established on the vertex-sides of $u$ and $v$ are precisely those described in \Cref{alg:find-snarls-in-RR} and thus the snarl is reported in Line~\ref{line:R-node-snarls} when the assignments to the sign variables $\alpha$ and $\beta$ match those of the snarl.
    \end{enumerate}

    Otherwise $\{u,v\}$ is an edge of $H$ and is not a separation pair of $H$, so we are in conditions of applying \Cref{prop:snarl-edge-case} from where two cases follow.
    We have $e=\{u\alpha,v\beta\} \in E(H)$ and all vertex-sides at $u$ and $v$ except for those in $e$ have signs $\hat{\alpha},\hat{\beta}$, respectively, or $e=\{u\hat{\alpha},v\hat{\beta}\} \in E(H)$ and all vertex-sides at $u$ and $v$ except for those in $e$ have signs $\alpha,\beta$, respectively.
    In the former case the snarl $\{u\alpha,v\beta\}$ is reported in Line~\ref{line:edge1-snarls}.
    In the latter case the snarl $\{u\alpha,v\beta\}$ is reported in Line~\ref{line:edge2-snarls} (where the signs are written with the respective opposites) if also no S-node of $T$ contains both $u$ and $v$.
    So we are left to argue the case when $T$ has an S-node $\sigma$ whose skeleton contains both $u$ and $v$. Our goal now is to contradict the fact $\{u\alpha,v\beta\}$ is a snarl or to show that it is reported in another phase of the algorithm.
    
    Notice that $u$ and $v$ are adjacent in $\skel(\sigma)$ because $\{u,v\}$ is an edge of $H$, so let $e=\{u,v\} \in E(\skel(\sigma))$. Let $e_u,e_v \neq e$ denote the edges incident to $u$ and $v$ in $\skel(\sigma)$, respectively. 
    We do case analysis on the type of $e$.
    \begin{enumerate}[nosep]
        \item If $e$ is a real edge then the snarl is reported when S-nodes are analyzed: $u$ and $v$ are classified as good vertices due to the assumption on the vertex-sides and moreover they are consecutive in the circular list of good vertices since they are adjacent in $\skel(\sigma)$.
        \item Otherwise $e$ is a virtual edge. Since $\{u,v\}$ is not a separation pair the pertaining node of $e$ necessarily is a P-node, say $\pi$, such that $\skel(\pi)$ consists of exactly two real edges and one virtual edge (which pertains to $\sigma$). Indeed, if $\skel(\pi)$ has three real edges then it is not hard to see that $\{u\alpha,v\beta\}$ is not separable, and if $\skel(\pi)$ has at least two virtual edges then $\{u,v\}$ is a separation pair, both leading to a contradiction.
        By the assumption on the vertex-sides at $u$ and $v$, we have that $\{u\alpha,v\beta\}$ and $\{u\hat\alpha,v\hat\beta\}$ are the real edges of $\skel(\sigma)$, and that the vertex-sides contained in $\expansion(e_u)$ all have sign $\alpha$ and those contained in $\expansion(e_v)$ have sign $\beta$.
        But then we are exactly in the conditions described in \Cref{obs:edge-case}, and thus the snarl is reported when P-nodes are analyzed.
    \end{enumerate}
    
    All cases were examined and so every snarl of $G$ is reported by the algorithm.

    \textbf{(Soundness.)}
    We argue that the algorithm only reports snarls.
    
    By \Cref{lem:snarls-G-F}, if $\{u\alpha,v\beta\}$ is a snarl in a sign-cut graph of $G$ then it is also a snarl in $G$. Thus, let $F$ denote the sign-cut graph where the pair $\{u\alpha,v\beta\}$ is reported. We show that $\{u\alpha,v\beta\}$ is separable and minimal in $F$.

    If $u\alpha$ and $v\beta$ are vertex-sides of the set built in Line~\ref{line:tip-tip-snarls} of \Cref{alg:snarls} then $u$ and $v$ are both tips in $F$. By (1) of \Cref{thm:where-are-snarls-after-cutting} it follows that $\{u\alpha,v\beta\}$ is a snarl and thus the set $\mathcal{T}_i$ only encodes snarls.

    If $\{u\alpha,v\beta\}$ is reported by virtue of Line~\ref{line:S-node-snarls} of \Cref{alg:find-snarls-in-S} then $u\alpha$ and $v\beta$ are consecutive elements of $W$ and $u\neq v$. Moreover, the conditions expressed in \Cref{alg:find-snarls-in-S} identify all and only good vertices. Thus \Cref{prop:S-node-snarls} implies that $\{u\alpha,v\beta\}$ is a snarl.

    If $\{u\alpha,v\beta\}$ and $\{u\hat{\alpha},v\hat{\beta}\}$ are reported in Line~\ref{line:R-node-snarls} of \Cref{alg:find-snarls-in-RR} then \Cref{prop:R-node-snarls} implies that $\{u\alpha,v\beta\}$ is a snarl: the conditions of the statement match those in the algorithm. Applying \Cref{prop:R-node-snarls} symmetrically to the other R-node implies that $\{u\hat{\alpha},v\hat{\beta}\}$ is a snarl.
    
    If $\{u\alpha,v\beta\}$ is reported in \Cref{alg:find-snarls-in-P} then \Cref{prop:P-node-snarls} implies that $\{u\alpha,v\beta\}$ is separable.
    Without loss of generality we argue on the minimality for Line~\ref{line:P-node-snarls-1-2}.
    If $|E^{\alpha}_u| > 1$ then the component containing $u$ and $v$ after splitting $u\alpha$ and $v\beta$ in $F$ has two internally vertex-disjoint $u$-$v$ paths (this is easily seen from the structure of P-nodes), so we can apply \Cref{prop:snarls-disjoint-paths} and conclude that $\{u\alpha,v\beta\}$ is a snarl.
    Otherwise we have $|E^{\alpha}_u| = 1$. If the unique edge in $E^{\alpha}_u$ is real then $\{u\alpha,v\beta\}$ is clearly a snarl, and otherwise it is virtual and it pertains to an R-node since no two P-nodes are adjacent in $T$ and if it is an S-node the algorithm does not report anything. Notice now that we are in conditions of applying \Cref{prop:R-node-snarls}, i.e., the conditions on the sets described in \Cref{prop:P-node-snarls} can be reinterpreted and plugged into \Cref{prop:R-node-snarls}, from where we can conclude that $\{u\alpha,v\beta\}$ is a snarl.

    If $\{u\alpha,v\beta\}$ is reported in Line~\ref{line:edge1-snarls} then the pair is clearly a snarl whose component has vertex set $\{u,v\}$.
    If $\{u\hat{\alpha},v\hat{\beta}\}$ is reported in Line~\ref{line:edge2-snarls} then $u$ and $v$ are non-tips and no skeleton of an S-node of $T$ contains the vertices $u$ and $v$. Clearly $\{u\hat{\alpha},v\hat{\beta}\}$ is separable, so we are left to argue minimality.
    Since $\{u,v\}$ is an edge of $U(H)$, $T$ has a node whose skeleton contains the real edge $\{u,v\}$. This node is thus a P- or an R-node, and so $H$ has three internally vertex-disjoint $u$-$v$ paths (the existence of these paths is easily seen from the description of the P- and R-nodes, nonetheless we point to Lemma 2 of~\citep{di1996line}). One of these paths consists of the edge $\{u\alpha,v\beta\}$, and thus $H$ has two internally vertex-disjoint $u\hat\alpha$-$v\hat\beta$ paths because no edge in $F$ distinct from $e$ has a vertex-side $u\alpha$ or $v\beta$.    
    So splitting $u\hat\alpha$ and $v\hat\beta$ in $F$ results in a graph with a component containing $u$ and $v$ which has two internally vertex-disjoint $u$-$v$ paths and hence we can apply \Cref{prop:snarls-disjoint-paths} to conclude that $\{u\hat{\alpha},v\hat{\beta}\}$ is a snarl.

    Every line where the algorithm reports a pair of vertex-sides is analyzed and therefore the algorithm only reports snarls.
\end{proof}

\maintheoremsnarls*
\begin{proof}
    First we argue that \Cref{alg:snarls} runs in linear-time.
    
    Block-cut trees can be built in linear time \citep{Hopcroft73blockcut} and the total size of the blocks is linear in $|V(G)|+|E(G)|$. 
    We also find the sign-cut graphs in linear time, since we only need to identify the sign-consistent cutvertices of the block-cut tree.
    We can suppose that we are analyzing a block $H$ that is 2-connected, since the other cases are trivial to check. Let $|H|=|V(H)|+|E(H)|$. 
    We show that the rest of the algorithm runs in time $O(|H|)$, thus proving the desired bound.

    After building the SPQR tree $T$ of $H$, which can be built in $O(|H|)$ time~\citep{gutwenger2001linear}, the algorithm examines each of the possible node types. By examination of \Cref{alg:find-snarls-in-S,alg:find-snarls-in-P,alg:find-snarls-in-RR} we conclude that the work done is at most linear in the size of the skeleton of the node, except for the neighborhood queries, which we must support in constant-time. But these neighborhood queries can easily be supported in constant-time, exactly how it was described in \Cref{thm:superbubbles}.

    Now we argue that the representation of the snarls has the desired size.
    The bound on the tip-tip snarls follows from the fact that the sum of tips over all sign-cut graphs of $G$ is at most $O(|V(G)|)$.
    For each 2-connected block $H$ examined by the algorithm, the total number of vertices over all skeletons of the SPQR tree of $H$ is $O(|V(H)|)$ (see \cite[Lemma 5]{di1996line}) and that SPQR tree has $O(|V(H)|)$ edges (recall \Cref{lem:spqr-total-size}), and it is not hard to see that $H$ contributes with $O(|V(H)|)$ many (explicitly) listed snarls to $\mathcal{S}$ (the edge-snarls case takes constant time with linear-time preprocessing on the S-nodes and the neighborhoods of the vertices therein). Finally, within each sign-cut graph $F_i$ we have $\sum_{H \text{ block of }F_i} |V(H)| = O(|V(F_i)|)$, and each vertex of $G$ appears in at most two sign-cut graphs. Hence $\sum_{j=1}^{\ell} |S_j| = O(|V(G)|)$.
\end{proof}

\begin{algorithm}[ht]
\small
\caption{$\mathsf{FindEdgeSnarls}(F,H,T)$\label{alg:find-edge-snarls}}
\KwIn{Sign-cut graph $F$ of a bidirected graph, maximal 2-connected subgraph $H$ of $F$, $T$ the SPQR tree of $H$}
\For{every edge $e = \{u\alpha, v\beta\}$ of $H$ such that $u$ and $v$ are non-tips in $F$}{
    \If{no edge in $F$ distinct from $e$ has a vertex-side $u\alpha$ or $v\beta$}{
        Report $\{u\alpha, v\beta\}$\; \label{line:edge1-snarls}
        \If{no S-node $\mu$ of $T$ is such that $\{u,v\}\subseteq V(\skel(\mu))$}{
            Report $\{u\hat{\alpha}, v\hat{\beta}\}$\; \label{line:edge2-snarls}
        }
    }
}
\end{algorithm}

\begin{algorithm}[ht]
\small
\caption{Snarls representation algorithm}
\label{alg:snarls}
\KwIn{Bidirected graph $G$}
\KwOut{A linear-size encoding of snarls of $G$ as two collections: $\mathcal{T} = \{T_1, \dots, T_k\}$ of vertex-side sets and $\mathcal{S} = \{S_1, \dots, S_\ell\}$ of unordered pairs of vertex-sides, where every unordered pair of distinct vertex-sides in $T_i$ is a snarl and each $S_j$ is a snarl.}
$F_1,\dots,F_k \gets \mathsf{BuildSignCutGraphs}(G)$\;
$\mathcal{T} \gets \{\}$\;
$\mathcal{S} \gets \{\}$\;
\For{$\iink$}{
    \If{$F_i$ is an isolated vertex}{
        \textbf{continue}\;
    }
    $T_i \gets$ the set of vertex-sides $v\alpha$ where $v$ is a tip in $F_i$ with sign $\alpha$\;\label{line:tip-tip-snarls}
    $\mathcal{T} \gets \mathcal{T} \cup \{T_i\}$\;
    \For{every block $H$ of $F_i$}{
        \If{$H$ is a multi-bridge}{
            \tcp{The snarls in multi-bridges are reported in \Cref{alg:find-edge-snarls}}
            \textbf{continue}\;
        }
        \Else{
            $T_H \gets \mathsf{BuildSPQRTree}(H)$\;
            $\mathcal{S} \gets \mathcal{S} \cup \mathsf{FindSnarlsInSnodes}(F_i,H,T_H)$\;
            $\mathcal{S} \gets \mathcal{S} \cup \mathsf{FindSnarlsInPnodes}(F_i,H,T_H)$\;
            $\mathcal{S} \gets \mathcal{S} \cup \mathsf{FindSnarlsBetweenRRnodes}(F_i,H,T_H)$\;   
            $\mathcal{S} \gets \mathcal{S} \cup \mathsf{FindEdgeSnarls}(F_i,H,T_H)$\;
        }
    }
}
\Return{$(\mathcal{T}, \mathcal{S})$}
\end{algorithm}

\section{Ultrabubbles}
\label{sec:ultrabubbles}


\subsection{Setup}

Ultrabubbles are snarls with two additional superbubble-like conditions.

\begin{definition}[Ultrabubble and Ultrabubble component~\citep{paten2018ultrabubbles}]
\label{def:ultrabubble}
    Let $G$ be a bidirected graph. Let $\{u\alpha, v\beta\}$ be a pair of vertex-sides with distinct $u,v\in V(G)$ and $\alpha,\beta \in\signs$. Then $\{u\alpha, v\beta\}$ is an \emph{ultrabubble} if:
    \begin{enumerate}[nosep]
        \item[(a)] \emph{separable:}
        the graph created by splitting $u\alpha$ and $v\beta$ contains a separate component $X \subseteq G$ containing $u$ and $v$ but not $u'$ and $v'$.
        We call $X$ the \emph{ultrabubble component} of $\{u\alpha, v\beta\}$.
        
        \item[(b)] \emph{tipless:} no vertex in $V(X)\setminus \{u,v\}$ is a tip.
        
        \item[(c)] \emph{acyclic:} $X$ is acyclic.

        \item[(d)] \emph{minimal:}
        no vertex-side $w\gamma$ with vertex $w \in X \setminus \{u,v\}$ is such that $\{u\alpha, w\gamma\}$ and $\{w\hat{\gamma}, v\beta\}$ are separable.
    \end{enumerate}
\end{definition}

The \emph{interior} of a separable pair of vertex-sides $\{u\alpha,v\beta\}$ with component $K$ is the vertex set $V(K) \setminus \{u,v\}$.
A \emph{trivial ultrabubble} is an ultrabubble whose interior is empty.

It should be clear that the analogues of \Cref{lem:bubbles-cutvertices} and \Cref{thm:bubbles-split-pairs} hold for ultrabubbles, as their proofs are a simple adaptation from directed to bidirected graphs. So ultrabubbles are confined to blocks, and any ultrabubble is either trivial, the whole graph, or induces a separation pair. Moreover, also an analogue of \Cref{lem:superbubbloid-superbubble-cutvertex} holds for ultrabubbles (see Lemma 4 of~\citep{Harviainen2026.03.28.714704}). The algorithm we propose is essentially the same as the one to find superbubbles presented in \Cref{sec:superbubbles}, having only a few major differences.

\subsection{The ultrabubble finding algorithm}

The algorithm starts by computing the blocks of the input bidirected graph $G$. For each block $H$ of $G$ it builds its SPQR tree $T$ and therein it runs two graph traversals akin to phases one and two of the superbubbles algorithm. The algorithm maintains the following information. Let node $\nu$ be the parent of node $\mu$ in $T$, $e_\mu$ and $e_\nu$ the usual edges, and $X:=\expansion(e_\mu)$.

\begin{itemize}
    \item $\noextremity{\nu}{\mu} := \true$ iff no vertex in $V(X)\setminus\{s,t\}$ is an extremity of $G$.
    
    \item $\acyclic{\nu}{\mu} := 
    \begin{cases}
        \Null, & \text{if $\noextremity{\nu}{\mu}$ is false,}\\
        \true, & \text{otherwise, if $X$ is acyclic,}\\
        \false, & \text{otherwise.}
    \end{cases}$
\end{itemize}

As for the reachability states, since we are working with bidirected graphs now we store four kinds of reachability. So for all $\alpha,\beta\in\signs$ we define the following states.

\begin{itemize}
    \item $\reachesuvab{\nu}{\mu} := 
    \begin{cases}
        \Null, & \text{if $\acyclic{\nu}{\mu}$ is $\false$ or $\Null$,}\\ 
        \true, & \text{otherwise, if $X$ has an $s\alpha$-$t\beta$ path,}\\
        \false, & \text{otherwise.}
    \end{cases}$
\end{itemize}

The $\dirskel$ construct defined for directed graphs naturally extends to bidirected graphs as follows.
For all $\alpha,\beta\in\signs$ let $B_{\alpha\beta} = \{ \{s_i\alpha,t_i\beta\} : \text{$\expansion(e_i)$ has an $s_i\alpha$-$t_i\beta$ path, for $i=1,\dots,k$} \}$. The \emph{bidirected skeleton} of $\mu$ is the graph $\dirskel(\mu) := (V(\skel(\mu)),B_{++} \cup B_{+-} \cup B_{-+} \cup B_{--})$.

We are now ready to proceed with the description of the phases. Some correctness details are omitted since they are just a straightforward adaptation of the proofs described in \Cref{subsec:superbubble-finding-algorithm}.

\paragraph*{Phase 1.}

Recall that $\nu$ is the parent of $\mu$ in $T$. Let $\{u,v\}$ denote the endpoints of $e_\nu$ and $e_\mu$.
Phase 1 is also a DFS starting at the root of $T$ and updates the states pointing downwards from the root.
The update of $\noextremity{\nu}{\mu}$ is identical to that for superbubbles.
The update of $\acyclic{\nu}{\mu}$ is also identical to that for superbubbles up to the point where graph $K:=\dirskel(\mu)-\{u+,v+\}-\{u+,v-\}-\{u-,v+\}-\{u-,v-\}$ can be built.
So at this point we have that $\noextremity{\nu}{\mu}$ is $\true$ and the states pointing from $\mu$ to its children all have the acyclicity state set to $\true$.
Clearly, as for directed graphs, $\acyclic{\nu}{\mu}$ is $\true$ if $K$ is acyclic and is $\false$ otherwise. Our aim is to use DFS in order to decide acyclicity, but before that we make some observations regarding bidirected graphs and the structure of $K$.

First notice that for any two vertices $x,y \in V(K)$ there is at most one edge with endpoints $\{x,y\}$ for otherwise $K$ has a cycle. This is justified in the next result (essentially, the obvious fact that a directed graph contains a cycle whenever two of its vertices reach each other generalizes to bidirected graphs).

\begin{proposition}\label{prop:reachability-bidirected-cycloid}
    Let $G$ be a bidirected graph, $u,v \in V(G)$ be vertices, and $\alpha,\beta \in \signs$.
    If $G$ has a $u\alpha$-$v\beta$ and a $u\hat\alpha$-$v$ bidirected path then $G$ has a cycloid.
\end{proposition}
\begin{proof}
    Let $x \neq u$ be the closest vertex to $u$ where these two paths intersect (possibly $x=v$). Let $p_1$ denote the first path and $p_2$ the second.
    The subpath of $p_1$ from $u\alpha$ up to $x$ concatenated with the reversed subpath of $p_2$ from $u\hat\alpha$ up to $x$ forms a cycloid: if the vertex-side at $x$ in the subpath of $p_1$ has a different sign than that at $x$ in the subpath of $p_2$ then we have a cycloid with alternation at every vertex, and otherwise we have a cycloid where only $x$ does not respect alternation.
\end{proof}

Furthermore, notice that if $K$ has at most one tip then it has a cycloid (this is clear if $G$ is directed since every directed acyclic graph contains a source and a sink, i.e., it contains two tips).
See~\citep{Harviainen2026.03.28.714704} for a proof of the next result.

\begin{proposition}\label{prop:atm1-tip-cycloid}
    Let $G$ be a bidirected graph having at least two vertices. If $G$ has at most one tip then $G$ has a cycloid.
\end{proposition}

Suppose now that $K$ has at least three tips, so pick $x$ to be a tip which moreover is distinct from the vertices $u$ and $v$.
Then the vertex-sides contained in the expansions of the edges incident to $x$ in $\skel(\mu)$ all have sign $\alpha$ for some $\alpha \in \signs$.
Suppose otherwise and let $e = \{x,y\} \in \skel(\mu)$ be an edge incident to $x$ (possibly $y=u$ or $y=v$).
Then $\expansion(e)$ has vertex-sides of opposite signs at $x$. Since $\expansion(e)$ has no extremities except possibly $x$ and $y$, $\expansion(e)$ has at most one tip, which is $y$, and thus it has a cycloid by \Cref{prop:atm1-tip-cycloid}, a contradiction (recall that the fact that we build $K$ implies in particular that $\expansion(e)$ is acyclic).
Therefore $\noextremity{\nu}{\mu}$ is $\false$ because $x$ is a tip in $\expansion(e_\mu)$, a contradiction.

Therefore $K$ has exactly two tips, which are $u$ and $v$. Under these conditions deciding if $K$ is acyclic is a simple (linear-time) task and we refer to~\citep{Harviainen2026.03.28.714704} for the technical details.
Essentially, we can run a DFS starting at $u$ such that when arriving at an unvisited vertex $z$ with a vertex-side $z\gamma$, the DFS prioritizes expanding to vertices in $N^{\hat\gamma}_G(z)$ and only then considers those in $N^{\gamma}_G(z)$. Moreover, whenever an edge $\{x\alpha,y\alpha\}$ is scanned, for some $\alpha \in \signs$, if $y$ is unvisited then we can ``flip'' $y$ in order to make this edge have vertex-sides of opposite signs, i.e., this bidirected edge becomes essentially a directed edge, and if $y$ is already visited then it is possible to find a cycloid\footnote{In fact, a cycloid with the ``exceptional'' vertex.} in the current graph induced by the edges which have been scanned so far; if the DFS halts without ever finding such an edge then it produces a bidirected graph where every edge has vertex-sides of opposite signs, i.e., a directed graph, in which case a standard DFS suffices to finally decide if $K$ is acyclic. The correctness of ``flipping'' the vertex-sides of the vertices during the first DFS is ensured by the next result (\Cref{prop:flipping-iff-walks}), i.e., flipping vertices during the first DFS preserves cycloids.

\begin{proposition}
\label{prop:flipping-iff-walks}
    Let $G$ be a bidirected graph, let $u \in V(G)$ be a vertex, and let $G'$ be the graph obtained from $G$ by changing the positive vertex-sides at $u$ into negative vertex-sides and the negative into positive (i.e., \emph{flipping} $u$). Let $W$ be a sequence of edges of $G$. Then $W$ is a bidirected walk in $G$ if and only if $W$ is a bidirected walk in $G'$.
\end{proposition}

As for the reachability states, if $u$ is a tip with sign $\gamma$ and $v$ is a tip with sign $\delta$ in $K$, then clearly $\mathsf{Reaches_{uv\gamma\delta}}$ is $\true$ and the remaining three reachability states are $\false$. The fact that at most one state is true is clear from \Cref{prop:reachability-bidirected-cycloid} (since otherwise there is a cycle and by definition these states are $\Null$ and we are done).
To see why indeed there is one state set to true, consider a maximal bidirected path in $K$ and observe that this path starts at $u\gamma$ and ends at $v\delta$ (similarly as to why a maximal path in an acyclic graph with unique source and unique sink starts at the source and ends at the sink). As usual, this path in $K$ can be mapped to a path in $X$.

\paragraph*{Phase 2.}

This phase is also a BFS starting at the root of the tree and is essentially identical to phase two of superbubbles.
The only difference is in the computation of the acyclicity states, which we describe next. Let $K:=\dirskel(\nu)$, let $\mu_1,\dots,\mu_k$ denote the children of $\nu$, and let $\mu_0$ denote the parent of $\nu$ (if $\nu$ is the root of $T$ then $\mu_0$ can be ignored).

Recall that at this point the acyclicity and absence-of-extremities states leaving $\nu$ to $\mu_i$ for $\iinkz$ are all set to $\true$, otherwise the states $\acyclic{\mu_i}{\nu}$ for $\iink$ are updated in the obvious way (see how to update the states in this case in the description of phase two for superbubbles). So we are indeed in conditions of building $K$.
Similarly as for superbubbles, if $K$ is acyclic then $\acyclic{\mu_i}{\nu}$ is $\true$ for each $\iink$ (to decide if $K$ is acyclic we can proceed identically as described above in phase one).
Otherwise $K$ has a cycloid and in order to maintain our algorithm linear-time we can compute its feedback edges.\footnote{For clarity, a feedback edge in a bidirected graph is an edge intersecting every cycloid of the graph.} Depending on which edges are indeed feedback edges we can update the acyclicity states accordingly, like we did for superbubbles.
One key observation about $K$ is that it contains no tips. Suppose otherwise.
If $K$ has a tip $x$ but $x$ is not a tip in $\expansion(e)$ for some edge $e \in E(\skel(\nu))$ then $\expansion(e)$ has vertex-sides of opposite signs at $x$ and therefore $\expansion(e)$ has at most one tip, which is the other endpoint of $e$. But then \Cref{prop:atm1-tip-cycloid} gives a cycloid in $\expansion(e)$, a contradiction since every acyclicity state pointing away from $\nu$ is to $\true$. So indeed $x$ is a tip for some absence-of-extremity state pointing towards $\nu$, and thus that state is $\false$, a contradiction.

In conclusion, it is enough to devise a linear-time algorithm for finding feedback edges in tipless bidirected graphs. Such an algorithm is presented in~\Cref{subsec:feedback-bidirected} and it is followed by a hardness result for the same problem in general bidirected graphs.

\paragraph*{Phase 3.}

This phase is completely analogous to phase three of the superbubble finding algorithm. The only change required to the algorithm is to remove the condition of the ``back-edge'', i.e., in a candidate superbubble $st$ we must ensure that $ts$ is not an edge, while for ultrabubbles the existence of that edge is allowed: clearly, for a separable pair of vertex-sides $\{u\alpha,v\beta\}$ in a graph $G$, if $\{u\hat\alpha,v\hat\beta\} \in E(G)$ then splitting $u\alpha$ and $v\beta$ results in a graph where the component containing $u$ and $v$ does not contain the edge $\{u'\hat\alpha,v'\hat\beta\}$. The remaining parts of the algorithm are completely identical and thus we obtain the next theorem. We remark that part as to why it is straightforward to get the ultrabubble algorithm from the superbubble algorithm has to do with the fact that ultrabubbles are, essentially, ``weak'' superbubbles (see~\citep{Harviainen2026.03.28.714704} for further results and intuition on this direction).

\begin{theorem}
    The ultrabubbles of a bidirected graph $G$ can be computed in time $O(|V(G)| + |E(G)|)$.
\end{theorem}

\subsection{Feedback edges in bidirected graphs}
\label{subsec:feedback-bidirected}

\paragraph{Feedback edges in tipless bidirected graphs}

We present a linear-time algorithm computing every feedback edge of a bidirected graph $G$ containing no tips. Recall that a feedback edge in a bidirected graph is an edge contained in every cycloid of the graph.

The algorithm constructs $G$ by ear additions. If it succeeds then the problem reduces to that of a directed graph, where known linear-time algorithms to find feedback edges can be used.
Otherwise the procedure finds an obstruction and can correctly output that $G$ has no feedback edges.
We begin with a simple observation.

\begin{lemma}
\label{lem:tipless-blossom-no-good-edges}
    Let $G$ be a bidirected graph without tips. If $G$ has a cycloid with exceptional vertex then $G$ has no feedback edges.
\end{lemma}
\begin{proof}
    Let $B$ be a cycloid in $G$ with exceptional vertex $x$ having sign $\alpha \in\signs$.
    Since $G$ has no tips, there is an edge $\{x\hat\alpha,u\gamma\} \in E(G)$ and let $p$ be a bidirected path consisting of that edge alone. Also because $G$ has no tips, we can greedily extend this bidirected path from $u$.
    During the process either we visit a vertex already in $p$, in which case we have found a cycloid edge-disjoint from $B$ and thus $G$ has no feedback edges, or we hit a vertex of $B$ for which we give the following argument.
    Notice that $p$ partitions $B$ into two subpaths. Moreover, notice that removing an edge of $B$ from either of these subpaths leaves $G$ with a cycloid via $\{x\hat\alpha,u\gamma\}$ and the untouched path (notice that this cycloid possibly has an exceptional vertex), but since any feedback edge is contained in $B$ we can conclude that $G$ has no feedback edges.
\end{proof}

We need two more definitions before describing the algorithm. Let $H \subset G$ be a nonempty graph. Say that an \emph{ear} is a path of $G$ whose first and last vertex (called \emph{attachment} vertices) are distinct and belong to $H$ and every other vertex on the path does not belong to $H$. Say that $H$ is \emph{digraphic} if every edge of $G$ has opposite signs in its vertex-sides.

Begin by applying \Cref{prop:atm1-tip-cycloid} to get a cycloid $C$. We can assume to be provided a cycloid without exception, for otherwise $G$ has no feedback edge by \Cref{lem:tipless-blossom-no-good-edges}.

We maintain the invariant that the graph is digraphic and is strongly connected (in the directed sense). Put $H_0=C$. Graph $H_0$ is digraphic (i.e., every edge has vertex-sides of opposite signs). To ensure it is strongly connected, it is not hard to see that it suffices for that effect to invert some of its vertices (which is correct due to \Cref{prop:flipping-iff-walks}). So $H_0$ respects the invariant.
We show that if $H_i$ respects the invariant then successfully adding an ear results in a graph $H_{i+1}$ also respecting the invariant, and if not, then we either found a cycloid with exception or an edge-disjoint cycle from $C$.
When no more ears can be added then we have recovered a graph equivalent to $G$ in the sense of \Cref{prop:flipping-iff-walks}. We remark that we only resign vertices of the newly added ears except for its attachment vertices as those are already in the current graph $H_i$.

Suppose that $V(H_i) \subset V(G)$ and let $x \in V(G) \setminus V(H_i)$. Since $G$ is tipless, vertex $x$ is not a tip. Then build a path $p^+$ starting with a vertex-side $x+$ until it hits a vertex of $H_i$ or a vertex previously on the path. In the latter case we halt, because we have found a cycloid disjoint from $C$. So this path hits $H_i$ first in a vertex $a$.
Similarly, build a path $p^-$ starting with a vertex-side $x-$ until it hits either a vertex in $V(H_i)$, a vertex of $p^+$, or a vertex of $p^-$.
If one of the last two cases occurs then we have found cycloid disjoint from $C$. So $p^-$ hits $H_i$ first in a vertex $b$. If $b=a$ then we have found a cycloid disjoint from $C$, so $b\neq a$.
Notice that the concatenation of $p^+$ and $p^-$ gives an ear with attachment vertices $a$ and $b$.
Let $\alpha,\beta \in\signs$ denote the signs of the vertex-sides of the ear in the attachment vertices.
Suppose that $\alpha=\beta$.
The graph $H_i$ has an $\alpha$-$\hat\beta$ path\footnote{$H_i$ also has a $\hat\alpha$-$\beta$ path, and thus $H_{i+1}$ also has a cycloid with exceptional vertex $b$.} since it is strongly connected. This path together with the new ear creates a cycloid with exceptional vertex $a$ in $H_{i+1}$, and so we can halt due to \Cref{lem:tipless-blossom-no-good-edges}.
Otherwise $\alpha \neq \beta$.
It is trivial to flip each vertex of the ear\footnote{Without loss of generality suppose that $\alpha=+$. We can start at $a+$ and move to the consecutive vertex in the $a+$-$b-$ path while greedily flipping vertices in order to make every edge directed. Essentially, the first edge has a plus and a minus, the second too, and so on and so forth, until we reach the last vertex, which must $b$ and moreover is contained in the vertex-side $b-$.} so that the every edge has vertex-sides of opposite signs, which makes $H_{i+1}$ digraphic.

We are left to argue that $H_{i+1}$ is strongly connected. Assume without loss of generality that $\alpha=+$.
Notice that $H_i$ has $+-$ and $-+$ paths for any two vertices since it is strongly connected.
Let $u$ be a vertex contained in the ear distinct from $a$ and $b$ and let $v \in V(H_i)$ (clearly, any two vertices in the ear are strongly connected in $H_{i+1}$).
We show that $H_{i+1}$ has $u+$-$v-$ and $v+$-$u-$ paths, thus showing that $H_{i+1}$ is strongly connected.
Since $a \in V(H_i)$, $H_i$ has a $v+$-$a-$ path, which prepended with the $a+$-$u-$ subpath of the ear gives a $v+$-$u-$ path in $H_{i+1}$. (Notice that the $a+$-$u-$ exists by construction of ear). Similarly, since $b \in V(H_i)$, $H_i$ has a $b+$-$v-$ path, which prepended with the $u+$-$b-$ subpath of the ear gives a $u+$-$v-$ path in $H_{i+1}$.

Suppose now that $V(H_i)=V(G)$ and $E(H_i) \subset E(G)$. Let $e \in E(G) \setminus E(H)$. If the vertex-sides of $e$ have the same sign then we have found a cycloid with exception (similarly to the case where the ear had vertex-sides of the same sign in the attachment vertices, i.e., when $\alpha=\beta$). Otherwise $H_{i+1}$ is digraphic since $H_i$ is digraphic, $e$ has vertex-sides of opposite signs, and $H_{i+1}=H_i+e$. Further, $H_{i+1}$ is strongly connected because $H_i$ is strongly connected and $V(H_{i+1})=V(H_i)$.
Finally we have $H_{i+1}=G$. Due to the invariant, $G$ is digraphic and is strongly connected.


In conclusion, if the ear addition procedure succeeds then we know that $G$ is essentially a strongly connected directed graph, where linear-time algorithms to find feedback edges are known. If it fails then we know that the graph has no feedback edges.
We are left to examine the running time of the ear-addition construction.

\begin{theorem}\label{thm:feedback-bidirected-linear}
    There is an algorithm that finds all the feedback edges of a tipless bidirected graph in time $O( |V(G)| + |E(G)| )$.
\end{theorem}
\begin{proof}
    We argue that the ear addition procedure takes $O(|V(G)| + |E(G)|)$ time.
    
    Finding $H_0$ and the ears takes linear time: we can greedily extend the walk (e.g., by always taking an arbitrary edge incident to the current vertex being extended) until we hit a relevant vertex for the construction.
    Flipping the vertices of the ears takes linear time in the size of the ear, and since every vertex is flipped at most once and the ears partition $E(G)$, the overall time taken to flip the vertices during the construction is $\Theta(|V(G)| + |E(G)|)$.
    The checks on the vertex-sides at the attachment vertices when adding ears take constant time.
\end{proof}

\paragraph{Hardness of computing feedback edges}

Recall the \textsc{Triangle} problem where one is given an undirected graph $G' = (V', E')$ and asked whether it contains a cycle of length $3$, that is, a triangle. The $k$-\textsc{Clique Conjecture} (see Conjecture 10 of \cite{Kunnemann24}) asserts in particular that a triangle cannot be found in time $O(n^{\omega - \epsilon})$ for matrix multiplication exponent~$\omega$ and any $\epsilon > 0$. We argue that under the $k$-\textsc{Clique Conjecture}, one cannot decide whether a bidirected graph has \emph{bidirected feedback edge}, i.e., whether it can be made acyclic by removing a single edge in time $O(n^{\omega - \epsilon})$ for any $\epsilon > 0$.

\begin{theorem}
    Under the $k$-\textsc{Clique Conjecture}, the existence of a bidirected feedback edge cannot be decided in time $O(n^{\omega - \epsilon})$ for any $\epsilon > 0$.
\end{theorem}
\begin{proof}
    Take an arbitrary instance $G_T = (V_T, E_T)$ of \textsc{Triangle}. By a standard color-coding-like reduction, we can reduce \textsc{Triangle} to the \textsc{Tripartite Triangle} problem \citep{Fellows09} where we look for a triangle in an undirected tripartite graph $G_3 = (A, B, C, E_3)$ with tripartite sets $A$, $B$, and $C$ such that both endpoints of all edges in $E_3$ come from distinct sets. The reduction multiplies the number of vertices by $3$ and the number of edges by $6$, so solving \textsc{Tripartite Triangle} is as hard as hard solving \textsc{Triangle}.

    We will next show a reduction from \textsc{Tripartite Triangle} to finding a bidirected feedback edge in a bidirected graph $G = (V, E)$. Let $G_3 = (A, B, C, E_3)$ be an instance of \textsc{Tripartite Triangle}. Construct a bidirected graph on the vertex set $A \cup B \cup C \cup \{x, y, z\}$ with auxiliary vertices $x$, $y$, and $z$. For every edge $\{u, v\} \in E_3$,
\begin{itemize}[nosep]
    \item[(1)] if $u \in A$ and $v \in B$, add the edge $\{u-, v+\}$;
    \item[(2)] if $u \in A$ and $v \in C$, add the edge $\{u+, v-\}$; and
    \item[(3)] if $u \in B$ and $v \in C$, add the edge $\{u-, v-\}$.
\end{itemize}
Finally, add bidirected edges $\{x+, y-\}$, $\{y+, z-\}$, and $\{z+, x-\}$. We claim that $G$ has a bidirected feedback edge if and only if $G_3$ does \emph{not} have a triangle.

Suppose $G_3$ has a triangle $\{a, b, c\}$ with $a \in A$, $b \in B$, and $c \in C$. Then, there are two disjoint cycles with the edges $\{x+, y-\}$, $\{y+, z-\}$, $\{z+, x-\}$ and $\{a-, b+\}$, $\{b-, c-\}$, $\{c-, a+\}$ in $G$, and no bidirected feedback edge can exist.

Suppose now instead that no triangle exists in $G_3$. We need to show that there are no other cycles in $G$ than the one with the edges $\{x+, y-\}$, $\{y+, z-\}$, $\{z+, x-\}$.
Consider any alternating closed walk $W$ in $G$ with edges $\{v_1\alpha_1, v_2\hat{\alpha_2}\}, \{v_2\alpha_2, v_3\hat{\alpha_3}\}, \dots, \{v_\ell\alpha_\ell, v_1\alpha_1'\}$ with $v_1, v_2, \dots, v_\ell \in A \cup B \cup C$ and~$\alpha_1, \alpha_2, \dots, \alpha_\ell, \alpha_1' \in \signs$. Note that $v_2, v_3, \dots, v_\ell \not\in C$, since for all $v \in C$ all vertex-sides are of the form $v-$. Because of the construction, it also cannot be that $v_{i}, v_{i+2\,\text{mod}\,\ell} \in X$ and $v_{i+1\,\text{mod}\,\ell} \in Y$ for any $i \in \{1, \dots, \ell\}$ with $X, Y \in \{A, B, C\}$.

Therefore, we must have that $\ell = 3$, $v_1 \in C$, and either $v_2 \in A$, $v_3 \in B$ or $v_2 \in B$, $v_3 \in A$. In this case, we have a tripartite triangle $\{v_1, v_2, v_3\}$, resulting in a contradiction. Hence, $G$ cannot contain any other cycles and thus, for example, $\{x+, y-\}$ is a bidirected feedback edge.
\end{proof}

By omitting the edges $\{x+, y-\}$, $\{y+, z-\}$, and $\{z+, x-\}$ from $G$, we immediately get the following corollary.
\begin{corollary}
     Under the $k$-\textsc{Clique Conjecture}, the acyclicity of a bidirected graph cannot be decided in time $O(n^{\omega - \epsilon})$ for any $\epsilon > 0$.
\end{corollary}

\section*{Acknowledgements}
\addcontentsline{toc}{section}{Acknowledgements}

We are grateful to Benedict Paten for very helpful explanations and clarifications on snarls.

Co-funded by the European Union (ERC, SCALEBIO, 101169716). Views and opinions expressed are however those of the author(s) only and do not necessarily reflect those of the European Union or the European Research Council. Neither the European Union nor the granting authority can be held responsible for them. Co-funded by the Research Council of Finland, Grant 1358744. Juha Harviainen was supported by the Research Council of Finland, Grant 351156. \\[0.2cm]
\includegraphics[width=4cm]{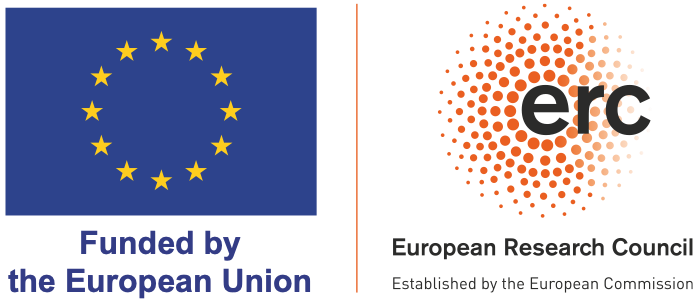}

\bibliography{bibliography}

@article{hopcroft1973dividing,
author = {Hopcroft, J. E. and Tarjan, R. E.},
title = {Dividing a Graph into Triconnected Components},
journal = {SIAM Journal on Computing},
volume = {2},
number = {3},
pages = {135-158},
year = {1973},
doi = {10.1137/0202012}
}

@article{Hopcroft73blockcut,
  author       = {John E. Hopcroft and
                  Robert Endre Tarjan},
  title        = {Efficient Algorithms for Graph Manipulation {[H]} (Algorithm 447)},
  journal      = {Commun. {ACM}},
  volume       = {16},
  number       = {6},
  pages        = {372--378},
  year         = {1973}
}

@article{rahman2022assembler,
  title={Assembler artifacts include misassembly because of unsafe unitigs and underassembly because of bidirected graphs},
  author={Rahman, Amatur and Medvedev, Paul},
  journal={Genome Research},
  volume={32},
  number={9},
  pages={1746--1753},
  year={2022},
  publisher={Cold Spring Harbor Lab}
}

@inproceedings{gabow1983efficient,
  title={An efficient reduction technique for degree-constrained subgraph and bidirected network flow problems},
  author={Gabow, Harold N},
  booktitle={Proc. 15th Annual ACM Symposium on Theory of Computing},
  pages={448--456},
  year={1983}
}

@article {billi,
	author = {Bhat, Shreeharsha G and Mahajan, Daanish and Jain, Chirag},
	title = {Billi: Provably Accurate and Scalable Bubble Detection in Pangenome Graphs},
	elocation-id = {2025.11.21.689636},
	year = {2025},
	publisher = {Cold Spring Harbor Laboratory},
	URL = {https://doi.org/10.1101/2025.11.21.689636},
	eprint = {https://www.biorxiv.org/content/early/2025/11/22/2025.11.21.689636.full.pdf},
	journal = {bioRxiv}
}

@article{hprc,
	author = {Liao, Wen-Wei and others},
	date = {2023/05/01},
	id = {Liao2023},
	isbn = {1476-4687},
	journal = {Nature},
	number = {7960},
	pages = {312--324},
	title = {A draft human pangenome reference},
	volume = {617},
	year = {2023}}

@InProceedings{gutwenger2001linear,
author="Gutwenger, Carsten
and Mutzel, Petra",
editor="Marks, Joe",
title="A Linear Time Implementation of {SPQR}-Trees",
booktitle="Graph Drawing",
year="2001",
publisher="Springer Berlin Heidelberg",
address="Berlin, Heidelberg",
pages="77--90",
isbn="978-3-540-44541-8",
doi="https://doi.org/10.1007/3-540-44541-2_8"
}

@article{bienstock1988complexity,
author = {Bienstock, Daniel and Monma, Clyde L.},
title = {On the Complexity of Covering Vertices by Faces in a Planar Graph},
journal = {SIAM Journal on Computing},
volume = {17},
number = {1},
pages = {53-76},
year = {1988},
doi = {10.1137/0217004}
}

@InProceedings{battista1990on-line,
author="Di Battista, Giuseppe
and Tamassia, Roberto",
editor="Paterson, Michael S.",
title="On-line graph algorithms with {SPQR}-trees",
booktitle="Automata, Languages and Programming",
year="1990",
publisher="Springer Berlin Heidelberg",
address="Berlin, Heidelberg",
pages="598--611",
isbn="978-3-540-47159-2"
}

@article{battista1996on-line,
author = {Di Battista, Giuseppe and Tamassia, Roberto},
title = {On-Line Planarity Testing},
journal = {SIAM Journal on Computing},
volume = {25},
number = {5},
pages = {956-997},
year = {1996},
doi = {10.1137/S0097539794280736}
}

@article{lane1937structural,
author = {Saunders Mac Lane},
title = {{A structural characterization of planar combinatorial graphs}},
volume = {3},
journal = {Duke Mathematical Journal},
number = {3},
publisher = {Duke University Press},
pages = {460 -- 472},
year = {1937},
doi = {10.1215/S0012-7094-37-00336-3},
URL = {https://doi.org/10.1215/S0012-7094-37-00336-3}
}

@inproceedings{onodera2013detecting,
  title={Detecting superbubbles in assembly graphs},
  author={Onodera, Taku and Sadakane, Kunihiko and Shibuya, Tetsuo},
  booktitle={Algorithms in Bioinformatics: 13th International Workshop, WABI 2013, Sophia Antipolis, France, September 2-4, 2013. Proceedings 13},
  pages={338--348},
  year={2013},
  organization={Springer}
}

@article{brankovic2016linear,
  title={Linear-time superbubble identification algorithm for genome assembly},
  author={Brankovic, Ljiljana and Iliopoulos, Costas S and Kundu, Ritu and Mohamed, Manal and Pissis, Solon P and Vayani, Fatima},
  journal={Theoretical Computer Science},
  volume={609},
  pages={374--383},
  year={2016},
  publisher={Elsevier}
}

@article{dabbaghie2022bubblegun,
    author = {Dabbaghie, Fawaz and Ebler, Jana and Marschall, Tobias},
    title = {BubbleGun: enumerating bubbles and superbubbles in genome graphs},
    journal = {Bioinformatics},
    volume = {38},
    number = {17},
    pages = {4217-4219},
    year = {2022},
    month = {07},
    issn = {1367-4803},
    doi = {10.1093/bioinformatics/btac448},
    url = {https://doi.org/10.1093/bioinformatics/btac448},
    eprint = {https://academic.oup.com/bioinformatics/article-pdf/38/17/4217/49889707/btac448.pdf},
}

@article{shafin2020nanopore,
  title={Nanopore sequencing and the Shasta toolkit enable efficient de novo assembly of eleven human genomes},
  author={Shafin, Kishwar and Pesout, Trevor and Lorig-Roach, Ryan and Haukness, Marina and Olsen, Hugh E and Bosworth, Colleen and Armstrong, Joel and Tigyi, Kristof and Maurer, Nicholas and Koren, Sergey and others},
  journal={Nature biotechnology},
  volume={38},
  number={9},
  pages={1044--1053},
  year={2020},
  publisher={Nature Publishing Group US New York}
}

@article{Even98,
  author       = {Guy Even and
                  Joseph Naor and
                  Baruch Schieber and
                  Madhu Sudan},
  title        = {Approximating Minimum Feedback Sets and Multicuts in Directed Graphs},
  journal      = {Algorithmica},
  volume       = {20},
  number       = {2},
  pages        = {151--174},
  year         = {1998}
}

@article{garg2018graph,
  title={A graph-based approach to diploid genome assembly},
  author={Garg, Shilpa and Rautiainen, Mikko and Novak, Adam M and Garrison, Erik and Durbin, Richard and Marschall, Tobias},
  journal={Bioinformatics},
  volume={34},
  number={13},
  pages={i105--i114},
  year={2018},
  publisher={Oxford University Press}
}

@article{garrison2018variation,
  title={Variation graph toolkit improves read mapping by representing genetic variation in the reference},
  author={Garrison, Erik and Sir{\'e}n, Jouni and Novak, Adam M and Hickey, Glenn and Eizenga, Jordan M and Dawson, Eric T and Jones, William and Garg, Shilpa and Markello, Charles and Lin, Michael F and others},
  journal={Nature biotechnology},
  volume={36},
  number={9},
  pages={875--879},
  year={2018},
  publisher={Nature Publishing Group US New York}
}

@article{paten2018ultrabubbles,
    author = {Paten, Benedict and Eizenga, Jordan M. and Rosen, Yohei M. and Novak, Adam M. and Garrison, Erik and Hickey, Glenn},
    title = {Superbubbles, Ultrabubbles, and Cacti},
    journal = {Journal of Computational Biology},
    volume = {25},
    number = {7},
    pages = {649-663},
    year = {2018},
    doi = {10.1089/cmb.2017.0251},
    note = {PMID: 29461862}
}

@article{di1996line,
  title={On-line maintenance of triconnected components with SPQR-trees},
  author={Di Battista, Giuseppe and Tamassia, Roberto},
  journal={Algorithmica},
  volume={15},
  number={4},
  pages={302--318},
  year={1996},
  publisher={Springer}
}

@book{diestel,
  title={Graph theory},
  author={Diestel, Reinhard},
  volume={173},
  year={2025},
  publisher={Springer Nature}

}

@article{kolmogorov2020metaflye,
  title={metaFlye: scalable long-read metagenome assembly using repeat graphs},
  author={Kolmogorov, Mikhail and Bickhart, Derek M and Behsaz, Bahar and Gurevich, Alexey and Rayko, Mikhail and Shin, Sung Bong and Kuhn, Kristen and Yuan, Jeffrey and Polevikov, Evgeny and Smith, Timothy PL and others},
  journal={Nature methods},
  volume={17},
  number={11},
  pages={1103--1110},
  year={2020},
  publisher={Nature Publishing Group US New York}
}

@article{mwaniki2024popping,
  title={Popping Bubbles in Pangenome Graphs},
  author={Mwaniki, Njagi and Garrison, Erik and Pisanti, Nadia},
  journal={arXiv preprint arXiv:2410.20932},
  year={2024}
}

@article{li2024exploring,
  title={Exploring gene content with pangene graphs},
  author={Li, Heng and Marin, Maximillian and Farhat, Maha R},
  journal={Bioinformatics},
  volume={40},
  number={7},
  pages={btae456},
  year={2024},
  publisher={Oxford University Press}
}

@ARTICLE{loglinear,
  author={Sung, Wing-Kin and Sadakane, Kunihiko and Shibuya, Tetsuo and Belorkar, Abha and Pyrogova, Iana},
  journal={IEEE/ACM Transactions on Computational Biology and Bioinformatics}, 
  title={An $O(m \log m)$-Time Algorithm for Detecting Superbubbles}, 
  year={2015},
  volume={12},
  number={4},
  pages={770-777},
  doi={10.1109/TCBB.2014.2385696}}

@inproceedings{medvedev2007computability,
  title={Computability of models for sequence assembly},
  author={Medvedev, Paul and Georgiou, Konstantinos and Myers, Gene and Brudno, Michael},
  booktitle={International workshop on algorithms in bioinformatics},
  pages={289--301},
  year={2007},
  organization={Springer}
}

@article{gartner-revisited,
  title={Superbubbles revisited},
  author={G{\"a}rtner, Fabian and M{\"u}ller, Lydia and Stadler, Peter F},
  journal={Algorithms for Molecular Biology},
  volume={13},
  number={1},
  pages={16},
  year={2018},
  publisher={Springer}
}

@misc{metagenomescope,
  title        = {MetagenomeScope},
  author={Fedarko, Marcus},
  howpublished = {\url{https://github.com/fedarko/MetagenomeScope}},
  year         = {2017},
  note         = {Accessed on Nov 1, 2025}
}

@article{garey1978linear,
  title={A linear-time algorithm for finding all feedback vertices},
  author={Garey, Michael R and Tarjan, Robert E},
  journal={Information Processing Letters},
  volume={7},
  number={6},
  pages={274--276},
  year={1978},
  publisher={Elsevier}
}

@article{gartner2019direct,
  title={Direct superbubble detection},
  author={G{\"a}rtner, Fabian and Stadler, Peter F},
  journal={Algorithms},
  volume={12},
  number={4},
  pages={81},
  year={2019},
  publisher={MDPI}
}

@InProceedings{rotenberg,
  author =	{Holm, Jacob and Italiano, Giuseppe F. and Karczmarz, Adam and Lacki, Jakub and Rotenberg, Eva},
  title =	{{Decremental SPQR-trees for Planar Graphs}},
  booktitle =	{26th Annual European Symposium on Algorithms (ESA 2018)},
  pages =	{46:1--46:16},
  series =	{Leibniz International Proceedings in Informatics (LIPIcs)},
  ISBN =	{978-3-95977-081-1},
  ISSN =	{1868-8969},
  year =	{2018},
  volume =	{112},
  editor =	{Azar, Yossi and Bast, Hannah and Herman, Grzegorz},
  publisher =	{Schloss Dagstuhl -- Leibniz-Zentrum f{\"u}r Informatik},
  address =	{Dagstuhl, Germany},
  URL =		{https://drops.dagstuhl.de/entities/document/10.4230/LIPIcs.ESA.2018.46},
  URN =		{urn:nbn:de:0030-drops-95091},
  doi =		{10.4230/LIPIcs.ESA.2018.46}
}

@InProceedings{onlinegraphalgorithms,
author="Di Battista, Giuseppe
and Tamassia, Roberto",
editor="Paterson, Michael S.",
title="On-line graph algorithms with SPQR-trees",
booktitle="Automata, Languages and Programming",
year="1990",
publisher="Springer Berlin Heidelberg",
address="Berlin, Heidelberg",
pages="598--611",
isbn="978-3-540-47159-2"
}

@article{minkin2020scalable,
  title={Scalable multiple whole-genome alignment and locally collinear block construction with SibeliaZ},
  author={Minkin, Ilia and Medvedev, Paul},
  journal={Nature communications},
  volume={11},
  number={1},
  pages={6327},
  year={2020},
  publisher={Nature Publishing Group UK London}
}

@article {Jafarzadeh2025.07.05.662656,
	author = {Jafarzadeh, Nafiseh and Eizenga, Jordan M and Paten, Benedict},
	title = {An Efficient Graph Algorithm for Diploid Local Ancestry Inference},
	elocation-id = {2025.07.05.662656},
	year = {2025},
	doi = {10.1101/2025.07.05.662656},
	publisher = {Cold Spring Harbor Laboratory},
	URL = {https://www.biorxiv.org/content/early/2025/07/09/2025.07.05.662656},
	eprint = {https://www.biorxiv.org/content/early/2025/07/09/2025.07.05.662656.full.pdf},
	journal = {bioRxiv}
}

@article{paten2011cactus,
  title={Cactus graphs for genome comparisons},
  author={Paten, Benedict and Diekhans, Mark and Earl, Dent and John, John St and Ma, Jian and Suh, Bernard and Haussler, David},
  journal={Journal of Computational Biology},
  volume={18},
  number={3},
  pages={469--481},
  year={2011},
  publisher={Mary Ann Liebert, Inc. 140 Huguenot Street, 3rd Floor New Rochelle, NY 10801 USA}
}

@article{wang2025population,
  title={Population-level structural variant characterization from pangenome graph},
  author={Wang, Songbo and Xu, Tun and Zhang, Pengyu and Ye, Kai},
  journal={bioRxiv},
  pages={2025--07},
  year={2025},
  publisher={Cold Spring Harbor Laboratory}
}

@article{siren2024personalized,
  title={Personalized pangenome references},
  author={Sir{\'e}n, Jouni and Eskandar, Parsa and Ungaro, Matteo Tommaso and Hickey, Glenn and Eizenga, Jordan M and Novak, Adam M and Chang, Xian and Chang, Pi-Chuan and Kolmogorov, Mikhail and Carroll, Andrew and others},
  journal={Nature Methods},
  volume={21},
  number={11},
  pages={2017--2023},
  year={2024},
  publisher={Nature Publishing Group US New York}
}

@article{hickey2024pangenome,
  title={Pangenome graph construction from genome alignments with Minigraph-Cactus},
  author={Hickey, Glenn and Monlong, Jean and Ebler, Jana and Novak, Adam M and Eizenga, Jordan M and Gao, Yan and Marschall, Tobias and Li, Heng and Paten, Benedict},
  journal={Nature biotechnology},
  volume={42},
  number={4},
  pages={663--673},
  year={2024},
  publisher={Nature Publishing Group US New York}
}

@article{chang2020distance,
  title={Distance indexing and seed clustering in sequence graphs},
  author={Chang, Xian and Eizenga, Jordan and Novak, Adam M and Sir{\'e}n, Jouni and Paten, Benedict},
  journal={Bioinformatics},
  volume={36},
  number={Supplement\_1},
  pages={i146--i153},
  year={2020},
  publisher={Oxford University Press}
}

@article{siren2021pangenomics,
  title={Pangenomics enables genotyping of known structural variants in 5202 diverse genomes},
  author={Sir{\'e}n, Jouni and Monlong, Jean and Chang, Xian and Novak, Adam M and Eizenga, Jordan M and Markello, Charles and Sibbesen, Jonas A and Hickey, Glenn and Chang, Pi-Chuan and Carroll, Andrew and others},
  journal={Science},
  volume={374},
  number={6574},
  pages={abg8871},
  year={2021},
  publisher={American Association for the Advancement of Science}
}

@article{rautiainen2023telomere,
  title={Telomere-to-telomere assembly of diploid chromosomes with Verkko},
  author={Rautiainen, Mikko and Nurk, Sergey and Walenz, Brian P and Logsdon, Glennis A and Porubsky, David and Rhie, Arang and Eichler, Evan E and Phillippy, Adam M and Koren, Sergey},
  journal={Nature biotechnology},
  volume={41},
  number={10},
  pages={1474--1482},
  year={2023},
  publisher={Nature Publishing Group US New York}
}

@article{menger,
  title={Zur allgemeinen kurventheorie},
  author={Menger, Karl},
  journal={Fundamenta Mathematicae},
  volume={10},
  number={1},
  pages={96--115},
  year={1927},
  publisher={Polska Akademia Nauk. Instytut Matematyczny PAN}
}

@article{gutwenger2005inserting,
  title={Inserting an edge into a planar graph},
  author={Gutwenger, Carsten and Mutzel, Petra and Weiskircher, Ren{\'e}},
  journal={Algorithmica},
  volume={41},
  number={4},
  pages={289--308},
  year={2005},
  publisher={Springer}
}

@inproceedings{rahman2022uncovering,
  title={Uncovering hidden assembly artifacts: when unitigs are not safe and bidirected graphs are not helpful},
  author={Rahman, Amatur and Medvedev, Paul},
  booktitle={International Conference on Research in Computational Molecular Biology},
  pages={377--379},
  year={2022},
  organization={Springer}
}

@article{kita2017bidirected,
  title={Bidirected Graphs I: Signed General Kotzig-Lov$\backslash$'asz Decomposition},
  author={Kita, Nanao},
  journal={arXiv preprint arXiv:1709.07414},
  year={2017}
}

@article{ghorbani2025generalisation,
  title={A generalisation of Menger's theorem in bidirected graphs},
  author={Ghorbani, Ebrahim and Nickel, Jana Katharina and Reich, Florian},
  journal={arXiv preprint arXiv:2511.12283},
  year={2025}
}

@article{bessouf2019transitive,
  title={Transitive closure and transitive reduction in bidirected graphs},
  author={Bessouf, Ouahiba and Khelladi, Abdelkader and Zaslavsky, Thomas},
  journal={Czechoslovak Mathematical Journal},
  volume={69},
  number={2},
  pages={295--315},
  year={2019},
  publisher={Springer}
}

@article{ando1996decomposition,
  title={Decomposition of a bidirected graph into strongly connected components and its signed poset structure},
  author={Ando, Kazutoshi and Fujishige, Satoru and Nemoto, Toshio},
  journal={Discrete Applied Mathematics},
  volume={68},
  number={3},
  pages={237--248},
  year={1996},
  publisher={Elsevier}
}

@book{schrijver2003combinatorial,
  title={Combinatorial optimization: polyhedra and efficiency},
  author={Schrijver, Alexander and others},
  volume={24},
  number={2},
  year={2003},
  publisher={Springer}
}

@misc{zisis2026ultrabubbleenumerationlowestcommon,
      title={Ultrabubble enumeration via a lowest common ancestor approach}, 
      author={Athanasios E. Zisis and Pål Sætrom},
      year={2026},
      eprint={2603.03909},
      archivePrefix={arXiv},
      primaryClass={cs.DS},
      url={https://arxiv.org/abs/2603.03909}, 
}

@article{Fellows09,
  author       = {Michael R. Fellows and
                  Danny Hermelin and
                  Frances A. Rosamond and
                  St{\'{e}}phane Vialette},
  title        = {On the parameterized complexity of multiple-interval graph problems},
  journal      = {Theor. Comput. Sci.},
  volume       = {410},
  number       = {1},
  pages        = {53--61},
  year         = {2009},
  url          = {https://doi.org/10.1016/j.tcs.2008.09.065},
  doi          = {10.1016/J.TCS.2008.09.065},
  bibsource    = {dblp computer science bibliography, https://dblp.org}
}

@article{maniu2017indexing,
  title={An indexing framework for queries on probabilistic graphs},
  author={Maniu, Silviu and Cheng, Reynold and Senellart, Pierre},
  journal={ACM Transactions on Database Systems (TODS)},
  volume={42},
  number={2},
  pages={1--34},
  year={2017},
  publisher={ACM New York, NY, USA}
}

@article{DEMACEDOFILHO2018101,
title = {Using SPQR-trees to speed up recognition algorithms based on 2-cutsets},
journal = {Discrete Applied Mathematics},
volume = {245},
pages = {101-108},
year = {2018},
note = {LAGOS’15 — Eighth Latin-American Algorithms, Graphs, and Optimization Symposium, Fortaleza, Brazil — 2015},
issn = {0166-218X},
doi = {https://doi.org/10.1016/j.dam.2017.01.009},
url = {https://www.sciencedirect.com/science/article/pii/S0166218X17300422},
author = {H.B. {de Macedo Filho} and C.M.H. {de Figueiredo} and Z. Li and R.C.S. Machado}
}

@inproceedings{mutzel2003spqr,
  title={The SPQR-tree data structure in graph drawing},
  author={Mutzel, Petra},
  booktitle={International Colloquium on Automata, Languages, and Programming},
  pages={34--46},
  year={2003},
  organization={Springer}
}

@inproceedings{Kunnemann24,
  author       = {Marvin K{\"{u}}nnemann and
                  Mirza Redzic},
  editor       = {{\'{E}}douard Bonnet and
                  Pawel Rzazewski},
  title        = {Fine-Grained Complexity of Multiple Domination and Dominating Patterns
                  in Sparse Graphs},
  booktitle    = {19th International Symposium on Parameterized and Exact Computation,
                  {IPEC} 2024, Royal Holloway, University of London, Egham, United Kingdom,
                  September 4-6, 2024},
  series       = {LIPIcs},
  volume       = {321},
  pages        = {9:1--9:18},
  publisher    = {Schloss Dagstuhl - Leibniz-Zentrum f{\"{u}}r Informatik},
  year         = {2024},
  url          = {https://doi.org/10.4230/LIPIcs.IPEC.2024.9},
  doi          = {10.4230/LIPICS.IPEC.2024.9},
  bibsource    = {dblp computer science bibliography, https://dblp.org}
}

@article {Harviainen2026.03.28.714704,
	author = {Harviainen, Juha and Sena, Francisco and Moumard, Corentin and Politov, Aleksandr and Schmidt, Sebastian and Tomescu, Alexandru I},
	title = {Scalable computation of ultrabubbles in pangenomes by orienting bidirected graphs},
	elocation-id = {2026.03.28.714704},
	year = {2026},
	doi = {10.64898/2026.03.28.714704},
	publisher = {Cold Spring Harbor Laboratory},
	journal = {bioRxiv},
	URL = {http://biorxiv.org/content/early/2026/03/31/2026.03.28.714704/}
}
\addcontentsline{toc}{section}{References}

\end{document}